\documentclass[a4paper]{article}
\usepackage[utf8]{inputenc}
\usepackage[english]{babel} 
\usepackage{amsmath,amssymb,amsthm} 
\usepackage[margin=2cm]{geometry}
\usepackage{esint}
\usepackage[colorlinks]{hyperref}
\usepackage{todonotes}
\usepackage{enumerate}
\usepackage{framed}
\usepackage{upgreek}
\usepackage{pgfplots}
\usepackage{float}
\usepackage{tikz,tikz-cd}
\usepackage{ulem}

\newtheorem{theorem}{Theorem}[section]
\newtheorem{corollary}{Corollary}[theorem]
\newtheorem{definition}{Definition}[section]

\newtheorem{lemma}[theorem]{Lemma}

\newtheorem{remark}[theorem]{Remark}

\usepackage{parskip}
\numberwithin{equation}{section}

\def\p{\partial}

\def\cd{\cdot}

\def\div{\hbox{div}}
\def\curl{\hbox{curl}}

\def\zh{\mathbf{\hat{z}}}

\def\ep{{\epsilon}}

\def\bx{{\bf x}}

\def\bR{{\bf R}}
\def\bu{{\bf u}}

\def\bv{{\bf v}}

\pgfplotsset{compat=1.14}

\begin{document}
\title{Stochastic wave-current interaction \\ in thermal shallow water dynamics\\ \bigskip\Large
Darryl D. Holm and Erwin Luesink 
\\ \bigskip\small
Department of Mathematics, Imperial College London SW7 2AZ, UK}
\date{}                                           

\maketitle

\makeatother

\noindent
\section*{Abstract}
Holm (Proc. Roy. Soc 2015) introduced a variational framework for stochastically parametrising unresolved scales of hydrodynamic motion. This variational framework preserves fundamental features of fluid dynamics, such as Kelvin's circulation theorem, while also allowing for dispersive nonlinear wave propagation, both within a stratified fluid and at its free surface. The present paper combines asymptotic expansions and vertical averaging with the stochastic variational framework to formulate a new approach for developing stochastic  parametrisation schemes for nonlinear waves in fluid dynamics. The approach is applied to two sequences of shallow water models which descend from Euler's three-dimensional fluid equations with rotation and stratification under approximation by asymptotic expansions and vertical averaging. In the entire family of nonlinear stochastic wave-current interaction equations derived here using this approach, Kelvin's circulation theorem reveals a barotropic mechanism for wave generation of horizontal circulation or convection (cyclogenesis) which is activated whenever the gradients of wave elevation and/or topography are not aligned with the gradient of the vertically averaged buoyancy. 


\newpage

\tableofcontents

\newpage

\section*{Introduction}
Weather forecasting, climate change prediction and global ocean circulation all face the same fundamental challenge to create models which incorporate the effects of measurement error and uncertainty due to unresolved scales, unknown physical phenomena and incompleteness of observed data. We tackle this issue of modelling effects of unknown causes in observational science by applying new methods in stochastic data-driven modelling which are designed to predict both future measurements and their uncertainty, based on analysing the available data for the problem at hand.

For example, a common approach for modelling and simulating climate and weather is based on \textit{stochastic parametrisation}.
For recent reviews of stochastic parametrisation in geophysical fluid dynamics (GFD), see, e.g. \cite{berner2012systematic, berner2017stochastic, gottwald2016stochastic}.  The fundamental conclusions of \cite{berner2012systematic} are twofold:
\begin{quote}
\textit{A posteriori} addition of stochasticity to an already tuned model is simply not viable.

Stochasticity must be incorporated at a very basic level within the design of physical process parametrisations and improvements to the dynamical core.
\end{quote}
A new approach \cite{holm2015variational} which meets the challenge of incorporating stochastic parametrisation at the fundamental level enunciated in \cite{berner2012systematic} introduces stochastic transport directly into the loop velocity in Kelvin's circulation theorem. The dynamical quantities of physical interest are then modelled \textit{together with their statistical uncertainty}, and data assimilation is used to reduce that uncertainty. This is the SALT approach.

The SALT (Stochastic Advection by Lie Transport) approach combines stochasticity in the velocity of the fluid material loop in Kelvin's circulation theorem with ensemble forecasting. The ensemble forecasting in SALT has been coordinated with the results of the particle filtering method of data assimilation. A protocol for applying the SALT approach in combination with data assimilation based on comparing fine scale and coarse scale computational simulations has recently been established in \cite{cotter2019numerically,cotter2018modelling}. These results demonstrate the capability of the SALT approach to successfully reduce forecast uncertainty in a variety of test problems for fluid dynamics in two spatial dimensions. The three dimensional SALT theory has been developed and analysed to determine their existence, uniqueness and blow-up criterion in \cite{crisan2019solution}, but it still awaits computational implementation at the present time. 

The present paper aims to use the SALT approach for fluid dynamics described above to provide a barotropic (vertically averaged) description of wave-current interaction (WCI) in a stratified incompressible  fluid flow, by incorporating stochastic fluid transport and circulation with nonlinear dispersive wave propagation internally and on the free surface. In doing so, this paper combines a variational principle approach with asymptotic analysis to derive simplified models. Historically in ocean modelling, the rapid propagation of the barotropic (or, external) mode representing disturbances on the free surface, for example, has required special handling; because otherwise incorporating the simulation of its rapid time scale and multicomponent physical processes would tend to occupy an inordinate amount of computer power \cite{dukowicz1994implicit, fox2019challenges}.  

In addressing this challenge, the Camassa--Holm 1992 model (referred to as CH92 hereafter) derived in \cite{camassa1992dispersive}  used vertical averaging to transform the 3D Euler--Boussinesq fluid equations into a family of 2D stratified `rotating shallow water equations' which incorporate effects of weak deviations from hydrostatic balance, weak stratification and strong topography. Through a series of approximations and asymptotic limits, the CH92 model was found to contain the Kadomsev-Petviashvili (KP) and Korteweg-de Vries (KdV) equations in a rotating frame, as shown in Figure \ref{fig:CH1992}.  

The present paper will develop two families of stochastic models of barotropic wave-current interaction for mesoscale and submesoscale ocean dynamics based on the deterministic CH92 model and its further development in \cite{camassa1996long, camassa1997long}. Our approach combines dimensional analysis, asymptotic expansions and vertical averaging to obtain the barotropic component of the fluid motion, as done in \cite{camassa1992dispersive}, extended first to the Euler--Poincar\'e variational approach of \cite{holm1998euler} and then to the SALT approach \cite{holm2015variational} for introducing stochasticity. In the Euler-Poincar\'e version of the SALT approach, the approximation of the Lagrangian is separate from the introduction of stochasticity. The asymptotic expansions are applied to the Lagrangian first, and then the stochasticity is introduced in the variations of the Eulerian variables, which depend on spatially smooth maps with stochastic time dependence. In the variational step to include stochasticity, one also introduces the Strouhal number. In the process, we handle the barotropic effects by vertically averaging, applied either to the equations of motion as in \cite{wu1981long}, or to the variational principle for SALT  \cite{holm2015variational}. 
Of course, the vertical averaging procedure eliminates vertical buoyancy gradients. However, horizontal gradients of the vertically-averaged buoyancy remain. Here, the equations obtained after vertical averaging which retain horizontal gradients of buoyancy will be called \textit{thermal} equations. This name applies because the buoyancy plays the role of entropy per unit mass in the equation of state for adiabatic compressible fluid flows. Likewise, the variation of the energy with respect to the buoyancy plays the role of temperature in the adiabatic compressible fluid case.
Thus, the present paper aims to incorporate stochasticity into the theory of nonlinear dispersive water waves interacting with horizontal buoyancy gradients, as governed by vertically-averaged fluid equations. This stochastic theory of wave-current interaction in thermal shallow water dynamics is expected to be useful for quantifying uncertainty and perhaps even reducing it by using data assimilation in the SALT approach
\cite{holm2020stochastic}. 

%
%
\begin{figure}[H]
\begin{center}
\textbf{CH92 derivation road map}
\end{center}
\bigskip

\centering
\begin{tikzcd}
[row sep = 4em, 
column sep=1.5em, 
cells = {nodes={top color=blue!40, bottom color=white,draw=blue!90}},
arrows = {draw = black, rightarrow, line width = .03cm}]
& & {\begin{matrix}\text{3D Euler Boussinesq equations}\\ \text{ for stratified, rotating,}\\  \text{incompressible fluids}\end{matrix}} \arrow[d,"{\begin{matrix}
\text{Expand in small dimensionless}\\ \text{parameters and integrate vertically}\end{matrix}}", shorten <= 1mm, shorten >= 1mm] & & \\
& & {\begin{matrix} \text{2D Dispersive, stratified, rotating}\\  \text{shallow water equations} \end{matrix}} \arrow[dr, "{\begin{matrix} \text{Expand in stretched,}\\ \text{moving coordinates}\end{matrix}}", shorten <= 4mm, shorten >= 4mm, near end] \arrow[dl, "{\begin{matrix} \text{Neglect stratification and}\\ \text{impose hydrostatic balance}\end{matrix}}" swap, shorten <= 4mm, shorten >= 4mm, near end]& & \\
& {\begin{matrix}\text{Rotating shallow water}\\ \text{equations}\end{matrix}} \arrow[d, "{\begin{matrix} \text{Rigid lid, no}\\ \text{horizontal divergence}\end{matrix}}" swap, shorten <= 1mm, shorten >= 1mm] & & {\begin{matrix}\text{Kadomtsev--Petviashvili}\\ \text{equation with stratification}\\ \text{and rotation}\end{matrix}} \arrow[d, "{\begin{matrix}\text{No transverse}\\ \text{coordinate dependence}\end{matrix}}",  shorten <= 1mm, shorten >= 1mm] & \\
& {\begin{matrix}\text{Barotropic equations}\\ \text{for potential vorticity}\end{matrix}}& & {\begin{matrix}\text{Korteweg--De Vries equation}\\ \text{with stratification}\end{matrix}} &
\end{tikzcd}
\caption{The flow diagram of approximations via vertical averages and asymptotic expansions in \cite{camassa1992dispersive}.}
\label{fig:CH1992}
\end{figure}
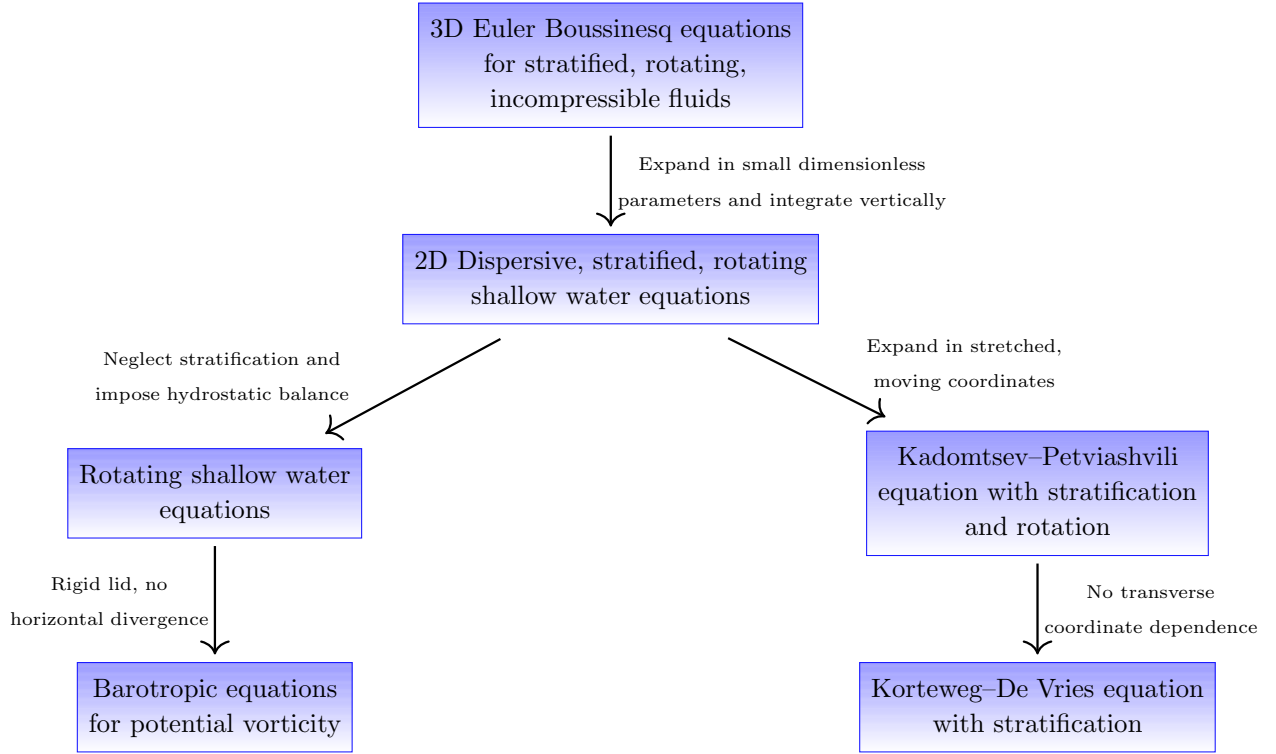

\paragraph{Background.} A framework for combining data with existing models in a probabilistic manner was presented in \cite{holm2015variational}, where a stochastic variational principle for continuum mechanics was introduced. This stochastic variational principle enables one to derive stochastic models of inviscid fluid dynamics which satisfy a Kelvin circulation theorem, starting from the Lagrangian of the corresponding deterministic fluid model and using a Clebsch constraint to introduce the Stochastic Advection by Lie Transport (SALT). This approach decomposes the fluid velocity vector field into the sum of a drift velocity and a Stratonovich stochastic velocity. The former is obtained from the constrained variational principle and the latter is determined by analysing available data according to the protocol established in \cite{cotter2019numerically,cotter2018modelling}. The constraints may be introduced either by imposing the advection equations for the relevant physical quantities of the model, or equivalently by imposing the advection equation for the fluid labels. 

Recently, in \cite{de2020implications}, the known Euler-Poincar\'e and Hamilton-Pontryagin stochastic variational principles were reformulated and shown to be equivalent to the Clebsch variant, by proving existence and uniqueness of the solution of the SALT advection constraint. The noise used in the \cite{holm2015variational} approach also appears in \cite{cotter2017stochastic}, where the decomposition  for SALT of the fluid velocity vector field into the sum of a drift velocity and a Stratonovich stochastic velocity was derived by using multi-time homogenisation theory. Many subsequent investigations of the properties of the equations of fluid dynamics with the SALT modification have appeared in the literature over the last four years. In particular, the SALT approach preserves most physical conservation laws by construction, while it also possesses much of the analytical structure of the underlying deterministic model. For example, in \cite{crisan2019solution}, the three dimensional SALT Euler equations are shown to have the same local-in-time existence and uniqueness analytical properties as the deterministic version, as well as the same Beale-Kato-Majda \cite{beale1984remarks} criterion for blow-up of solutions. In \cite{geurts2017lyapunov}, the Lorenz 63 equations are derived from Rayleigh-B\'enard convection with this type of stochasticity and the rate of convergence towards the attractor is shown to be preserved by this type of noise. From a more operational point of view, in \cite{cotter2019numerically}, SALT was introduced into the two dimensional Euler equations and it was shown that the stochastic equations, which are solved on a coarse grid, mimic the deterministic equations, which are solved on a fine grid, for a significant period of time. In \cite{cotter2018modelling}, a similar result was established for the flow in a channel of a two layer quasigeostrophic system. 
\bigskip 

In this paper, we are concerned with consistency of SALT under asymptotic expansion and analysis for the simulation of the barotropic mode in ocean dynamics. As mentioned earlier, the barotropic mode in ocean dynamics is the fastest excitation in the free-surface dynamics. It is treated separately (for example, by subcycling) in most 3D simulations of large scale ocean circulation. The issue of the free-surface treatment which motivated the original investigation of the various types of nonlinear wave behaviour in \cite{camassa1992dispersive} is still of current concern. 
\bigskip

A motivating question for introducing SALT into nonlinear dispersive water wave theory to be addressed in the present paper is: How can one use available data to quantify the uncertainty due to the barotropic mode in the free-surface treatment for computational simulations? This work is done in preparation for using the data assimilation methods of  \cite{cotter2019numerically,cotter2018modelling} to reduce that uncertainty, e.g., by using satellite data.  \bigskip

As in \cite{camassa1992dispersive}, we will combine asymptotic analysis with the vertical averaging principle of \cite{wu1981long} to derive a sequence of two dimensional barotropic models. This averaging principle will be applied both on the equations and also on the variational principle. The latter turns out to be advantageous in situations where the Strouhal number (the ratio of the chosen time scale over the natural time scale induced by the length and fluid velocity scales) is not equal to unity. The starting point of these derivations is the three dimensional rotating stratified Euler model, a three dimensional fluid model that includes the effects of rotation and buoyancy stratification. By making assumptions about the buoyancy stratification, we transition into the Euler-Boussinesq model. Here, we apply the averaging principle to derive two dimensional models with nonhydrostatic effects, rotation and stratification. The two dimensional models will be derived with respect to two different time scales: the first time scale is the natural one and the second is the time scale that corresponds to gravity waves. When the time scale is the natural one, the Strouhal number is equal to unity, which means that the asymptotic analysis applied to the equations and the asymptotic analysis applied inside the variational principle lead to the same result at each order in the asymptotic expansion. The assumption that the free surface amplitude is very small leads to the Great Lake, Lake and Benney long wave equations, first derived in \cite{camassa1996long, camassa1997long,benney1973some}, respectively, although in this paper we also include the effects of rotation, stratification and stochasticity. The second scaling regime is where the Strouhal number is equal to the inverse of the Froude number. This scaling regime leads to equations in the Green-Naghdi \cite{green1976derivation} class, if the free surface amplitude is assumed to be small, rather than very small. This derivation was first accomplished in \cite{camassa1992dispersive}, where also a Kadomtsev-Petviashvili equation is derived, augmented by the effects of rotation and bathymetry. In the presence of stochasticity, however, this derivation cannot be done directly. As we shall see, in the situation where the Strouhal number is not equal to unity, asymptotic analysis applied to the equations fails to respect the geometric structure of the problem, but the asymptotic analysis of the variational principle does preserve the geometric structure.
\bigskip

In \cite{holm2015variational}, the SALT  vector field whose characteristic curves generate the stochastic Lagrangian fluid trajectories is defined as
\begin{equation*}
{\sf d}\boldsymbol \chi_t := \mathbf{u}(\mathbf{x},t) dt + \sum_{i=1}^M \boldsymbol \xi_i(\mathbf{x})\circ dW_t^i.
\label{SALT}
\end{equation*} 
Here $\mathbf{u}(\mathbf{x},t)$ is the fluid velocity field, $\boldsymbol \xi_i(\mathbf{x})$ are the vector fields that represent spatial velocity-velocity correlations, $W_t^i$ denotes independent, identically distributed Wiener processes for each $i=1,\dots,M$, and the symbol $\circ$ means Stratonovich integration. The number $M$ of eigenvectors $\boldsymbol \xi_i(\mathbf{x})$ required for a given level of accuracy can be determined via the amount of variance required from a principal component analysis, or via empirical orthogonal function analysis. Via data assimilation procedures, in particular via novel high-dimensional particle filtering methods, the uncertainty may be controlled and reduced dramatically when even a small amount of new data is observed, as shown in \cite{cotter2019numerically,cotter2018modelling}. As we shall see, the  variational approach of SALT used here has the additional advantage of preserving the Kelvin circulation theorem and the Hamiltonian framework, both of which have been fundamental in the history of studying wave-current interaction and now can be made stochastic.
\bigskip

\noindent
\subsection*{Overview of the paper} 
The starting point, described in Section \ref{sec:StochVarPrinc}, will be the introduction of a number of tools which are invaluable for this work. First we will introduce the stochastic Euler-Poincar\'e variational principle, Kelvin's circulation theorem and an averaging principle. Then, starting with the rotating, stratified Euler equations, we will assume that the buoyancy stratification is weak enough to allow us to work with the Euler-Boussinesq equations. This is a justified assumption when the goal is to model the ocean. The flow in wave-current interaction is primarily incompressible, so the models used here will reflect this property. The ocean is shallow, compared to the horizontal distances of interest. In particular, the characteristic height scale is much smaller than the characteristic horizontal scales. This situation allows a reduction in spatial dimension by vertically integrating the Euler-Boussinesq equations to find the vertical average of the nonlinearity and an unknown vertically averaged pressure. Not surprisingly, these are the two terms which we cannot determine from the averaged equations alone. In order to derive a set of closed equations, we will turn to asymptotic analysis, which we will execute in two different regimes. Within each of those two regimes, we will apply asymptotic analysis in two different ways. In the first regime, called ``long time - very small wave scaling",  the time scale is determined by the ratio of the characteristic velocity scale and horizontal length scale, and with very small wave amplitude. The second regime, called the ``short time - small wave scaling", will employ the time scale based on the gravity wave speed and a characteristic horizontal length. The vertical averaging principle of \cite{wu1981long} will be applied, both on the 3D equations, and on the corresponding Euler--Poincar\'e Lagrangian. We will show that in the first regime, the approaches coincide and produce the same equations. In the second regime the asymptotic analysis requires special treatment, as the Strouhal number is not equal to unity in that situation. This difference in Strouhal number means that the material derivative contributes at two different orders in the asymptotic expansion. We shall focus on deriving  two dimensional stochastic fluid models in these two different time-scale regimes, starting from a model for a three dimensional stochastic fluid with rotation and stratification in a shallow box with bathymetry and a free surface.
\bigskip

In Section \ref{sec:LT-Vsmall}, ``the long time - very small wave scaling regime" of the Euler-Boussinesq equations with negligible buoyancy stratification will be derived from asymptotics applied to the equations and to the corresponding Lagrangian. At leading order, this will give rise to the Benney long wave equations, before making the columnar motion assumption. It will produce the stochastic and rotating version of the Lake equations after assuming that the motion is columnar. The Benney equations have an interesting mathematical structure, such as an infinity of conservation laws, as presented in \cite{kupershmidt2006extended}. From a different perspective, the rotating Lake equations are also obtained after assuming that the rotating shallow water equations have a rigid lid. At the next order we find the stochastic and rotating version of the Great Lake equations \cite{camassa1996long,camassa1997long}, which can be interpreted as the rigid lid version of the Green-Naghdi equations. The deterministic versions of the Lake and Great Lake equations are both globally wellposed in time, as shown in \cite{levermore1996global,levermore1996global2}.
\bigskip

In Section \ref{sec:STL-smallwavelimit}, ``the short time - small wave scaling" of the Euler-Boussinesq equations with non negligible buoyancy will be considered. The results of the asymptotics obtained in this regime are quite different from those obtained in the previous section. In this regime, the Strouhal number is not unity and the asymptotics on the equations provides us with a set of equations which are not closed. The reader may refer to \cite{camassa1992dispersive} for the deterministic derivation of these equations and their relation to the Kadomtsev--Petviashvili equation. The corresponding asymptotic analysis on the Lagrangian does give a closed set of equations, though, as it results in a buoyant version of the Green--Naghdi equations. As it turns out, a variational derivation of equations for the free surface alone is not available. Hence, the corresponding Boussinesq type water wave equations are not available, unless model assumptions in the variational principle were to be changed. It will be shown that a hierarchy of stochastic Camassa--Holm equations can be derived from this point of view leading to the stochastic Korteweg--De Vries equation, as well.

In Section \ref{sec:Conclude} we summarise by diagramming the pathways which relate the sequences of approximations leading to the results obtained in this paper for each of the two families of nonlinear fluid wave equations.  

\section{Stochastic variational principle and averaging principle}\label{sec:StochVarPrinc}
Central to this work is the stochastic Euler-Poincar\'e variational principle, presented in \cite{de2020implications}, which is equivalent to the variational principle in \cite{holm2015variational}. However, the Euler-Poincar\'e variational principle uses prescribed variations, rather than variations induced by constraints used in \cite{holm2015variational}. The most general version of the Euler-Poincar\'e theorem is formulated on the Lie algebra of a semidirect product Lie group and uses the language of differential geometry and representation theory, which first appeared deterministically in \cite{holm1998euler}. For fluids, the group of interest is the diffeomorphism group, which is the space of differentiable maps whose inverse maps are equally differentiable. The group action is composition of functions. The group of diffeomorphisms is a suitable group for geometric mechanics in the sense of \cite{ebin1970groups}. In order to state the Euler-Poincar\'e theorem, we first need to introduce some notation. 

\paragraph{Notation.}
The domain of interest for the paper is a three dimensional box with bathymetry specified by $h(x,y)$ and a free surface $\zeta(x,y,t)$, as illustrated in figure \ref{fig:domain} below.  The domain, which we will call $\Omega$, is a subset of $\mathbb{R}^3$ which is equipped with Cartesian coordinates. As we proceed, we will present the Euler-Poincar\'e theorem and its sequence of implications in $\mathbb{R}^3$ vector calculus, rather than using the more abstract differential geometric notation. 

Two dimensional and three dimensional objects will be distinguished by putting a subscript on the three dimensional objects, as follows
\begin{equation}
\mathbf{x}_3 = (\mathbf{x},z), \qquad \mathbf{u}_3 = (\mathbf{u},w), \qquad \nabla_3 = \left(\nabla, \frac{\partial}{\partial z}\right).
\end{equation} 
Here $\mathbf{x}_3$ denotes the coordinate system, $\mathbf{u}_3$ is the fluid vector field and $\nabla_3$ is the gradient. In oceanographic terms, we will work primarily with mesoscale dynamics, where the typical horizontal length scale $L$ is in the order of one hundred kilometres, or more, and the typical depth $H$ is four kilometres. Hence, the domain will be shallow. The rotation of the planet is included by introducing the vector potential $\mathbf{R}_3(\mathbf{x}_3)=(\mathbf{R}(\mathbf{x}),0)$ for the Coriolis parameter, $f(\mathbf{x})\hat{\mathbf{z}}$, so that
\begin{equation}
\nabla_3\times \mathbf{R}_3(\mathbf{x}) = f(\mathbf{x})\hat{\mathbf{z}}
\end{equation}
and $\hat{\mathbf{z}}$ is the unit vector in the vertical direction. 

\begin{figure}[H]
\centering
\includegraphics[width = .9\textwidth]{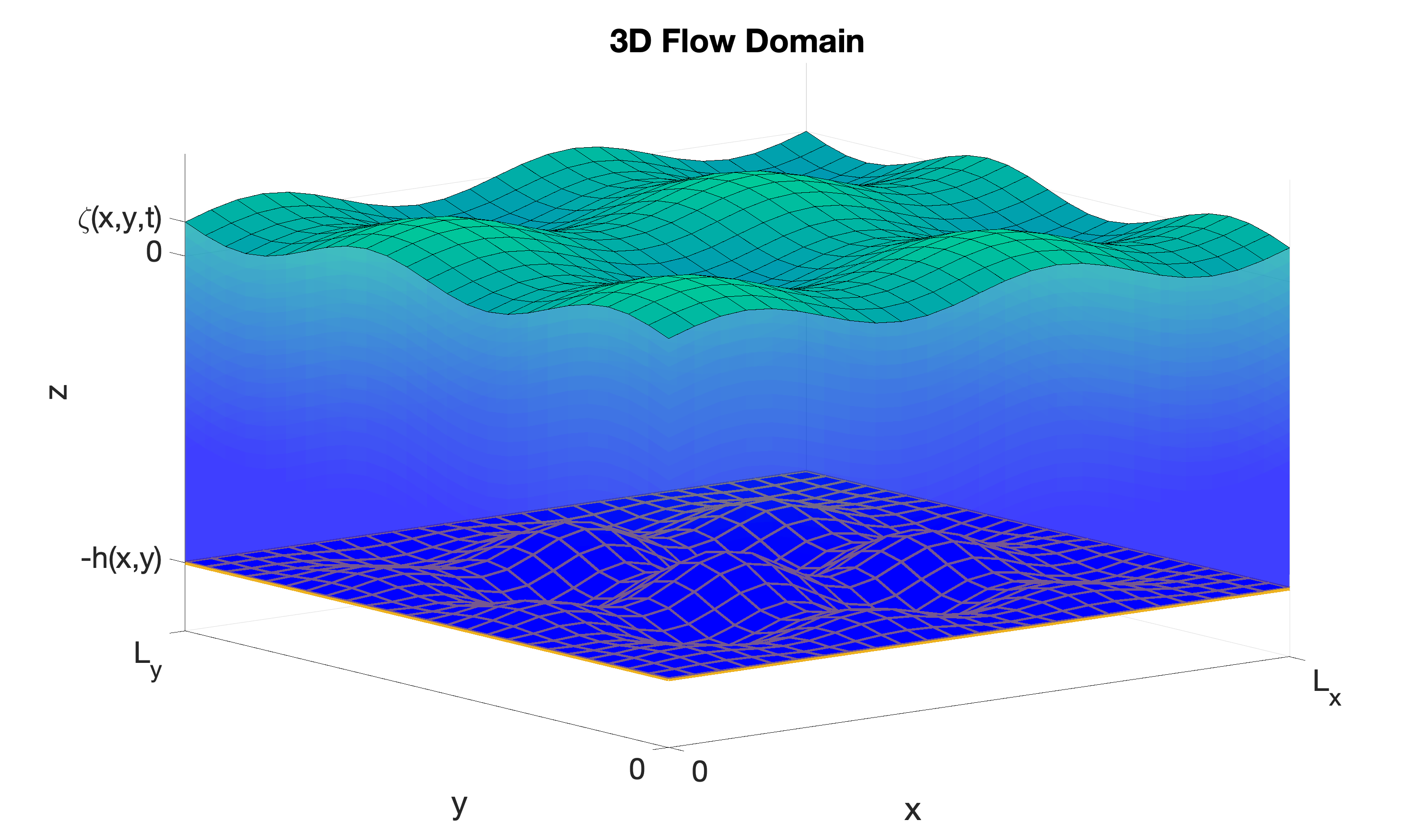}
\caption{The 3D flow domain, $\Omega$. The wavy green surface is the free surface $\zeta(x,y,t)$ and the wavy blue surface is the bathymetry $h(x,y)$. This figure is not to scale, as the horizontal length scale of an ocean domain is typically much larger than its height scale. In the paper we will assume that $L_x = L_y = L$.}
\label{fig:domain}
\end{figure} 

By $\mathfrak{X}(\Omega)$ we denote the space of vector fields over $\Omega$ and by $V^*$ we mean the abstract vector space of advected quantities, which are usually tensor fields of different degrees. In this paper, the elements of $V^*$ that we will consider are buoyancy $b$, which is a scalar function, the density $D$, which is a volume form and later in the two dimensional setting, we will consider the depth $\eta(\mathbf{x},t):= \zeta(\mathbf{x},t) + h(\mathbf{x})$, which is the volume form in that scenario. The stochastic vector fields which generate the Lagrangian transport in three and two dimensions are given, respectively, below. In this paper we will always work with the same stochastic basis, consisting of a set, a $\mathbb{P}$-complete $\sigma$-algebra, a probability measure $\mathbb{P}$ and a right-continuous filtration. All stochastic processes that we will consider are adapted to the filtration. The stochastic vector field for three dimensional transport is given by
\begin{equation}
{\sf d}\boldsymbol \chi_{3t} := \mathbf{u}_3(\mathbf{x}_3,t)dt + \sum_{i=1}^M \boldsymbol \xi_{3i}(\mathbf{x}_3)\circ dW_t^i\,.
\label{3Dnoise}
\end{equation}
The stochastic vector field $\boldsymbol \chi_{3t}$ in \eqref{3Dnoise} is an example of a semimartingale. 
\begin{definition}[Semimartingale]\label{def:smartingale}
A c\'adl\'ag process $Y$ is a semimartingale if it can be written as 
\[
Y_t = Y_0 + M_t + V_t
\,,\]
where $M$ is a c\'adl\'ag local martingale, $V$ is a c\'adl\'ag process of finite variation and $M_0 = V_0 = 0$.
\end{definition}
The adjective c\'adl\'ag stands for ``right continuous with a limit on the left''. The processes we will be considering will have continuous paths almost surely. Continuous semimartingales can be decomposed into a martingale part and a finite variation part uniquely. For more background on semimartingales and stochastic integration, see e.g. \cite{protter2005stochastic}. Semimartingales have several nice properties. In particular:
\begin{itemize}
\item For a suitably bounded predictable process $X$ and a semimartingale $Y$, the stochastic integral $\int X dY$ is again a semimartingale.
\item For a twice differentiable function $f$, the quantity $f(Y)$ is again a semimartingale.
\end{itemize} 
Note that the stochastic integral is of the Stratonovich type. This stochastic integral has the benefit that the stochastic calculus associated with it uses the usual chain rule and product rule, which are both required to derive the coordinate free stochastic Euler-Poincar\'e theorem. The Stratonovich integral relates to the It\^o integral via the following transformation. Let $X$ and $Y$ both be continuous semimartingales, then
\begin{equation}
\int_0^t X_s\circ dY_s = \int_0^t X_s dY_s + \frac{1}{2}\langle X,Y\rangle_t,
\end{equation}
where $\langle X,Y\rangle_t$ denotes the quadratic covariation process of $X$ and $Y$. In particular, the It\^o integral is an adapted process and also the quadratic covariation process is an adapted process. This means that also the Stratonovich integral is an adapted process. The two dimensional version of \eqref{3Dnoise} is given by
\begin{equation}
\label{eq:chi-defn}
{\sf d}\boldsymbol \chi_{t} := \mathbf{u}(\mathbf{x},t)dt + \sum_{i=1}^M \boldsymbol \xi_i(\mathbf{x})\circ dW_t^i.
\end{equation}
Hereafter, we will use Einstein's summation convention (summing repeated indices over their range) to keep the notation compact.
\paragraph{Boundary conditions.}
Having defined these vector fields for fluid transport, we can now specify the boundary conditions for the domain illustrated in figure \ref{fig:domain} above. One assumes that the free surface at the top is a Lagrangian surface, and that no fluid penetrates the bottom and vertical walls. Consequently, the following stochastic kinematic boundary conditions hold for the vertical velocity
\begin{equation}
wdt+\zh\cdot\boldsymbol \xi_{3i}\circ dW_t^i = {\sf d}\zeta + ({\sf d}\boldsymbol \chi_t\cdot\nabla)\zeta 
\quad \text{ at } z=\zeta(\mathbf{x},t), \quad\text{ and }\quad 
wdt+\zh\cdot\boldsymbol \xi_{3i}\circ dW_t^i = -({\sf d}\boldsymbol \chi_t \cdot \nabla)h \quad \text{ at } z = -h(\mathbf{x})\,.
\label{bc:kinematic}
\end{equation}
Here $\zh\cdot\boldsymbol \xi_{3i}$ selects the vertical component of the data vector fields. Since the stochastic flow does not penetrate the lateral boundaries, the horizontal velocity is taken to be tangential to the lateral boundaries
\begin{equation}
{\sf d}\boldsymbol \chi_{t}\cdot \mathbf{\hat{n}} = 0, \quad \text{ on any vertical lateral boundary,}
\label{bc:lateral}
\end{equation}
where $\mathbf{\hat{n}}$ is the unit vector normal to the lateral boundaries.
Finally, we assume the dynamic boundary condition for the pressure, namely,
\begin{equation}
p = 0 \qquad \text { at } \quad z = \zeta(\mathbf{x},t),
\label{bc:dynamic}
\end{equation}
or, alternatively, one can take $p=\zeta$ at $z=0$. This condition means that at the free surface the pressure is purely hydrostatic. In this formulation, surface tension has been neglected and the ambient pressure has been set to be zero at the surface. The lateral boundary condition is consistent with the incompressibility condition
\begin{equation}
\nabla_3\cdot{\sf d}\boldsymbol\chi_{3t} = 0.
\label{eq:incompressible}
\end{equation}
We want to be able to recover the deterministic fluid equations upon removing the stochastic terms in \eqref{3Dnoise} and \eqref{eq:chi-defn}, each of which is the sum of a deterministic vector and a stochastic vector. That is,  the stochastic fluid equations must return to the deterministic fluid equations when the noise term on the transport velocity is switched off. This type of consideration will be repeated as a `sanity check' throughout the paper. For example, this consideration requires that both terms in the transport vector field in \eqref{3Dnoise} must be divergence-free,
\begin{equation}
\nabla_3\cdot \mathbf{u}_3 = 0,  \text{ and } \nabla_3\cdot\boldsymbol \xi_i = 0, \hbox{ for all } i=1,\dots,M.
\label{eq:incompdiscussion}
\end{equation}
We will assume that the free surface and the pressure are both semimartingales. 

\subsection{Stochastic Euler-Poincar\'e theorem and averaging}

\paragraph{Variational derivatives of functionals.}
\begin{definition}[Functionals and functional derivatives]$\,$ \\
\noindent
A functional $F[\rho]$ is defined as a map  $F: \mathcal{B}\to \mathbb{R}$, where $\mathcal{B}$ is a Banach space. The variational derivative of a functional $F(\rho)$, denoted $\delta F/ \delta \rho$, is defined by
\begin{align}\label{eq:var-op-def}
\delta F[\rho]
:= \lim_{\varepsilon\to 0}\frac{F[\rho+\varepsilon \phi]-F[\rho]}{\varepsilon} 
=:  \frac{d}{d\varepsilon}F[\rho+\varepsilon \phi]\bigg|_{\epsilon=0}
=
 \int_\Omega \frac{\delta F}{\delta\rho}(x) \phi(x) \; dx 
 =: \left\langle \frac{\delta F}{\delta\rho}\,,\,\phi \right\rangle
\end{align}
where $\varepsilon\in \mathbb{R}$ is a real parameter, $\phi\in \mathcal{B}$  is arbitrary and the angle brackets $\langle\,\cdot\,,\,\cdot\,\rangle$ indicate $L^2$ pairing on $\mathcal{B}$. The derivative itself can be interpreted as a Fr\'echet derivative. The function $ \phi(x)$ above is called the `variation of $\rho$' and it will be denoted as $\delta \rho:= \phi(x) $. For notational convenience, we denote the functional derivative $\delta$ operationally as
\[
\delta F[\rho] = \left\langle \frac{\delta F}{\delta\rho}\,,\,\delta \rho \right\rangle.
\]
\end{definition}

\paragraph{Euler-Poincar\'e theorem.}
Given the boundary conditions and definitions above, the following form of the Euler-Poincar\'e theorem with stochastic variations provides the corresponding stochastic equations of motion derived from Hamilton's principle with a deterministic Lagrangian functional $\ell:\mathfrak{X}\times V^*\to \mathbb{R}$ defined on the domain of flow, $\Omega$. Here $\mathfrak{X}$ denotes the Lie algebra of smooth vector fields whose  action in three dimensional space by the Jacobi--Lie bracket is denoted as $[\,\cdot\,,\,\cdot]:\mathfrak{X}\times \mathfrak{X}\to \mathfrak{X}$, and is defined for $u,v\in \mathfrak{X}$ by the commutator relation, which in turn defines the \textit{minus} adjoint operator, ${\rm ad}$, given by
\begin{equation}\label{eq:ad-def}
\big[u,v\big] := \big( (\mathbf{u}_3\cdot\nabla_3)\mathbf{v}_3 
-
(\mathbf{v}_3\cdot\nabla_3)\mathbf{u}_3 \big)\cdot\nabla_3
=: - \,{\rm ad}_uv
\,.
\end{equation}
\begin{theorem}[Stochastic Euler-Poincar\'e equations \cite{holm2015variational,de2020implications}]
\label{thm:SEP}$\,$\\
The following two statements are equivalent:
\begin{enumerate}[i)]
\item Hamilton's variational principle in Eulerian coordinates, with $\mathbf{u}_3\in\mathfrak{X}(\Omega)$ and $b, D\in V^*(\Omega)$,\begin{equation}
\delta S := \delta\int_{t_1}^{t_2}\ell(\mathbf{u}_3,b,D) \,dt= 0,
\label{eq:hvp}
\end{equation}
holds on $\mathfrak{X}(\Omega)\times V^*$, upon using variations of the form
\begin{equation}
\delta \mathbf{u}_3 \,dt = {\sf d}\mathbf{v}_3 - [{\sf d}\boldsymbol \chi_{3t}, \mathbf{v}_3], \qquad \delta b \,dt = -(\mathbf{v}_3\cdot\nabla_3) b \,dt, \qquad \delta D\, dt= -\nabla_3\cdot(D\mathbf{v}_3)dt\,,
\label{eq:variations}
\end{equation}
where the arbitrary vector field $\mathbf{v}_3$ is a semimartingale.
\item The stochastic Euler-Poincar\'e equations hold. These equations are
\begin{equation}
{\sf d}\frac{\delta \ell}{\delta \mathbf{u}_3} 
+ ({\sf d}\boldsymbol \chi_{3t}\cdot\nabla_3)\frac{\delta\ell}{\delta \mathbf{u}_3} 
+ (\nabla_3{\sf d}\boldsymbol \chi_{3t})\cdot\frac{\delta\ell}{\delta \mathbf{u}_3} 
+ \frac{\delta\ell}{\delta \mathbf{u}_3} (\nabla_3\cdot{\sf d}\boldsymbol \chi_{3t})
= -\frac{\delta \ell}{\delta b}\nabla_3 b \,dt + D\nabla_3\frac{\delta\ell}{\delta D}dt
\label{eq:EPeq}
\end{equation}
or, equivalently,
\begin{equation}
{\sf d}\frac{\delta \ell}{\delta \mathbf{u}_3} 
- {\sf d}\boldsymbol \chi_{3t}\times \left(\nabla_3 \times  \frac{\delta\ell}{\delta \mathbf{u}_3} \right) 
+ \nabla_3 \left({\sf d}\boldsymbol \chi_{3t}\cdot\frac{\delta\ell}{\delta \mathbf{u}_3} \right)
+ \frac{\delta\ell}{\delta \mathbf{u}_3} (\nabla_3\cdot{\sf d}\boldsymbol \chi_{3t})
= -\frac{\delta \ell}{\delta b}\nabla_3 b\, dt + D\nabla_3\frac{\delta\ell}{\delta D}dt
\,,\label{eq:EPeq-Cartan}
\end{equation}
with advection equations 
\begin{align}
{\sf d}b = - \,{\sf d}\boldsymbol \chi_{3t} \cdot \nabla b
\quad\hbox{and}\quad
{\sf d}D = -\, \nabla_3\cdot(D{\sf d}\boldsymbol \chi_{3t})
\,.
\label{eq:EB-advectionLaws}
\end{align}
\end{enumerate}
\end{theorem}

\begin{remark}
The abstract statement of the Euler--Poincar\'e Theorem \ref{thm:SEP}, formulated on general semidirect product Lie groups, is presented in \cite{holm1998euler} deterministically and in \cite{holm2015variational,de2020implications} stochastically.
\end{remark}\medskip

\begin{remark}
In Theorem \ref{thm:SEP}, the operator $\delta$ in \eqref{eq:hvp} is the functional derivative defined in \eqref{eq:var-op-def}, the brackets $[\,\cdot\,,\,\cdot\,]$ denote the commutator of vector fields defined in \eqref{eq:ad-def}, and $\mathbf{v}_3 \in\mathfrak{X}(\Omega)$ is an arbitrary semimartingale vector field in three dimensions which vanishes at the endpoints in time, $t_1$ and $t_2$. Note that the stochasticity is introduced in the variation of the Eulerian velocity in \eqref{eq:variations}. This stochasticity in the variation is inherited from the stochasticity in the Lagrangian particle paths, as in \cite{holm2015variational}.
\end{remark}\medskip

\begin{remark}[Newton's Law interpretation of Euler-Poincar\'e equation \eqref{eq:EPeq}] $\,$\\
One may interpret the stochastic Euler-Poincar\'e equation \eqref {eq:EPeq} as the Newton's law of motion for a stochastic process. That is, the stochastic rate of change of the covector momentum $\mathbf{P}:=\delta\ell/\delta \mathbf{u}_3$ equals the sum of forces on the right hand side of equation \eqref{eq:EPeq}. Of course, when the stochasticity is removed from the vector field in \eqref{eq:chi-defn}, equation \eqref{eq:EPeq} recovers its deterministic version. 
\end{remark}

\begin{proof}
Hamilton's variational principle implies
\begin{align*}
0 &= \int_{t_1}^{t_2}\left[ \left\langle\frac{\delta\ell}{\delta\mathbf{u}_3},\delta\mathbf{u}_3 dt\right\rangle_\mathfrak{X} + \left\langle \frac{\delta\ell}{\delta b},\delta b dt\right\rangle_{V^*} + \left\langle\frac{\delta\ell}{\delta D},\delta D dt\right\rangle_{V^*}\right]
\\
&=  \int_{t_1}^{t_2}\left[ \left\langle\frac{\delta\ell}{\delta\mathbf{u}_3},{\sf d}\mathbf{v}_3 - [{\sf d}\boldsymbol \chi_{3t}, \mathbf{v}_3]\right\rangle_\mathfrak{X} + \left\langle \frac{\delta\ell}{\delta b},-(\mathbf{v}_3\cdot\nabla_3)b dt\right\rangle_{V^*} + \left\langle\frac{\delta\ell}{\delta D},-\nabla_3\cdot(D \mathbf{v}_3)dt\right\rangle_{V^*}\right]
\\
&= \int_{t_1}^{t_2}\left[ \left\langle -{\sf d}\frac{\delta\ell}{\delta \mathbf{u}_3}-({\sf d}\boldsymbol \chi_{3t}\cdot\nabla_3)\frac{\delta\ell}{\delta \mathbf{u}_3} - (\nabla_3{\sf d}\boldsymbol\chi_{3t})\cdot\frac{\delta\ell}{\delta \mathbf{u}_3} + \frac{\delta\ell}{\delta\mathbf{u}_3}(\nabla_3\cdot{\sf d}\boldsymbol \chi_{3t}),\mathbf{v}_3\right\rangle_\mathfrak{X} + \left\langle -\frac{\delta\ell}{\delta b}\nabla_3 b dt, \mathbf{v}_3\right\rangle_\mathfrak{X}\right.\\
&\quad \left.+ \left\langle D\nabla_3\frac{\delta\ell}{\delta D} dt, \mathbf{v}_3\right\rangle_\mathfrak{X} \right].
\end{align*}
The subscripts $\mathfrak{X}$ and $V^*$ on the $L^2$ pairings indicate over which space that the pairing is defined. Since the semimartingale $\mathbf{v}_3$ is arbitrary, except for vanishing at the endpoints $t_1$ and $t_2$ in time, the following equation holds,
\begin{align*}
{\sf d}\frac{\delta \ell}{\delta \mathbf{u}_3} + ({\sf d}\boldsymbol \chi_{3t}\cdot\nabla_3)\frac{\delta\ell}{\delta \mathbf{u}_3} + (\nabla_3{\sf d}\boldsymbol \chi_{3t})\cdot\frac{\delta\ell}{\delta \mathbf{u}_3} + \frac{\delta\ell}{\delta\mathbf{u}_3}(\nabla_3\cdot{\sf d}\boldsymbol \chi_{3t})= -\frac{\delta \ell}{\delta b}\nabla_3 b \,dt + D\nabla_3\frac{\delta\ell}{\delta D}\,dt.
\end{align*}
This finishes the proof of the stochastic Euler-Poincar\'e equation in \eqref{eq:EPeq}. The equivalent form in equation \eqref{eq:EPeq-Cartan} follows by means of a standard vector identity. 
\end{proof}

\subsection{Stochastic Kelvin--Noether circulation theorem}
A straight forward calculation combining equation \eqref{eq:EPeq} and the second advection equation in \eqref{eq:EB-advectionLaws} proves the following.
\begin{lemma}[Circulation form of the stochastic Euler-Poincar\'e equation \cite{holm2015variational,de2020implications}]\label{lemma:circulationform}
The stochastic Euler-Poincar\'e equation in \eqref{eq:EPeq} is equivalent to the following,
\begin{equation}
{\sf d}\left(\frac{1}{D}\frac{\delta \ell}{\delta \mathbf{u}_3}\right)
+ ({\sf d}\boldsymbol \chi_{3t}\cdot\nabla_3)\left(\frac{1}{D}\frac{\delta \ell}{\delta \mathbf{u}_3}\right)
+ (\nabla_3{\sf d}\boldsymbol \chi_{3t})\cdot\left(\frac{1}{D}\frac{\delta \ell}{\delta \mathbf{u}_3}\right)
= -\,\frac{1}{D}\frac{\delta \ell}{\delta b}\nabla_3 b \,dt + \nabla_3\frac{\delta\ell}{\delta D}dt
\,.\label{eq:EPeq-circ}
\end{equation}
\end{lemma}

One of the main benefits of Theorem \ref{thm:SEP} is that its stochastic Euler--Poincar\'e equations satisfy the following Kelvin circulation theorem.
\begin{theorem}[Stochastic Kelvin--Noether circulation theorem \cite{holm2015variational,de2020implications}] \label{thm:Kelvin}
For an arbitrary loop $c({\sf d}\boldsymbol \chi_{3t})$ which is advected by the stochastic velocity field ${\sf d}\boldsymbol \chi_{3t}$, the following circulation dynamics holds
\begin{equation}\label{eq:EPeq-circ-thm}
{\sf I} := \oint_{c({\sf d}\boldsymbol \chi_{3t})}\frac{1}{D}\frac{\delta\ell}{\delta \mathbf{u}_3}\cdot d\mathbf{x}_3 
\,,\qquad
{\sf d} {\sf I}
= -\oint_{c({\sf d}\boldsymbol \chi_{3t})} \left(\frac{1}{D}\frac{\delta\ell}{\delta b}\right)\nabla_3 b\cdot d\mathbf{x}_3 \,dt\,.
\end{equation}
\end{theorem}
\begin{proof}
The Kelvin circulation law \eqref{eq:EPeq-circ-thm} follows from Newton's law of motion obtained from the stochastic Euler-Poincar\'e equation \eqref {eq:EPeq-circ} for the evolution of momentum/mass $D^{-1}\delta\ell/\delta \mathbf{u}_3$ concentrated on an advecting material loop, $c({\sf d}\boldsymbol \chi_{3t}) = \phi_t c(0)$, where $\phi_t$ is the stochastic flow map generated by the stochastic vector field ${\sf d}\boldsymbol \chi_{3t}$ defined in equation \eqref{3Dnoise}. Upon changing variables to pull back the integrand to its initial position, the stochastic differential can be moved inside and the Kunita-It\^o-Wentzell formula may be applied \cite{de2020implications}. Then by inverting the pull-back we have the following
\begin{align*}
{\sf d}\oint_{c({\sf d}\boldsymbol \chi_{3t})}\frac{1}{D}\frac{\delta \ell}{\delta \mathbf{u}_3}\cdot d\mathbf{x}_3 
&= \oint_{c({\sf d}\boldsymbol \chi_{3t})} ({\sf d} + {\sf d}\boldsymbol\chi_{3t}\cdot\nabla_3 
+ (\nabla_3{\sf d}\boldsymbol\chi_{3t})\cdot)\left(\frac{1}{D}\frac{\delta \ell}{\delta \mathbf{u}_3}\right)\cdot d\mathbf{x}_3
\\
&= 
-\oint_{c({\sf d}\boldsymbol \chi_{3t})} \frac{1}{D}\frac{\delta \ell}{\delta b}\nabla_3 b \,dt\cdot d\mathbf{x}_3 
+ \oint_{c({\sf d}\boldsymbol \chi_{3t})}\nabla_3 \frac{\delta\ell}{\delta D} \cdot d\mathbf{x}_3\,dt
\\
&=  
-\oint_{c({\sf d}\boldsymbol \chi_{3t})} \left(\frac{1}{D}\frac{\delta \ell}{\delta b}\right)\nabla_3 b \cdot d\mathbf{x}_3 \,dt\,.
\end{align*}
In the second line we have used the Euler-Poincar\'e equation \eqref{eq:EPeq} and the advection equation for the density. The last step applies the fundamental theorem of calculus to show vanishing of the last loop integral in the second line. For the corresponding proof in the deterministic case, see \cite{holm1998euler}.
For detailed discussion of pull-back by stochastic flow maps, see \cite{de2020implications}.
\end{proof}

\begin{corollary}[Generation of circulation, ${\sf I}$]\label{SKN-Stokes}$\,$\\
By Stokes Law, equation  \eqref{eq:circLaw} in the stochastic Kelvin--Noether circulation theorem \ref{thm:Kelvin} implies 
\begin{equation}\label{eq:circLaw}
{\sf d} {\sf I} = 
-\int\!\!\int_{\partial S = c({\sf d}\boldsymbol \chi_{3t})} 
\nabla_3\left(\frac{1}{D}\frac{\delta\ell}{\delta b}\right)\times\nabla_3 b\cdot d\mathbf{S}_3 \,dt\,.
\end{equation}
Therefore, circulation is created by misalignment of the gradients of buoyancy $b$ and its dual quantity $D^{-1}\delta\ell/\delta b$.
\end{corollary}\medskip

\begin{remark}[A mechanism for cyclogenesis]
Formula \eqref{eq:circLaw} expresses the mechanism for generation of circulation (i.e., convection) driven by misalignment of certain potential gradients with gradients of scalar advected fluid quantities such as the buoyancy, $b$. In particular, formula \eqref{eq:circLaw} is the fundamental mechanism for generation of circulation or convection by wave-current interaction in stratified fluids. For the vertically averaged stratified fluid models treated later in the present paper, this formula will express a barotropic mechanism for generating horizontal circulation by misalignment of horizontal gradients of certain barotropic fluid quantities (such as wave elevation or bottom topography) with the horizontal gradient of vertically averaged buoyancy. 
\end{remark}

In three dimensional stochastic fluid dynamics, the Lagrangian in the Euler-Poincar\'e theorem is a functional defined over the volume of flow which, as we will discuss below, involves the kinetic energy density of the fluid relative to the rotating frame and the potential energy density. Our aims in the remainder of the paper are to combine asymptotic expansions and vertical averaging with the stochastic Euler-Poincar\'e variational theorem to formulate a new approach for developing stochastic parametrisation methods. To achieve these aims, we will apply asymptotic expansions in a vertically averaged  (barotropic) stochastic Euler-Poincar\'e variational principle. For this purpose, we will apply asymptotic expansions to the nondimensionalised Lagrangian for 3D incompressible flows of a stratified and rotating Euler fluid, then evaluate the vertical integral at an appropriate order in the expansion and finally use the Euler-Poincar\'e theorem to derive the equations of motion and advection we seek. We will then analyse and discuss their solution properties from the viewpoints of Newton's laws of motion and the Kelvin--Noether circulation theorem. We will also discuss the conservation laws for these equations.

\subsection{Nondimensionalising the Lagrangian}
The dimensional form of the Lagrangian in Hamilton's principle for the rotating, stratified Euler equations (rsE) is given by
\begin{equation}
\ell_{rsE}(\mathbf{u}_3,b,\zeta,D) := \int_\Omega \rho_0 (1+b)\left(\frac{1}{2}|\mathbf{u}|^2 + \frac{1}{2}w^2 + \mathbf{u}\cdot \mathbf{R} - gz\right)D\,dx\,dy\,dz.
\label{lag:rsE}
\end{equation}
Here, $\rho_0$ represents the reference density and $g$ represents gravity. The ocean has quite a few small dimensionless numbers which can be used to simplify the rsE Lagrangian and will allow one to access a hierarchy of simplified models. In particular, we want to derive the Lagrangian for the Euler-Boussinesq equations, which requires assumptions on the smallness of buoyancy, in terms of the Rossby number. To derive the equations of motion associated to the Lagrangian, we introduce the following action
\begin{equation}
S_{rsE} = \int_{t_1}^{t_2} \ell_{rsE} \,dt - \langle {\sf d}p, D-1\rangle =: \int_{t_1}^{t_2} c\ell_{rsE},
\label{action:rsE}
\end{equation} 
where ${\sf d}p$ is the Lagrange multiplier that enforces the density ratio $D$ to be equal to one, the times $0\leq t_1 \leq t_2$ are arbitrary, and the angle brackets refer to the $L^2$ pairing over the domain $\Omega$. The notation $c\ell_{rsE}$ refers to constrained Lagrangian and is introduced to keep the notation similar to the stochastic Euler-Poincar\'e theorem \ref{thm:SEP}. This constraint implies incompressibility and is required because it affects the measure $D\,dx\,dy\,dz$ in the Lagrangian. The treatment of the stochastic pressure is explained in the following remark.
\bigskip

\begin{remark}[Semimartingale pressure]
At this point one recognises a departure from the stochastic Euler-Poincar\'e equations without constraints derived in the Euler--Poincar\'e theorem \ref{thm:SEP}. Namely, we have written the Lagrange multiplier ${\sf d}p$ which imposes the constraint $D-1=0$. The notation stresses that ${\sf d}p$ is imposing a constraint that is stochastic. Now, setting $D=1$ in the advection equation for $D$ by the stochastic vector field ${\sf d}\boldsymbol \chi_{3t}$ implies that $\nabla_3\cdot( {\sf d}\boldsymbol \chi_{3t})=0$. Following the discussion leading to \eqref{eq:incompdiscussion}, this in turn must also imply $\nabla_3\cdot\mathbf{u}_3=0$. By its definition in \eqref{eq:chi-defn}, the quantity $\boldsymbol \chi_{3t}$ is a semimartingale. Therefore, accounting for both the deterministic and stochastic parts of the motion equation in \eqref{eq:SEB} will require that the pressure ${\sf d}p$ must also be a semimartingale, hence the notation. The point is that the semimartingale $D$ cannot be enforced to be a constant by a deterministic Lagrange multiplier. The Lagrange multiplier must also be obtained from a semimartingale equation. In the present case, this can be accomplished by acknowledging that the pressure is a semimartingale and writing its contribution in the motion equation as ${\sf d}p$, in a notation which implies a sum of both Lebesque and stochastic time integrations. Then, upon imposing the consequence of $D=1$ in the form $\nabla_3\cdot \mathbf{u}_3=0$ we find a semimartingale Poisson equation for ${\sf d}p$ which encompasses both the deterministic and stochastic parts of the constrained motion equation. Finally, the time integration of the solution of the Poisson equation for ${\sf d}p$ determines the semimartingale $p$. For a treatment of general semimartingale driven variational principles, see \cite{street2020semi}.
\end{remark}
\bigskip

The nondimensional versions of all the relevant variables and parameters are given below,
\begin{framed}
\begin{equation}
\begin{matrix}
\mathbf{x}_3 = L(\mathbf{x}', \sigma z'), & \mathbf{u}_3 = U(\mathbf{u}', \sigma w'), &\nabla_3 = \dfrac{1}{L}\left(\nabla', \dfrac{1}{\sigma}\dfrac{\partial}{\partial z'}\right), & t = Tt', & W_t = \dfrac{1}{\sqrt{T}}W_{t'},
\\[10pt] 
h = Hh', & \zeta = \alpha H\zeta',& \mathbf{R} = f_0 L \mathbf{R}', &  \rho = \rho_0 \rho', & {\sf d}p = \rho_0 gH {\sf d}p',
\\[10pt]
\sigma = \dfrac{H}{L}, & \alpha = \dfrac{\zeta_0}{H}, & {\rm Fr} = \dfrac{U}{\sqrt{gH}}, & {\rm Ro} = \dfrac{U}{f_0 L}, & {\rm Sr} = \dfrac{L}{UT}.
\end{matrix}
\label{tab:scaling}
\end{equation}
\end{framed}
Here $L$ denotes the horizontal scale, $H$ is the vertical scale, $U$ is the typical horizontal velocity, $f_0$ is the rotation frequency, $\zeta_0$ is the typical free surface amplitude and $T$ is the time scale. The dimensionless numbers in the bottom row are, respectively, the aspect ratio $\sigma$, the wave amplitude $\alpha$, the Froude number ${\rm Fr}$, the Rossby number ${\rm Ro}$ and the Strouhal number ${\rm Sr}$. Note that we have also scaled the Brownian motion so that in the nondimensional setting, the noise is again a standard Brownian motion. The dimensional factor that arises can be absorbed into the $\boldsymbol \xi_{3i}$ for each $i$. The vertical component of the data vector fields $\boldsymbol \xi_{3i}$ is scaled with the aspect ratio as well, that is $\boldsymbol \xi_{3i} = (\boldsymbol \xi_i',\sigma\zh\cdot\boldsymbol \xi_{3i}')$. In particular, this means that we can write
\begin{equation}
{\sf d}\boldsymbol \chi_{3t} = U({\sf d}\boldsymbol \chi_t', \sigma \zh\cdot{\sf d}\boldsymbol \chi_{3t}').
\end{equation} 
We do not make any assumptions on the size of the data vector fields relative to the velocity field itself. Last, but not least, the stratification parameter $\mathfrak{s}$ is introduced. Since the buoyancy is already dimensionless, it does not appear in the table, but it works as follows
\begin{equation}
b = \mathfrak{s}b'.
\end{equation} 
The purpose of the stratification parameter is to make sure that the buoyancy variable $b$ is an order $\mathcal{O}(1)$ variable. By controlling the size of the stratification parameter $\mathfrak{s}$, the Boussinesq approximation can be introduced. The nondimensional rsE Lagrangian is obtained by substituting \eqref{tab:scaling} into \eqref{lag:rsE} and dropping the primes, which yields
\begin{equation}
\ell_{rsE}(\mathbf{u}_3,b,D) = \int_{\Omega} (1+\mathfrak{s}b)\left(\frac{1}{2}|\mathbf{u}|^2 + \frac{\sigma^2}{2}w^2 + \frac{1}{{\rm Ro}}\mathbf{u}\cdot\mathbf{R}- \frac{1}{{\rm Fr}^2}z\right) D\,dx\,dy\,dz,
\label{DimensionlessEuler}
\end{equation} 
and the dimensionless action is given by
\begin{equation}
S_{rsE} = \int_{t_1}^{t_2}\ell_{rsE}\,dt - \left\langle\frac{1}{{\rm Fr}^2}{\sf d}p,D-1\right\rangle =: \int_{t_1}^{t_2}c\ell_{rsE}.
\end{equation}
In the ocean, the horizontal scale $L$ is of the order of hundreds of kilometres, whereas the vertical scale $H$ is typically about four kilometres. The free surface amplitude is five metres and the horizontal velocity is about a tenth of a metre per second. Hence the aspect ratio $\sigma \ll 1$, the wave amplitude $\alpha \ll 1$ and the Froude number ${\rm Fr}\ll1$. The Rossby number at these scales is also small, ${\rm Ro}\ll 1$. Also, the buoyancy stratification is weak, which allows us to apply the Boussinesq approximation. This approximation corresponds to $\mathfrak{s}\ll 1$, that is, requiring the stratification parameter to be small. When the dimensionless numbers satisfy $\mathcal{O}(\alpha) = \mathcal{O}(\mathfrak{s}) = \mathcal{O}({\rm Ro}) = \mathcal{O}({\rm Fr}) = \mathcal{O}(\sigma^2)$, the Lagrangian can be approximated. Consequently, the rsE Lagrangian simplifies, as the remaining effect of buoyancy is restricted to the potential energy term. This yields the Euler-Boussinesq (EB) Lagrangian, given by
\begin{equation}
\ell_{EB}(\mathbf{u}_3,b,D) = \int_{\Omega} \left(\frac{1}{2}|\mathbf{u}|^2 + \frac{\sigma^2}{2}w^2 + \frac{1}{{\rm Ro}}\mathbf{u}\cdot\mathbf{R} - \frac{1}{{\rm Fr}^2}(1+\mathfrak{s}b)z\right)D\,dx\,dy\,dz.
\label{lag:EB}
\end{equation}
The Euler-Boussinesq equations are obtained by applying the Euler-Poincar\'e theorem to the action obtained by taking the Lagrangian in \eqref{lag:EB} with the pressure constraint, as in \eqref{action:rsE}. The action for the EB equations is then given by 
\begin{equation}
S_{EB} = \int_{t_1}^{t_2}\ell_{EB} \, dt - \left\langle \frac{1}{{\rm Fr}^2}{\sf d}p, D-1\right\rangle =: \int_{t_1}^{t_2} c\ell_{EB}.
\label{action:EB}
\end{equation} 
Besides assuming the buoyancy is small, we will assume that the variations of the Coriolis parameter and of the bathymetry profile are also small, of order $O({\rm Ro})$,
\begin{equation}
f(\mathbf{x}) = 1 + {\rm Ro}\, f_1(\mathbf{x}), \qquad h(\mathbf{x}) = 1 + {\rm Ro}\, h_1(\mathbf{x}).
\label{as:bathymetry}
\end{equation}
These assumptions are made because they are consistent with the assumptions for quasi-geostrophy. The Lagrangian of interest in \eqref{lag:EB} is in dimensionless form, but the constraints in theorem \ref{thm:SEP} are still dimensional. Since $\mathbf{v}_3$ is arbitrary, multiplying it by some constant does not change its arbitrary nature. Hence, besides the $\delta\mathbf{u}_3$ constraint, nothing changes upon nondimensionalisation. As said earlier, the $\delta\mathbf{u}_3$ variational constraint does change, as follows,
\begin{equation}
\delta\mathbf{u}_3 dt = {\rm Sr}\,{\sf d}\mathbf{v}_3 - [{\sf d}\boldsymbol \chi_{3t},\mathbf{v}_3].
\label{eq:nondimconstraint}
\end{equation}
Time does not appear explicitly anywhere in the rsE and EB Lagrangians. Thus, the Strouhal number has not appeared before; but time rescaling has a significant impact on the behaviour of the model. In \eqref{eq:nondimconstraint} one can see that if the Strouhal number is not unity, advection will no longer be balanced. This observation will be crucial later, when we look at the short time limit. So far, we have obtained a theorem which, for a certain deterministic Lagrangian for three-dimensional fluids, provides us with the corresponding stochastic equations. By explicitly evaluating the vertical integral, when possible, in that theorem, we have a systematic way to obtain the vertically averaged version of the three dimensional fluid equations of interest. We also have introduced a general nondimensionalisation and identified the scales in the problem which determine the small dimensionless numbers in the ocean. Now, an application of theorem \ref{thm:SEP} to the EB Lagrangian \eqref{action:EB}, with variations given by
\begin{equation}
\begin{aligned}
\frac{\delta c\ell_{EB}}{\delta\mathbf{u}} &= D\left(\mathbf{u} + \frac{1}{{\rm Ro}}\textbf{R}\right),\\
\frac{\delta c\ell_{EB}}{\delta w} &= \sigma^2 D w,\\
\frac{\delta c\ell_{EB}}{\delta D} &= \frac{1}{2}|\mathbf{u}|^2 + \frac{\sigma^2}{2}w^2 + \frac{1}{{\rm Ro}}\mathbf{u}\cdot\mathbf{R} - \frac{1}{{\rm Fr}^2}(1+\mathfrak{s}b)z - \frac{1}{{\rm Fr}^2}{\sf d}p,\\
\frac{\delta c\ell_{EB}}{\delta b} &= -\frac{\mathfrak{s}}{{\rm Fr}^2}Dz,\\
\frac{\delta c\ell_{EB}}{\delta {\sf d}p} &= \frac{1}{{\rm Fr}^2} (D-1).
\end{aligned}
\end{equation} 
implies the following stochastic Euler-Poincar\'e equations in circulation form (see lemma \ref{lemma:circulationform})
\begin{equation}
\begin{aligned}
{\rm Sr}\,{\sf d}\mathbf{u} + ({\sf d}\boldsymbol \chi_{3t}\cdot\nabla_3)\mathbf{u} + (\nabla\boldsymbol\xi_{3i})\cdot\mathbf{u}_3\circ dW_t^i &= -\frac{1}{{\rm Fr}^2} \nabla {\sf d}p  - \frac{1}{{\rm Ro}}f\hat{\mathbf{z}}\times{\sf d}\boldsymbol \chi_t - \frac{1}{{\rm Ro}}\nabla(\boldsymbol \xi_i\cdot \mathbf{R})\circ dW_t^i,\\
\sigma^2\left({\rm Sr}\,{\sf d}w + ({\sf d}\boldsymbol \chi_{3t}\cdot\nabla_3)w+ \Big(\frac{\partial}{\partial z}\boldsymbol \xi_{3i}\Big)\cdot\bu_3\circ dW_t^i\right) &= -\frac{1}{{\rm Fr}^2}\frac{\partial}{\partial z}{\sf d}p + \frac{1}{{\rm Fr}^2}(1+\mathfrak{s}b)\,dt,\\
{\rm Sr}\,{\sf d}b + ({\sf d}\boldsymbol \chi_{3t}\cdot\nabla_3)b &= 0,\\
\nabla_3\cdot( {\sf d}\boldsymbol \chi_{3t}) & = 0.
\end{aligned}
\label{eq:SEB}
\end{equation}
The Euler-Boussinesq equations satisfy the following Kelvin circulation theorem, for any closed loop $c({\sf d}\boldsymbol \chi_{3t})$ which is advected with the stochastic velocity ${\sf d}\boldsymbol \chi_{3t}$ in equation \eqref{3Dnoise}, 
\begin{equation}
\begin{aligned}
{\rm Sr}\,{\sf d}\oint_{c({\sf d}\boldsymbol \chi_{3t})}\left((\mathbf{u},\sigma^2 w) + \frac{1}{{\rm Ro}}(\mathbf{R},0)\right)\cdot d\mathbf{x}_3 &= -\frac{\mathfrak{s}}{{\rm Fr}^2}\oint_{c({\sf d}\boldsymbol \chi_{3t})} z\nabla_3 b\cdot d\mathbf{x}_3 dt\\
&= -\frac{\mathfrak{s}}{{\rm Fr}^2}\int\!\!\int_{\partial S=c({\sf d}\boldsymbol \chi_{3t})} \hat{\mathbf{z}}\times\nabla_3 b\cdot d\mathbf{S} dt,
\end{aligned}
\label{thm:KCEB}
\end{equation}
where the notation $(\mathbf{u}, \sigma^2 w)$ denotes a three dimensional vector field, two horizontal components from $\mathbf{u}$ and the vertical component $\sigma^2 w$. As $\mathbf{R}$ is strictly horizontal, the vertical component is zero. Hence the misalignment of the unit vector in the vertical direction and the gradient of buoyancy creates vertical circulation, or convection. 

Additionally, the Euler-Boussinesq equations satisfy the Silberstein-Ertel theorem for potential vorticity. This theorem states that the potential vorticity, defined by
\begin{equation}
q := \mathfrak{s}\nabla_3 b \cdot \nabla_3\times\left((\textbf{u},\sigma^2 w) + \frac{1}{{\rm Ro}}(\mathbf{R},0)\right),
\end{equation}
is conserved along particle trajectories and thus satisfies the following equation
\begin{equation}
{\rm Sr}\,{\sf d}q + ({\sf d}\boldsymbol \chi_{3t}\cdot\nabla_3)q = 0.
\label{eq:PVEB}
\end{equation}
Since the buoyancy and the potential vorticity are constant along particle trajectories, the spatially integrated quantity,
\begin{equation}
C_\Phi = \int_{\Omega} \Phi(b,q) \,dx\,dy\,dz,
\label{eq:CasimirsEB}
\end{equation}
is also preserved in time for any differentiable function, $\Phi$, for which the integral exists. The proof is analogous to the deterministic case, which is shown in \cite{holm1998euler,holm1999euler}. A special case of this statement is the preservation of the enstrophy, which is defined as the $L^2$ norm of the potenial vorticity. Since the flow is divergence free, one can also define the enstrophy in terms of the gradients of the velocity. This shows that the Euler-Boussinesq equations, even in the presence of SALT, have an infinite number of conservation laws. This structure must also be preserved by the vertical averaging. The spatially integrated quantities $C_\Phi$ are also referred to as Casimirs, as they are the functions whose Lie--Poisson bracket corresponding to the Euler-Boussinesq equations vanishes for any Hamiltonian expressed in the Eulerian fluid variables. 

\subsection{Averaging of Newton's second law}
Besides evaluating the vertical integral in the variational principle, one can also choose to use Newton's second law to derive the equations of fluid motion in this domain, rather than using the Euler-Poincar\'e theorem. By means of the method of control volumes, it is possible to derive the equations and also come up with an averaging principle. This is what is shown in \cite{wu1981long} for the deterministic case. The stochastic case is not that different, but there is one issue that requires careful treatment: there is an additional advection term. Let us denote the vertical average by putting a bar over the relevant quantity
\begin{equation}
\overline{f} := \frac{1}{\eta}\int_{-h}^{\alpha\zeta} f dz.
\label{eq:average}
\end{equation} 
The stochastic vector field in the averaged situation is denoted
\begin{equation}
\overline{{\sf d}\boldsymbol \chi}_t = \overline{\bu}\,dt + \overline{\boldsymbol \xi}_i\circ dW_t^i.
\end{equation}
For incompressible flows, the advection equation for a scalar and the continuity equation for a density can be written in the same form. That is, the average of a scalar function $f(\mathbf{x}_3,t)$ and that of a volume form $f(\mathbf{x}_3,t) d^3x$, for incompressible flows,
are of the same form,
\begin{equation}
{\rm Sr}\, {\sf d}\int_{-h}^\zeta f(\mathbf{x}_3,t) dz + \nabla\cdot\int_{-h}^{\alpha\zeta} f(\mathbf{x}_3,t){\sf d}\boldsymbol \chi_t dz = 0.
\label{eq:verticalintegral}
\end{equation}
In the deterministic case, it is possible to substitute in the fluid velocity for $f$ in \eqref{eq:verticalintegral} and obtain the vertically averaged momentum equation after applying \eqref{eq:average}. The formula above holds for scalars and densities, but fluid velocity is neither. However, the fluid velocity equation obtained in this way is correct, but only in the deterministic case. The explanation for this coincidence is the following. In the deterministic setting, the advective terms in the equation for the fluid velocity for incompressible fluids are $(\mathbf{u}\cdot\nabla)\mathbf{u} + (\nabla\mathbf{u})\cdot \mathbf{u}$. The latter term is equal to the gradient of the kinetic energy, so a cancellation occurs in Newton's second law. When SALT is introduced in this problem, the kinetic energy is the same as in the deterministic situation, but the advective terms are now stochastic, hence this cancellation no longer occurs. 

Applying \eqref{eq:average} and \eqref{eq:verticalintegral} to the Euler-Boussinesq equations \eqref{eq:SEB} yields the following vertically averaged nonlinear equations, 
\begin{equation}
\begin{aligned}
{\rm Sr}\,{\sf d} (\eta\overline{\mathbf{u}}) 
+ \nabla\cdot(\eta\overline{{\sf d}\boldsymbol \chi_{t}\otimes\mathbf{u}}) + \eta(\nabla\overline{\boldsymbol\xi}_i)\cdot\overline{\mathbf{u}} \circ dW_t^i &= -\frac{1}{{\rm Fr}^2}\eta \overline{\nabla{\sf d} p} - \frac{1}{{\rm Ro}}\eta f\hat{\mathbf{z}}\times\overline{{\sf d}\boldsymbol \chi_t} - \frac{1}{{\rm Ro}}\eta \nabla(\overline{\boldsymbol \xi}_i\cdot \mathbf{R})\circ dW_t^i,\\
{\rm Sr}\,{\sf d}(\eta\overline{b}) + \nabla\cdot (\eta\overline{b {\sf d}\boldsymbol\chi_t}) &= 0,\\
{\rm Sr}\,{\sf d}\eta + \nabla\cdot(\eta \overline{{\sf d}\boldsymbol \chi_t}) &= 0.
\end{aligned}
\label{eq:ASEB}
\end{equation}
The last equation is obtained by substituting unity into \eqref{eq:verticalintegral}. It corresponds to conservation of volume in the two dimensional setting. As the problem is incompressible, the vertical velocity can be expressed in terms of the horizontal velocity field and the vertical component of the data vector fields $\boldsymbol \xi_{3i}$ can be expressed in terms of the horizontal components as
\begin{equation}
\begin{aligned}
w(\mathbf{x},z) &= - \nabla\cdot \int_{-h}^z \mathbf{u}(\mathbf{x},z') dz' = \nabla\cdot\int_z^{\alpha\zeta} \mathbf{u}(\mathbf{x},z')dz',\\
\zh\cdot\boldsymbol \xi_{3i}(\mathbf{x},z) &= - \nabla\cdot \int_{-h}^z \boldsymbol \xi_i(\mathbf{x},z')dz' = \nabla\cdot\int_{z}^{\alpha\zeta}\boldsymbol \xi_i(\mathbf{x},z')dz'.
\end{aligned}
\label{eq:wdivu}
\end{equation}
This expression has been derived by vertically integrating the three dimensional incompressibility condition \eqref{eq:incompressible}, using the uniqueness of the semimartingale decomposition and using the boundary conditions on the vertical velocity to pull the divergence outside of the integral. Importantly, the boundary conditions introduce a dependence between the horizontal components of the vector fields and the vertical component. A horizontal two dimensional model with bathymetry and a free surface will therefore give some information about what is happening in the vertical direction. This holds also for the $\boldsymbol \xi_i$. Even though the Newtonian averaging approach is very insightful, there is a drawback. Namely, the averaged equations \eqref{eq:ASEB} are not closed. Indeed, they contain three terms which are unknown. In the momentum equation, the average of the nonlinear term and the average of the pressure are unknown. In the buoyancy equation, the advection term is unknown. In order to close this set of equation, we will use asymptotic analysis, which we shall employ in two different scaling regimes.
These are the long time - very small wave (LT-VSW) scaling regime in section \ref{sec:LT-Vsmall} and the short time - small wave (ST-SW) scaling regime in section \ref{sec:STL-smallwavelimit}.
Here, `long time scale' is $T = L/U$, the time it takes for a fluid parcel to cross the horizontal length scale; and `short time scale' is $T = L/\sqrt{gH}$, the time it takes for a gravity wave to cross the horizontal  length scale. Likewise, `small wave' means that the wave amplitude is small, but not small enough to consider taking the rigid lid limit; while `very small wave' means that the wave amplitude is the small parameter of interest.

\section{Long time - very small wave scaling regime}\label{sec:LT-Vsmall}
Long time corresponds to choosing the time scale to be $T = L/U$ and very small wave means that the amplitude of the wave $\alpha$ is the small parameter of interest. In this setting we therefore have the following dimension-free parameters,
\begin{framed}
\begin{equation}
\begin{matrix}
\mathbf{x}_3 = H\left(\dfrac{1}{\sigma}\mathbf{x}', z'\right), & \mathbf{u}_3 = U(\mathbf{u}', \sigma w'), &\nabla_3 = \dfrac{1}{L}\left(\nabla', \dfrac{1}{\sigma}\dfrac{\partial}{\partial z'}\right), & t = \dfrac{L}{U} t', & W_t = \sqrt{\dfrac{L}{U}}W_{t'},
\\[10pt] 
h = Hh', & \zeta = \alpha H\zeta',& \mathbf{R} = f_0 L \mathbf{R}', &  \rho = \rho_0 \rho', & {\sf d}p = \rho_0 gH {\sf d}p',
\\[10pt]
\sigma = \dfrac{H}{L}, & \alpha = \dfrac{\zeta_0}{H}, & {\rm Fr} = \dfrac{U}{\sqrt{gH}}, & {\rm Ro} = \dfrac{U}{f_0 L}, & {\rm Sr} = 1.
\end{matrix}
\label{tab:ltscaling}
\end{equation}
\end{framed}
In particular, the velocity field and the data vector fields are scaled in the same way, hence we have
\begin{equation}
{\sf d}\boldsymbol \chi_{3t} = U({\sf d}\boldsymbol \chi_t', \sigma \zh\cdot{\sf d}\boldsymbol \chi_{3t}').
\end{equation}
With these scaling relations and the stratification parameter $\mathfrak{s}$, the constrained EB Lagrangian in equation \eqref{action:EB} takes the following form
\begin{equation}
\begin{aligned}
S_{EB}(\mathbf{u}_3,b,D) &= \int_{t_1}^{t_2}\int_{\Omega} \left(\frac{1}{2}|\mathbf{u}|^2 + \frac{\sigma^2}{2}w^2 + \frac{1}{{\rm Ro}}(\mathbf{u}\cdot\mathbf{R}) - \frac{1}{{\rm Fr}^2}(1+\mathfrak{s}b)z\right) D\,dx\,dy\,dz\,dt - \langle {\sf d}p, D-1\rangle
,\\
&=: \int_{t_1}^{t_2} c\ell_{EB}
\end{aligned}
\label{lag:ltEB}
\end{equation}
Note that no information about the very small free surface amplitude appears in the Lagrangian; it only contains the aspect ratio, which controls the size of the vertical kinetic energy. However, information about the size of the free surface amplitude does appear in the boundary conditions, which are
\begin{equation}
\begin{matrix}
p = \zeta \qquad & \text{ at } z = 0,\\[4pt]
wdt + \zh\cdot\boldsymbol \xi_{3i}\circ dW_t^i= \alpha\big({\sf d}\zeta + ({\sf d}\boldsymbol \chi_t\cdot\nabla)\zeta\big) \qquad & \text{ at } z = \alpha\zeta(\mathbf{x},t),\\[4pt]
wdt + \zh\cdot\boldsymbol \xi_{3i}\circ dW_t^i= -({\sf d}\boldsymbol \chi_t\cdot\nabla)h \qquad & \text{ at } z = -h(\mathbf{x}),\\[4pt]
{\sf d}\boldsymbol \chi_t\cdot \mathbf{n} = 0 \qquad & \text{ on lateral boundaries}.
\end{matrix}
\label{bc:ltEB}
\end{equation}
An application of the stochastic Euler-Poincar\'e Theorem \ref{thm:SEP} on the long time scale Lagrangian in \eqref{lag:ltEB} now yields the following equations
\begin{equation}
\begin{aligned}
{\sf d}\mathbf{u} + ({\sf d}\boldsymbol \chi_{3t}\cdot\nabla_3)\mathbf{u} + (\nabla\boldsymbol\xi_{3i})\cdot \mathbf{u}_3\circ dW_t^i &= -\frac{1}{{\rm Fr}^2}\nabla {\sf d}p - \frac{1}{{\rm Ro}}f\hat{\mathbf{z}}\times {\sf d}\boldsymbol \chi_t - \frac{1}{{\rm Ro}}\nabla(\boldsymbol \xi_i\cdot \mathbf{R})\circ dW_t^i,\\
\sigma^2\left({\sf d}w + ({\sf d}\boldsymbol \chi_{3t}\cdot\nabla_3)w + \Big(\frac{\partial}{\partial z}\boldsymbol \xi_{3i}\Big)\cdot\bu_3 \circ dW_t^i\right) &= - \frac{1}{{\rm Fr}^2}\frac{\partial}{\partial z}{\sf d}p - \frac{1}{{\rm Fr}^2}(1+\mathfrak{s}b) dt,\\
\nabla\cdot {\sf d}\boldsymbol \chi_{3t} &= 0.
\end{aligned}
\label{eq:ltEB}
\end{equation}
The equations in \eqref{eq:ltEB} satisfy the Kelvin circulation theorem as in \eqref{thm:KCEB}, and have conservation of potential vorticity along particle trajectories as in \eqref{eq:PVEB}. These equations also conserve an infinity of integral quantities as in \eqref{eq:CasimirsEB}. In the long-time scaling in \eqref{tab:ltscaling} the Strouhal number is equal to one. In this scaling regime, the equations take a particularly nice form.  The dimensionless numbers of interest are the aspect ratio $\sigma$ and the wave amplitude $\alpha$, the Rossby number ${\rm Ro}$ shall be left untouched. In particular, we consider $\alpha\ll\sigma\ll 1$, where we let the wave amplitude tend to zero while holding the aspect ratio fixed. 

\paragraph{Rigid lid approximation.}
The effect of sending the wave amplitude $\alpha$ to zero is the \textit{rigid lid approximation}, where the free surface is no longer allowed to vary and becomes a rigid boundary, instead. This removes gravity waves from the problem. However, the leading order dynamics can still be recovered from the dynamic boundary condition on the pressure. The effect of sending $\alpha\to 0$ before touching the aspect ratio is that one can derive equations that include the nonhydrostatic effect due to the vertical velocity. The corresponding equations are the so-called Great Lake equations, first derived in \cite{camassa1996long,camassa1997long}. Taking $\sigma\to 0$ after the rigid lid limit leads to the Lake equations. If one takes $\sigma\ll\alpha\ll 1$, the result is the same, but the route is slightly different. Upon sending $\sigma\to 0$, the vertical component in the Lagrangian \eqref{lag:ltEB} vanishes and upon assuming columnar motion, one can integrate the Lagrangian vertically. This leads to the Lagrangian for rotating shallow water. Sending $\alpha\to 0$ corresponds to putting a rigid lid on top of the rotating shallow water equations and this leads to the Lake equations. Upon taking $\alpha \to 0$ while keeping $\sigma$ fixed, the equations \eqref{eq:ltEB} do not change, but the boundary conditions in \eqref{bc:ltEB} do:
\begin{equation}
\begin{matrix}
p = \zeta \qquad & \text{ at } z = 0,\\[4pt]
w = 0 \qquad & \text{ at } z = 0,\\[4pt]
wdt + \zh\cdot\boldsymbol \xi_{3i}\circ dW_t^i = -({\sf d}\boldsymbol \chi_t\cdot\nabla)h \qquad & \text{ at } z = -h(\mathbf{x}),\\[4pt]
{\sf d}\boldsymbol \chi_t\cdot \mathbf{n} = 0 \qquad & \text{ on lateral boundaries}.
\end{matrix}
\label{bc:ltEB1}
\end{equation}
In the limit $\alpha \to 0$, the depth $\eta = h$, as the contribution of the free surface vanishes. Also, the expressions for the vertical velocity and the vertical component of the data vector field simplify, as the free surface contribution vanishes, and take the form
\begin{equation}
\begin{aligned}
w &= \nabla\cdot \int_z^0 \bu\, dz',\\
\zh\cdot\boldsymbol \xi_{3i} &= \nabla\cdot \int_z^0 \boldsymbol \xi_i\, dz'.
\end{aligned}
\end{equation}
Averaging with the Newtonian approach leads to the following vertically averaged versions of the equations \eqref{eq:ltEB}, 
\begin{equation}
\begin{aligned}
{\sf d} \overline{\mathbf{u}} + \frac{1}{h}\nabla\cdot(h\overline{{\sf d}\boldsymbol \chi_{t}\otimes \mathbf{u}}) + (\nabla\overline{\boldsymbol\xi}_{i})\cdot\overline{\mathbf{u}}\circ dW_t^i &= -\frac{1}{{\rm Fr}^2}\overline{\nabla {\sf d}p} - \frac{1}{{\rm Ro}} f \hat{\mathbf{z}}\times \overline{{\sf d}\boldsymbol \chi}_t - \frac{1}{{\rm Ro}}\nabla(\overline{\boldsymbol \xi}_i\cdot \mathbf{R})\circ dW_t^i,\\
{\sf d}\overline{b} + \nabla\cdot(\overline{b {\sf d}\boldsymbol \chi}_t)&= 0,
\label{eq:ltEBav}
\end{aligned}
\end{equation}
and
\begin{equation}
\nabla \cdot (h\overline{{\sf d}\boldsymbol \chi}_t) = 0.
\label{eq:weightedincomp}
\end{equation}
The continuity equation has become a weighted incompressibility condition \eqref{eq:weightedincomp}, where the weight is determined by the bathymetry profile. As in the discussion above about the incompressibility condition \eqref{eq:incompdiscussion}, the weighted incompressibility must hold for the velocity field and the $\boldsymbol \xi_i$ independently. If the bathymetry is flat, one finds the two-dimensional incompressibility condition. However, the momentum equation and the buoyancy equation above still suffer from the problem that terms are present which we, as yet, have not determined.

\subsection{Leading order expansion in the long time -- very small wave scaling  regime}\label{sec:rotatinglake}

As an initial approach, let us assume a leading order expansion  in $\sigma^2$. Even though the Rossby number is small as well, we will consider a single scale expansion in $\sigma^2$ for the variables:
\begin{equation}
\begin{matrix}
\mathbf{u} = \mathbf{u}_0 + o(1), & w = w_0 + o(1), & \boldsymbol \xi_{3i} = \boldsymbol \xi_{0,3i} + o(1),\\
{\sf d}\boldsymbol \chi_{3t} = {\sf d}\boldsymbol \chi_{0,3t} + o(1), & {\sf d}p = {\sf d}p_0 + o(1), & \zeta = \zeta_0 + o(1), \\ b = 0 + o(1). & &
\end{matrix}
\label{exp:leadingorder}
\end{equation}
The buoyancy does not contribute in the leading order expansion, since the stratification parameter is required to satisfy $\mathfrak{s} \ll 1$ for the Boussinesq approximation. Also note that the data vector fields $\boldsymbol \xi_{3i}$ are expanded in the same way as the velocity itself. No assumptions are made about the size of the data vector fields. Upon substituting \eqref{exp:leadingorder} into \eqref{eq:ltEB}, the vertical velocity equation at leading order implies hydrostatic balance
\begin{equation}
\frac{\partial}{\partial z}{\sf d}p_0 + 1\,dt= 0,
\label{eq:hydrostaticbalance}
\end{equation}
and the dynamic boundary condition \eqref{bc:dynamic} implies that the leading order pressure is equal to the leading order free surface elevation.
\bigskip

\begin{remark}\label{semimartingaledecomposition}
Note that there is no stochasticity entering \eqref{eq:hydrostaticbalance} explicitly. Due to the assumption of the pressure being a semimartingale, the pressure has the standard semimartingale decomposition. When there is no stochasticity in the equation, the martingale part of the pressure must vanish and we have the expression ${\sf d}p_0 = p_0 dt$ with a slight abuse of notation.
\end{remark} 

Interestingly, the substitution of the leading order expansion leads to a closed model even before averaging, when one uses the expression \eqref{eq:wdivu} for the vertical velocity as an additional equation. Given the boundary conditions in \eqref{bc:ltEB1}, the leading order expansion leads to a set of equations reminiscent of the Benney long wave model \cite{benney1973some}. There are a few twists, though, since stochasticity and rotation are also involved. Moreover, the wave amplitude $\alpha$ is very small, which enforces the rigid lid approximation in the vertical integral. At leading order, there cannot be any confusion as to which order of the expansion we are considering. Consequently, we may drop the subscript $o$ in writing the following set of equations,
\begin{equation}
\begin{aligned}
{\sf d}\mathbf{u} + ({\sf d}\boldsymbol \chi_{3t}\cdot\nabla_3)\mathbf{u} + (\nabla\boldsymbol\xi_{i})\cdot\mathbf{u}\circ dW_t^i &= -\frac{1}{{\rm Fr}^2}\nabla {\sf d}p - \frac{1}{{\rm Ro}}f\hat{\mathbf{z}}\times{\sf d}\boldsymbol \chi_t - \frac{1}{{\rm Ro}}\nabla(\boldsymbol \xi_i \cdot \mathbf{R})\circ dW_t^i,\\
w &= \nabla\cdot\int_z^0 \mathbf{u} \,dz',\\
\zh\cdot\boldsymbol \xi_{3i} &= \nabla\cdot\int_z^0 \boldsymbol \xi_i\, dz'.
\end{aligned}
\label{eq:benneylike}
\end{equation}
Together with the weighted incompressibility condition in \eqref{eq:weightedincomp}, the dynamic boundary condition on the pressure \eqref{bc:dynamic} and the lateral boundary condition \eqref{bc:lateral}, the Benney-like equations \eqref{eq:benneylike} form a closed set. The Benney long wave equations are interesting because they have a very rich mathematical structure, including an infinite hierarchy of conservation laws, as shown in \cite{kupershmidt2006extended}. If we now make the additional assumption that the leading order component of the horizontal velocity field and the leading order component of the horizontal data vector field are independent of the vertical coordinate; that is, if we assume that the leading order component is columnar, then a considerable simplification of \eqref{eq:benneylike} occurs. Namely, the derivative in the vertical direction drops out. Consequently, it is no longer necessary to determine the vertical velocity and now every term in the equation is horizontal. This set of equations we will refer to as the stochastic, rotating, Lake equations, given by
\begin{equation}
{\sf d}\mathbf{u} + ({\sf d}\boldsymbol \chi_t\cdot\nabla)\mathbf{u}+ (\nabla\boldsymbol\xi_i)\cdot\mathbf{u}\circ dW_t^i = -\frac{1}{{\rm Fr}^2}\nabla {\sf d}\zeta- \frac{1}{{\rm Ro}}f\hat{\mathbf{z}}\times{\sf d}\boldsymbol \chi_t - \frac{1}{{\rm Ro}}\nabla(\boldsymbol \xi_i \cdot \mathbf{R})\circ dW_t^i,
\label{eq:lake}
\end{equation}
accompanied by the weighted incompressibility condition in \eqref{eq:weightedincomp} and the lateral boundary condition \eqref{bc:lateral}. The dynamic boundary condition can now be used to determine the pressure at the free surface. The deterministic, irrotational version of these equations has been shown by \cite{levermore1996global,levermore1996global2,
levermore1997analyticity} to be globally wellposed. These equations satisfy a Kelvin circulation theorem, namely
\begin{equation}
{\sf d}\oint_{c({{\sf d}\boldsymbol \chi_t})}\left(\mathbf{u}+\frac{1}{{\rm Ro}}\mathbf{R}\right)\cdot d\mathbf{x} = 0,
\end{equation}
where $c({{\sf d}\boldsymbol \chi_t})$ is any fluid loop that is advected by the stochastic vector field ${{\sf d}\boldsymbol \chi_t}$. This means that circulation is conserved, as there are no terms on the right hand side to generate circulation. Hence the enstrophy in this model is conserved as well. The proof of the Kelvin circulation theorem is either a direct computation, or a corollary of the Euler--Poincar\'e theorem. We will derive these equations from a variational point of view as well, which will prove the Kelvin circulation theorem above.

\subsection{Higher order expansion in the long time -- very small wave scaling regime}\label{sec:rotatinggreatlake}
Let us now consider a higher order perturbation expansion:
\begin{equation}
\begin{matrix}
\mathbf{u} = \mathbf{u}_0 + \sigma^2\mathbf{u}_1 + o(\sigma^2), & w = w_0 +\sigma^2 w_1 + o(\sigma^2),& \boldsymbol \xi_{3i} = \boldsymbol \xi_{0,3i} + \sigma^2 \boldsymbol \xi_{1,3i}+ o(\sigma^2),\\
{\sf d}\boldsymbol \chi_{3t} = {\sf d}\boldsymbol \chi_{0,3t} + \sigma^2 {\sf d}\boldsymbol \chi_{1,3t}+ o(\sigma^2), & {\sf d}p = {\sf d}p_0 + \sigma^2 {\sf d}p_1 + o(\sigma^2), & \zeta = \zeta_0 + \sigma^2 \zeta_1+ o(\sigma^2),\\
\qquad \mathfrak{s}b = \sigma^2 b + o(\sigma^2). & &
\end{matrix}
\label{exp:higherorder}
\end{equation}
It is natural to assume that the leading order terms satisfy the Lake equations \eqref{eq:lake}. Note that the stratification parameter is assumed to satisfy $\mathcal{O}(\mathfrak{s})=\mathcal{O}(\sigma^2)$. This will allow us to consider the buoyancy independently from the higher order terms that will appear in the equations to come. Hence, at leading order in the vertical velocity equation, we have hydrostatic balance \eqref{eq:hydrostaticbalance} and in the horizontal component we have columnar motion. At the next order, we substitute \eqref{eq:wdivu} for the vertical velocity and obtain
\begin{equation}
\frac{\partial}{\partial z}{\sf d}p_1 + b_1\,dt = z\big({\sf d}\nabla\cdot \mathbf{u}_0 + ({\sf d}\boldsymbol \chi_{0,t}\cdot\nabla)(\nabla\cdot \mathbf{u}_0) - (\nabla\cdot{\sf d}\boldsymbol \chi_{0,t})(\nabla\cdot\bu_0)\big).
\end{equation}
On the right hand side, everything in the brackets is independent of the vertical coordinate, so integration is particularly simple and leads to
\begin{equation}
{\sf d}p_1 = {\sf d}\zeta_1 - b_1 z dt + \frac{1}{2}z^2 \big({\sf d}\nabla\cdot \mathbf{u}_0 + ({\sf d}\boldsymbol \chi_{0,t}\cdot\nabla)(\nabla\cdot \mathbf{u}_0) - (\nabla\cdot{\sf d}\boldsymbol \chi_{0,t})(\nabla\cdot\bu_0)\big).
\end{equation}
This shows that the pressure deviates from hydrostatic balance at order $\sigma^2$, as the pressure is a function of free surface elevation, buoyancy and horizontal velocity. The vertical average of the horizontal gradient of the pressure above is
\begin{equation}
\begin{aligned}
\overline{{\sf d}\nabla p_1} &= \nabla {\sf d}\zeta_1 + \frac{1}{2}h\nabla b_1 dt + \frac{1}{6}h^2\big({\sf d}\nabla\nabla\cdot \mathbf{u}_0 + ({\sf d}\boldsymbol\chi_{0,t} \cdot\nabla)(\nabla\nabla\cdot \mathbf{u}_0) + (\nabla{\sf d}\boldsymbol\chi_{0,t})\cdot(\nabla\nabla\cdot \mathbf{u}_0)\\
&\quad - (\nabla\nabla\cdot {\sf d}\boldsymbol \chi_{0,t})(\nabla\cdot \mathbf{u}_0)- (\nabla\cdot{\sf d}\boldsymbol \chi_{0,t})(\nabla\nabla\cdot\bu_0)\big).
\end{aligned}
\end{equation}
By using the weighted incompressibility condition \eqref{eq:weightedincomp}, the expression above can be simplified and combined into
\begin{equation}
\overline{{\sf d}\nabla p_1} = \nabla {\sf d}\zeta_1 + \frac{1}{2}h\overline{\nabla b_1} dt + \bigg({\sf d} + ({\sf d}\boldsymbol \chi_{0,t}\cdot\nabla) + (\nabla{\sf d}\boldsymbol \chi_{0,t})\cdot \bigg)\left(\frac{1}{6}h^2(\nabla\nabla\cdot \mathbf{u}_0)\right).
\label{eq:avpress}
\end{equation}
The following observation allows us to deal with the average of the nonlinear term. Namely, if the leading order terms satisfy the stochastic, rotating Lake equations \eqref{eq:lake}, then the leading order component of the stochastic velocity field is independent of the vertical coordinate. The higher order component of the stochastic vector field is not independent of the vertical coordinate, though, so its average is not trivial. Hence the average of the full stochastic velocity field is
\begin{equation}
\overline{{\sf d}\boldsymbol \chi}_t = {\sf d}\boldsymbol \chi_{0,t} + \sigma^2 \overline{{\sf d}\boldsymbol \chi}_{1,t} + o(\sigma^2).
\label{eq:averageidentity}
\end{equation}
From this expression, it is clear that the average of the product minus the product of the average is a higher order term:
\begin{equation}
\overline{{\sf d}\boldsymbol \chi_t\otimes \mathbf{u}} - \overline{{\sf d}\boldsymbol \chi}_t\otimes \overline{\mathbf{u}} = \mathcal O(\sigma^4).
\end{equation}
Therefore, by adding and subtracting the product of the average in \eqref{eq:ltEBav}, we can write a closed system of equations. For notational convenience, we define
\begin{equation}
\overline{\mathbf{V}}(\mathbf{x},t) := \overline{\mathbf{u}}(\mathbf{x},t) + \frac{\sigma^2}{6}h^2\nabla(\nabla\cdot \overline{\mathbf{u}}) + o(\sigma^2), 
\label{def:v}
\end{equation}
and use our expression for the average of the pressure \eqref{eq:avpress} into \eqref{eq:ltEBav} to write
\begin{equation}
\begin{aligned}
{\sf d}\overline{\mathbf{V}} + (\overline{{\sf d}\boldsymbol \chi_t}\cdot\nabla)\overline{\mathbf{V}} + (\nabla\overline{{\sf d}\boldsymbol \chi_t})\cdot \overline{\mathbf{V}} &= -\frac{1}{{\rm Fr}^2}\nabla{\sf d}\zeta + \frac{1}{2}|\overline{\mathbf{u}}|^2dt - \frac{\mathfrak{s}}{2\,{\rm Fr}^2}h\nabla \overline{b} dt - \frac{1}{{\rm Ro}}f\hat{\mathbf{z}}\times\overline{{\sf d}\boldsymbol\chi}_t - \frac{1}{{\rm Ro}}\nabla(\overline{\boldsymbol \xi}_i \cdot \mathbf{R})\circ dW_t^i,\\
{\sf d}\overline{b} + (\overline{{\sf d}\boldsymbol \chi}_t \cdot\nabla)\overline{b}&= 0.
\end{aligned}
\label{eq:greatlake}
\end{equation}
The stratification parameter $\mathfrak{s}$ is assumed to be of the same order as $\sigma^2$, the aspect ratio squared. When the buoyancy becomes negligible, the stratification parameter tends to zero. This removes the buoyancy from equation \eqref{eq:greatlake}, but the nonhydrostatic pressure terms stay. Taking the shallow water limit by letting the aspect ratio tend to zero also removes the buoyancy contribution, since the buoyancy is linked to the aspect ratio in the expansions introduced in \eqref{exp:higherorder}. Together with the weighted incompressibility condition \eqref{eq:weightedincomp} and lateral boundary condition \eqref{bc:lateral}, the set of equations \eqref{eq:greatlake} comprises the stochastic, rotating, thermal Great Lake equations. The deterministic, non-rotating version of these equations is presented in \cite{camassa1996long,camassa1997long}, together with the elliptic operator that relates $\overline{\mathbf{V}}$ and $\overline{\mathbf{u}}$. To solve for the pressure ${\sf d}\zeta$, one uses the elliptic operator just mentioned, which is defined by
\begin{equation}
\begin{aligned}
h\overline{\mathbf{V}} &= h\overline{\mathbf{u}} + \left[-\frac{\sigma^2}{3}\nabla(h^3\nabla\cdot\overline{\mathbf{u}}) - \frac{\sigma^2}{2}\nabla(h^2\overline{\mathbf{u}}\cdot\nabla h) + \frac{\sigma^2}{2}h^2(\nabla\cdot\overline{\mathbf{u}})\nabla h + \sigma^2 h(\overline{\mathbf{u}}\cdot\nabla h)\nabla h\right],\\
&=: \mathfrak L(h)\overline{\mathbf{u}}.
\end{aligned}
\label{def:ellipticoperator}
\end{equation} 
This operator is positive-definite and self-adjoint since $h>0$. An application of the Lax-Milgram theorem guarantees the continuous dependence of $\overline{\mathbf{u}}$ on $\overline{\mathbf{V}}$ \cite{levermore1996global,levermore1996global2}. By operating with $\nabla\cdot h\mathfrak L(h)^{-1} h$ on the velocity equation in \eqref{eq:greatlake} and using the weighted incompressibility condition \eqref{eq:weightedincomp}, one finds an elliptic problem for ${\sf d}\zeta$.
The Kelvin circulation theorem for the stochastic, rotating, thermal Great Lake equations is given by
\begin{equation}
{\sf d}\oint_{c(\overline{{\sf d}\boldsymbol \chi_t})}\left(\overline{\mathbf{V}}+\frac{1}{{\rm Ro}}\mathbf{R}\right)\cdot d\mathbf{x} = -\frac{\mathfrak{s}}{2}\oint_{c(\overline{{\sf d}\boldsymbol \chi_t})}h\nabla\overline{b} \cdot d\mathbf{x}dt.
\label{thm:KCGL}
\end{equation}
Here $c(\overline{{\sf d}\boldsymbol \chi_t})$ is any fluid loop that is being advected by the vertically averaged stochastic vector field $\overline{{\sf d}\boldsymbol \chi_t}$. The right hand side of the circulation theorem reveals that circulation will be generated when the gradients of the buoyancy and the bathymetry are not aligned. This term can be seen as a baroclinic torque. The proof that the rotating, thermal, Great Lake equations satisfy this Kelvin theorem is postponed to end of the next subsection, where we will derive the same set of equations from a variational principle.
\bigskip

\begin{remark}
Note that the small aspect ratio limit $\sigma \to 0$ reduces the Great Lake equations in \eqref{eq:greatlake} to the Lake equations in \eqref{eq:lake}. If the bathymetry is flat, then the weighted incompressibility condition in \eqref{eq:weightedincomp} reduces to the usual two dimensional incompressibility condition. In this case, the nonhydrostatic pressure term that is part of $\mathbf{V}$ vanishes and one obtains the two dimensional version of the stochastic, rotating, Euler equations.
\end{remark}

\subsection{Averaged Euler-Poincar\'e Lagrangian for long time -- very small wave scaling}\label{sec:variationallongtime}
To apply vertical averaging in the Euler-Poincar\'e setting, we return to the dimensionless Lagrangian \eqref{lag:ltEB} with boundary conditions given in \eqref{bc:ltEB}. In line with the derivation of the Great Lake equations from the Newtonian point of view above, we assume that the horizontal velocity is independent of the vertical coordinate. This can be guaranteed upon replacing the horizontal velocity by its vertical average. In that situation the expression for the vertical velocity in terms of the horizontal velocity in \eqref{eq:wdivu} can be integrated explicitly and we obtain as before
\begin{equation}
w = -\nabla\cdot(z+h)\overline{\bu} = \nabla\cdot(\alpha\zeta - z)\overline{\bu}.
\label{eq:wexpression}
\end{equation}
The same reasoning applies to the data vector fields, for which we obtain
\begin{equation}
\zh\cdot\boldsymbol \xi_{3i} = -\nabla\cdot(z+h)\overline{\boldsymbol \xi}_i = \nabla\cdot(\alpha\zeta - z)\overline{\boldsymbol \xi}_i.
\label{eq:xiexpression}
\end{equation}
Note that in the limit $\alpha\to 0$, the expression on the right hand side in \eqref{eq:wexpression} and \eqref{eq:xiexpression} implies the free surface boundary condition when $w$ and $\zh\cdot\boldsymbol \xi_{3i}$ are evaluated on the free surface. However, evaluation on the bottom boundary does not imply the boundary condition \eqref{bc:kinematic} unless the weighted incompressibility condition \eqref{eq:weightedincomp} holds. Substituting \eqref{eq:wexpression} into the Euler-Boussinesq Lagrangian \eqref{lag:ltEB} and replacing $\bu$ by $\overline{\bu}$ means that we can evaluate the vertical integral. Hence we have the approximate Euler-Boussinesq Lagrangian
\begin{equation}
\begin{aligned}
\ell_{EB}
&\approx \int_{CS}\int_{-h}^{\alpha\zeta} \left(\frac{1}{2}|\overline{\mathbf{u}}|^2 + \frac{\sigma^2}{2}\big((z+h)(\nabla\cdot\overline{\bu}) + (\overline{\bu}\cdot\nabla)h\big)^2 + \frac{1}{{\rm Ro}}\overline{\bu}\cdot\mathbf{R} -\frac{1}{{\rm Fr}^2}(1+\mathfrak{s}\overline{b})z\right)dz\,dx\,dy.
\end{aligned}
\end{equation} 
This Lagrangian is an integral over the horizontal cross section of the domain $\Omega$, which we call $CS$. Evaluating the vertical integral leads to
\begin{equation}
\begin{aligned}
\ell_{TRGN} &= \int_{CS}\left(\frac{1}{2}|\overline{\bu}|^2 + \frac{\sigma^2}{6}\eta^2(\nabla\cdot\overline{\bu})^2 + \frac{\sigma^2}{2}\eta(\nabla\cdot\overline{\bu})(\overline{\bu}\cdot\nabla h)+ \frac{\sigma^2}{2}(\overline{\bu}\cdot\nabla h)^2 + \frac{1}{{\rm Ro}}\overline{\bu}\cdot\bR \right.\\
&\quad \left. - \frac{1}{2\,{\rm Fr}^2}(1+\mathfrak{s}\overline{b})(\eta-2h)\right)\eta\,dx\,dy.
\end{aligned}
\label{lag:trgn}
\end{equation}
The subscript on the Lagrangian in \eqref{lag:trgn} stands for \textit{thermal rotating Green-Naghdi}, because the equations that this Lagrangian gives rise to are a thermal and rotating extension to the usual Green-Naghdi equations \cite{green1976derivation}. The incompressibility constraint has been used to ensure that the expression for the vertical velocity is valid and are thus no longer required. However, the weighted incompressibility condition \eqref{eq:weightedincomp} must still hold; so, we introduce a new constraint to make the total depth equal to the bathymetry. Weighted incompressibility has to be enforced via a constraint because it affects the measure $\eta\,dx\,dy$ in the Lagrangian above. The constraint is equivalent to saying that the free surface elevation is zero, that is, $\eta-h=0$. Thus, the action for the thermal rotating Great Lake equations is given by
\begin{equation}
S_{TRGL} = \int_{t_1}^{t_2}\ell_{TRGN}\,dt + \left\langle\frac{1}{{\rm Fr}^2}{\sf d}\pi,\eta-h\right\rangle =: \int_{t_1}^{t_2}c\ell_{TRGL}.
\label{action:trgl}
\end{equation}
The action in \eqref{action:trgl} has been suggestively called the thermal rotating Great Lake action and defines the constrained thermal rotating Great Lake Lagrangian. Note that this Lagrangian features the $H_{div}$ Sobolev norm in the situation where the bathymetry is flat, which has interesting relations with integrable systems and geometric statistics, as shown in \cite{khesin2013geometry}. When the bathymetry is nontrivial, the norm is more complicated. Here ${\sf d}\pi$ is a semimartingale Lagrange multiplier, whose purpose is to ensure that the weighted incompressibility condition holds. In order to apply the Euler-Poincar\'e theorem \ref{thm:SEP} to this Lagrangian, we need to define the variations. By substituting the higher order perturbation expansion \eqref{exp:higherorder} into the formulas for the variations in the theorem, we obtain
\begin{equation}
\begin{aligned}
\delta \overline{\mathbf{u}} \,dt &= {\sf d}\mathbf{v} - \big[{\sf d}\overline{\boldsymbol \chi_t}, \mathbf{v}\big],
\label{var:3D}
\end{aligned}
\end{equation}
where the arbitrary vector field $\mathbf{v}$ is a vector field semimartingale. The variations of the advected quantities are obtained by directly integrating the formulae for the variations in the three dimensional case. First we notice that the only advected quantities in this problem are scalar functions and volume forms, which due to incompressibility, satisfy the same form of advection equation, as we saw above in the Newtonian averaging principle. The functional derivative and spatial derivatives commute. Hence, if $\mathbf{u}_3$ is incompressible, then $\delta \mathbf{u}_3$ must be incompressible, as well. This argument implies that the arbitrary vector field is also incompressible, which means that the constraints for the variations of the buoyancy and the density can be shown to satisfy
\begin{equation}
\begin{aligned}
\delta\overline{b}\, dt &= - (\mathbf{v}\cdot\nabla)\overline{b}\,dt,\\
\delta\int_{-h}^{\alpha\zeta} D dz\, dt &= -\nabla\cdot \left(\int_{-h}^{\alpha\zeta} D dz\,\mathbf{v}\right)dt.
\end{aligned}
\end{equation}
In this paper, $D = 1$, so the vertical integral of $D$ is the depth $\eta = \alpha\zeta + h$, showing that the depth $\eta$ functions as a two dimensional density; hence, its variation satisfies
\begin{equation}
\delta\eta \,dt = -\nabla\cdot \left(\eta\mathbf{v}\right)dt.
\end{equation} 
In the $\alpha \to 0$ limit, the depth is given by the bathymetry $\eta = h$, which is the constraint introduced to imply weighted divergence. The variations of the thermal rotating Great Lake Lagrangian in \eqref{action:trgl} are
\begin{equation}
\begin{aligned}
\frac{\delta c\ell_{TRGL}}{\delta \overline{\mathbf{u}}} 
&= 
\eta\overline{\mathbf{u}} - \frac{\sigma^2}{3}\nabla(\eta^3\nabla\cdot \overline{\mathbf{u}}) - \frac{\sigma^2}{2}\nabla\big(\eta^2(\overline{\bu}\cdot\nabla h)\big) + \frac{\sigma^2}{2}\eta^2(\nabla\cdot\overline{\bu})\nabla h + \sigma^2\eta(\overline{\bu}\cdot\nabla h)\nabla h + \frac{1}{{\rm Ro}}\eta\mathbf{R}\,,
\\
\frac{\delta c\ell_{TRGL}}{\delta \eta} 
&=
 \frac{1}{2}|\overline{\mathbf{u}}|^2 + \frac{\sigma^2}{2}(\eta\nabla\cdot\overline{\bu}+\overline{\bu}\cdot\nabla h)^2
 + \frac{1}{{\rm Ro}}(\overline{\mathbf{u}}\cdot \mathbf{R}) - \frac{1}{{\rm Fr}^2}(1+\mathfrak{s}\overline{b})(\eta-h) - \frac{1}{{\rm Fr}^2}{\sf d}\pi,
\\
\frac{\delta c\ell_{TRGL}}{\delta \overline{b}} 
&= 
-\frac{\mathfrak{s}}{2\,{\rm Fr}^2}(\eta^2-2\eta h) \,, \\
\frac{\delta c\ell_{TRGL}}{\delta {\sf d}\pi} 
&=
\frac{1}{{\rm Fr}^2}(\eta - h)\,.
\end{aligned}
\end{equation} 
The variational derivative with respect to $\overline{\bu}$ of the thermal rotating Great Lake Lagrangian shows that the elliptic operator that relates $\overline{\mathbf{V}}$ and $\overline{\bu}$ that we encountered in \eqref{def:ellipticoperator} arises naturally in the variational context upon evaluating $\eta=h$. Note that the variational derivative with respect to $\eta$ simplifies considerably upon evaluating $\eta=h$. The second term in the variational derivative with respect to $\eta$ is the square of the weighted incompressibility condition, which is equal to zero. Hence this term disappears. The fourth term in the variational derivative with respect to $\eta$ vanishes since $\eta-h=0$. The variational derivative with respect to the buoyancy simplifies. An application of the stochastic Euler-Poincar\'e Theorem \ref{thm:SEP} to the Great Lake Lagrangian  in \eqref{action:trgl} with these variational derivatives and the variations in \eqref{var:3D} leads to the stochastic Great Lake equations \eqref{eq:greatlake}, with rotation and buoyancy. For notational convenience, let us use \eqref{def:ellipticoperator} to define $\overline{\mathbf{V}}$. Then the thermal rotating Great Lake equations are given by
\begin{equation}
\begin{aligned}
{\sf d}\overline{\mathbf{V}} + (\overline{{\sf d}\boldsymbol \chi_t}\cdot\nabla)\overline{\mathbf{V}} + (\nabla\overline{{\sf d}\boldsymbol \chi_t})\cdot \overline{\mathbf{V}} &= -\frac{1}{{\rm Fr}^2}\nabla{\sf d}\pi + \frac{1}{2}|\overline{\mathbf{u}}|^2 dt - \frac{\mathfrak{s}}{2\,{\rm Fr}^2}h\nabla \overline{b} dt - \frac{1}{{\rm Ro}}f\hat{\mathbf{z}}\times\overline{{\sf d}\boldsymbol \chi_t} - \frac{1}{{\rm Ro}}\nabla(\overline{\boldsymbol \xi}_i\cdot \mathbf{R})\circ dW_t^i,\\
{\sf d}\overline{b} + (\overline{{\sf d}\boldsymbol \chi_{t}} \cdot \nabla)\overline{b} &= 0,\\
\nabla\cdot(h\overline{{\sf d}\boldsymbol \chi_t}) &= 0,
\end{aligned}
\label{GL-pressure}
\end{equation}
and with the boundary condition 
\begin{equation}
\overline{{\sf d}\boldsymbol \chi}_t \cdot \mathbf{n} = 0.
\end{equation}
The pressure ${\sf d}\pi$ is solved for using the elliptic operator defined in \eqref{def:ellipticoperator}. Hence one can make the identification $\pi=\zeta$. This calculation shows that the Great Lake equations with rotation, stratification and stochasticity can be obtained by averaging the equations and using a perturbation series approach, or by taking a variational approach. The results are identically equal. Since the Lagrangian framework implies the Kelvin circulation theorem \eqref{thm:KCGL}, the proof is now immediate that the circulation theorem has the form
\begin{equation}
\begin{aligned}
{\sf d}\oint_{c(\overline{{\sf d}\boldsymbol \chi_t})}\left(\overline{\mathbf{V}}+\frac{1}{{\rm Ro}}\mathbf{R}\right)\cdot d\mathbf{x} &= -\frac{\mathfrak{s}}{2\,{\rm Fr}^2}\oint_{c(\overline{{\sf d}\boldsymbol \chi_t})}h\nabla\overline{b}\cdot d\mathbf{x} \,dt,\\
&= -\frac{\mathfrak{s}}{2\,{\rm Fr}^2}\int\!\!\int_{\partial S = c(\overline{{\sf d}\boldsymbol\chi_t})}\nabla h \times\nabla\overline{b} \,d\mathbf{S}\,dt.
\end{aligned}
\label{GL-KelThm}
\end{equation} 
Thus, in this scaling regime, applying asymptotics to the equations implies the same result as applying the asymptotics in the variational principle. \medskip

\begin{remark}[Kelvin theorem result for generation of horizontal circulation]
The Kelvin circulation theorem in \eqref{GL-KelThm} shows that any \textit{misalignment of the horizontal gradients} of the bathymetry and of the vertically averaged buoyancy \textit{will generate horizontal circulation} in the material  loop $c(\overline{{\sf d}\boldsymbol \chi_t})$ which follows the stochastic Lagrangian flow velocity $\overline{{\sf d}\boldsymbol \chi_t}$ in the horizontal plane given in equation \eqref{eq:chi-defn}. The Kelvin circulation theorem \eqref{GL-KelThm} implies an evolution equation for potential vorticity, as well.
\end{remark}

In the next section, we will extend the comparative asymptotic expansion approach to consider the short time - small wave limit. This extension will be accomplished by first deriving equations using asymptotics in the Euler-Boussinesq equations and later doing asymptotics in the Lagrangian and applying the Euler-Poincar\'e theorem.

\section{Short time - small wave scaling regime}\label{sec:STL-smallwavelimit}

Short time corresponds to choosing the time scale to be $T = L/\sqrt{gH}$, the time it takes for a gravity wave to traverse the horizontal  length scale. `Small wave' means that the amplitude of the wave is small, but not small enough to consider taking the rigid lid limit. In this setting, the scales are given by
\begin{framed}
\begin{equation}
\begin{matrix}
\mathbf{x}_3 = L(\mathbf{x}', \sigma z'), & \mathbf{u}_3 = U(\mathbf{u}', \sigma w'), &\nabla_3 = \dfrac{1}{L}\left(\nabla', \dfrac{1}{\sigma}\dfrac{\partial}{\partial z'}\right), & t = \dfrac{L}{\sqrt{gH}} t', & W_t = \sqrt{\dfrac{L}{\sqrt{gH}}}W_{t'},
\\[10pt] 
h = Hh', & \zeta = \alpha H\zeta',& \mathbf{R} = f_0 L \mathbf{R}', &  \rho = \rho_0 \rho', & {\sf d}p = \rho_0 gH {\sf d}p',
\\[10pt]
\sigma = \dfrac{H}{L}, & \alpha = \dfrac{\zeta_0}{H}, & {\rm Fr} = \dfrac{U}{\sqrt{gH}}, & {\rm Ro} = \dfrac{U}{f_0 L}, & {\rm Sr} = \dfrac{1}{{\rm Fr}}.
\end{matrix}
\label{tab:stscaling}
\end{equation}
\end{framed}
In this scaling regime, the EB Lagrangian takes the form
\begin{equation}
\ell_{EB}(\mathbf{u}_3,b,D) = \int_{\Omega}D\left( \frac{1}{2}|\mathbf{u}|^2 + \frac{\sigma^2}{2} w^2 + \frac{1}{{\rm Ro}}\mathbf{u}\cdot \mathbf{R} - \frac{1}{{\rm Fr}^2}(1+b)z\right)\,dx\,dy\,dz,
\label{lag:stEB}
\end{equation}
so the corresponding action is given by
\begin{equation}
S_{EB} = \int_{t_1}^{t_2}\ell_{EB} \,dt - \left\langle \frac{1}{{\rm Fr}^2}{\sf d}p, D-1 \right\rangle =: \int_{t_1}^{t_2} c\ell_{EB},
\end{equation}
with boundary conditions given by
\begin{equation}
\begin{matrix}
p = \alpha\zeta \qquad & \text{ at } z = \alpha\zeta(\mathbf{x},t),\\[4pt]
wdt + \zh\cdot\boldsymbol \xi_{3i}\circ dW_t^i = \alpha\left(\dfrac{1}{{\rm Fr}}{\sf d}\zeta + ({\sf d}\boldsymbol \chi_t\cdot\nabla)\zeta\right) \qquad & \text{ at } z = \alpha\zeta(\mathbf{x},t),\\[6pt]
wdt + \zh\cdot\boldsymbol \xi_{3i}\circ dW_t^i = -({\sf d}\boldsymbol \chi_t\cdot\nabla)h \qquad & \text{ at } z = -h(\mathbf{x}),\\[4pt]
{\sf d}\boldsymbol \chi_t\cdot \mathbf{n} = 0 \qquad & \text{ on lateral boundaries}.
\end{matrix}
\label{bc:stEB}
\end{equation}
An application of the stochastic Euler-Poincar\'e theorem \ref{thm:SEP} on the short-time scaled Lagrangian in \eqref{lag:stEB} yields the following equations
\begin{equation}
\begin{aligned}
\frac{1}{{\rm Fr}}{\sf d}\mathbf{u} + ({\sf d}\boldsymbol \chi_{3t}\cdot\nabla_3)\mathbf{u} + (\nabla\boldsymbol \xi_{3i})\cdot \mathbf{u}_3\circ dW_t^i &= -\frac{1}{{\rm Fr}^2}\nabla {\sf d}p - \frac{1}{{\rm Ro}}f\hat{\mathbf{z}}\times {\sf d}\boldsymbol \chi_t - \frac{1}{{\rm Ro}}\nabla(\boldsymbol \xi_i\cdot R)\circ dW_t^i,\\
\sigma^2\left(\frac{1}{{\rm Fr}} {\sf d}w + ({\sf d}\boldsymbol \chi_{3t}\cdot\nabla_3)w + \Big(\frac{\partial}{\partial z}\boldsymbol \xi_{3i}\Big)\cdot\bu_3\circ dW_t^i \right) &= -\frac{1}{{\rm Fr}^2}\frac{\partial}{\partial z}{\sf d}p  -\frac{1}{{\rm Fr}^2}(1 + b)dt,\\
\frac{1}{{\rm Fr}}{\sf d}b + ({\sf d}\boldsymbol \chi_{3t}\cdot\nabla_3)b &= 0,\\
\nabla_3\cdot{\sf d}\boldsymbol \chi_{3t} &= 0.
\end{aligned}
\label{eq:stEB}
\end{equation}
These equations satisfy the Kelvin circulation theorem, which for the Euler-Boussinesq equations takes the form of \eqref{thm:KCEB}, and also have conservation of potential vorticity along fluid trajectories, as in \eqref{eq:PVEB}, as well as conservation of an infinity of integral quantities \eqref{eq:CasimirsEB}, but now the Strouhal number is explicitly given in terms of the Froude number. In this scaling, the free surface is small rather than very small. Hence, we will not take the limit of the Froude number going to zero explicitly. Instead, we will introduce a regular perturbation expansion with small parameters $\epsilon$ and $\gamma$ whose magnitudes need to be determined with respect to $\alpha$, ${\rm Fr}$ and $\sigma$. 
\begin{equation}
\begin{matrix}
\mathbf{u} = \mathbf{u}_0 + \epsilon\mathbf{u}_1 + o(\epsilon), & w = w_0 +\epsilon w_1 + o(\epsilon), & \boldsymbol \xi_i = \boldsymbol \xi_{0,i} + \epsilon\boldsymbol \xi_{1,i}+ o(\epsilon),\\
{\sf d}\boldsymbol \chi_t = {\sf d}\boldsymbol \chi_{0,t} + \epsilon{\sf d}\boldsymbol \chi_{1,t}+ o(\epsilon), & p = p_0 + \gamma p_1 + \gamma^2 p_2 + o(\gamma^2), & b = \mathfrak{s} b_1 + \mathfrak{s}^2 b_2 +  o(\gamma^2).
\end{matrix}
\label{exp:multi}
\end{equation}
Substitution of \eqref{exp:multi} into \eqref{eq:stEB} provides equations of unknown order. By requiring certain balances to hold, the order of the dimensionless numbers can be related to each other. The boundary condition related to the vertical velocity at the free surface in \eqref{bc:stEB} implies that $\alpha = \mathcal O({\rm Fr})$. In the horizontal velocity equation, the leading order velocity ${\rm Fr}\,{\sf d}\mathbf{u}_0$ needs to be of the same order as $\gamma\nabla{\sf d} p_1$, which means that $\gamma = \mathcal O({\rm Fr})$. At the next order, ${\rm Fr}\, \epsilon\,{\sf d}\mathbf{u}_1$ is required to be of the same order as $\gamma^2\nabla{\sf d} p_2$, which implies that $\epsilon = \mathcal O({\rm Fr})$. In the vertical velocity equation, we want hydrostatic balance to be broken at $\mathcal O(\gamma^2)$, which means that ${\rm Fr}\,\sigma^2 {\sf d}w_0$ has to be of the same order as $\gamma^2\frac{\partial}{\partial z}{\sf d}p_2 $. It also implies for our ordering scheme that $\sigma^2 = \mathcal O({\rm Fr})$. In the Boussinesq approximation, we assumed that $\mathcal O(\mathfrak{s})=\mathcal O({\rm Fr})$. To summarise, our ordering scheme is now fixed to be
\begin{equation}
\mathcal O(\alpha) = \mathcal O(\mathfrak{s}) = \mathcal O(\gamma) = \mathcal O(\epsilon) = \mathcal O({\rm Fr}) = \mathcal O(\sigma^2).
\label{os:orderingscheme}
\end{equation}

\subsection{Averaging of Newton's 2nd law in the short time - small wave scaling}\label{stswnewton}
Averaging in the Newtonian equations leads to the following vertically averaged version of \eqref{eq:stEB},
\begin{equation}
\begin{aligned}
\frac{1}{{\rm Fr}}\,{\sf d} \overline{\mathbf{u}} + \frac{1}{\eta}\nabla\cdot(\eta\overline{{\sf d}\boldsymbol \chi_t\otimes \mathbf{u}}) + (\nabla\boldsymbol\xi_i)\cdot\overline{\mathbf{u}}\circ dW_t^i &= -\frac{1}{{\rm Fr}^2}\nabla \overline{{\sf d}p} - \frac{1}{{\rm Ro}} f \hat{\mathbf{z}}\times \overline{{\sf d}\boldsymbol \chi_t} - \frac{1}{{\rm Ro}}\nabla( \overline{\boldsymbol\xi}_i\cdot \mathbf{R})\circ dW_t^i,\\
\frac{1}{{\rm Fr}}{\sf d}\overline{b} + \nabla\cdot(\overline{b {\sf d}\boldsymbol \chi_t})&= 0,\\
\frac{1}{{\rm Fr}}{\sf d}\eta + \nabla\cdot(\eta\overline{{\sf d}\boldsymbol\chi_t}) &= 0,
\end{aligned}
\label{eq:stEBav}
\end{equation}
where $\overline{{\sf d}\boldsymbol \chi_t} $ is the vertical average of ${\sf d}\boldsymbol \chi_t $ in equation \eqref{exp:multi}; namely,
\begin{equation}
\overline{{\sf d}\boldsymbol \chi_t} 
:= {\sf d}\boldsymbol \chi_{0,t} + \epsilon\overline{{\sf d}\boldsymbol \chi_{1,t}}+ o(\epsilon).
\label{eq:sfchi-bar-def} 
\end{equation}
In this part of our discussion, we will not consider a leading order expansion before doing a higher order expansion. Instead, we work with directly with the expansion introduced in \eqref{exp:multi} and use the ordering scheme \eqref{os:orderingscheme} to apply single scale asymptotics. \bigskip

\begin{remark}
It is possible to study the system \eqref{eq:stEBav} on its own. One can simplify the system by dropping the Coriolis terms and assume that the flow is irrotational. The equations \eqref{eq:stEBav} can then be written in the so-called Zakharov-Craig-Sulem formulation. Alternatively, one can reformulate the system in terms of the free surface elevation and the horizontal discharge. Both of these approaches are explained in great detail in lecture notes by D. Lannes \cite{lannes2019lecture}. See also \cite{lannes2013water} for a comprehensive and complete treatment of the general water wave problem and \cite{lannes2005well} for the wellposedness results on the water wave problem in two and three dimensions.
\end{remark}

At leading order in the vertical velocity equation one finds
\begin{equation}
\frac{\partial}{\partial z}{\sf d}p_0 + 1\,dt = 0,
\end{equation}
and from the horizontal velocity equation at the same order,
\begin{equation}
\nabla {\sf d}p_0 = 0,
\end{equation}
which implies hydrostatic balance.  This information determines the leading order pressure, upon integrating in the vertical direction, to find
\begin{equation}
{\sf d}p_0 = (const. - z)dt,
\end{equation}
for the leading order pressure. In remark \ref{semimartingaledecomposition} we discussed how to deal with the semimartingale equations when the stochasticity is absent. This allows us to compute the expression for $p_0$ above. The arbitrary constant is due to integration and will be eliminated later using the boundary condition for the pressure. At the next order in the vertical velocity equation, one finds
\begin{equation}
\frac{\partial}{\partial z}{\sf d}p_1 + b_1\,dt = 0.
\end{equation}
Vertical integration of the expression above leads to
\begin{equation}
{\sf d}p_1 = \left(-\int^z b_1 dz' + \psi(\mathbf{x},t)\right)dt,
\label{eq:p1}
\end{equation}
where $\psi(\mathbf{x},t)$ is an arbitrary function of horizontal coordinates and time, introduced by the integration. From the horizontal velocity equation at the same order, we have
\begin{equation}
{\sf d}\mathbf{u}_0 = -\nabla {\sf d}p_1.
\label{eq:u0}
\end{equation}
By applying the gradient to \eqref{eq:p1} and taking the vertical derivative of \eqref{eq:u0}, we can derive a relation between the horizontal velocity field and the buoyancy,
\begin{equation}
\frac{\partial}{\partial z}{\sf d}\mathbf{u}_0 = \nabla b_1 \,dt. 
\label{eq:u0b}
\end{equation}
From the buoyancy equation at order $\mathcal O(\mathfrak{s})$, it is clear that $b_1$ is independent of time. Upon integrating \eqref{eq:u0b} both vertically and in time, one finds
\begin{equation}
\mathbf{u}_0(\mathbf{x},z,t) = t\int^z\nabla b_1(\mathbf{x},z) dz' + \mathbf{u}_0'(\mathbf{x},t) + \widetilde{\mathbf{u}_0}(\mathbf{x},z).
\label{eq:u0full}
\end{equation}
Unless $\nabla b_1 = 0$, the first term in \eqref{eq:u0full} grows linearly in time. Consequently, we choose the buoyancy $b_1$ to have the following profile
\begin{equation}
b_1(z) = \widetilde{b} - Sz,
\label{eq:b1}
\end{equation}
where $\widetilde{b}$ is some constant background buoyancy and $S$ is some $\mathcal O(1)$ positive constant. Of course, one can choose a more complicated and more realistic dependence on the vertical coordinate, at the cost of making some computations slightly more involved. The first term in \eqref{eq:u0full} now vanishes. The third term in \eqref{eq:u0full} arose due to integration with respect to time, hence $\widetilde{\mathbf{u}_0}$ plays the role of the initial condition. It is also the only term that has $z$-dependence. So, let us choose an initial condition which is independent of the vertical coordinate. This choice leaves us with
\begin{equation}
\mathbf{u}_0(\mathbf{x},t) = \mathbf{u}_0'(\mathbf{x},t) + \widetilde{\mathbf{u}_0}(\mathbf{x}).
\end{equation} 
Hence $\mathbf{u}_0$ has no vertical dependence. We can then use the incompressibility condition \eqref{eq:incompressible} to obtain an expression for the vertical velocity as in \eqref{eq:wdivu}, but now only looking at the leading order component of this relation. This leads to 
\begin{equation}
w_0 = -(z+h)\nabla\cdot \mathbf{u}_0,
\label{eq:w0}
\end{equation}
provided the variations of the bathymetry are small enough. Substituting the expression for the leading order vertical velocity into the vertical velocity equation at order $\mathcal O(\gamma^2)$ yields
\begin{equation}
-(z+h){\sf d}(\nabla\cdot \mathbf{u}_0) + \frac{\partial}{\partial z} {\sf d}p_2 + b_2\,dt = 0.
\end{equation}
From the equation above, we can determine an expression for $p_2$. Rearranging and taking a vertical integral yield 
\begin{equation}
{\sf d}p_2 = \left(\frac{1}{2}z^2+zh\right){\sf d}(\nabla\cdot \mathbf{u}_0) - \int^z b_2 dz'dt + \psi'(\mathbf{x},t).
\label{eq:p2}
\end{equation}
Since the expressions for ${\sf d}p_1$ and ${\sf d}p_2$ in \eqref{eq:p1} and \eqref{eq:p2}, respectively, involve the unknown functions $\psi(\mathbf{x},t)$ and $\psi'(\mathbf{x},t)$, we are not yet in the position to write down  the average of the pressure. By means of the dynamic boundary condition \eqref{bc:dynamic} and the expansion for the pressure in \eqref{exp:multi}, though, we can write
\begin{equation}
\begin{aligned}
0 &= [{\sf d}p_0 + \gamma {\sf d}p_1 + \gamma^2 {\sf d}p_2 + \mathcal O(\gamma^3)]|_{z=\alpha\zeta} dt\\
&= (const.\,dt - \alpha{\sf d}\zeta + \gamma\left(-\int^{\alpha\zeta}b_1 dz' + \psi(\mathbf{x},t)\right)dt + \gamma^2\Bigg[ \left(\frac{1}{2}\alpha^2 \zeta^2 + \alpha\zeta h\right){\sf d}(\nabla\cdot \mathbf{u}_0)\\
&\quad - \left(\int^{\alpha\zeta} b_2 dz' + \psi'(\mathbf{x},t)\right)\Bigg]dt + \mathcal O(\gamma^3).
\end{aligned}
\end{equation}
The difference between the pressure at the free surface and elsewhere in the domain can now be evaluated. In particular, functions that are independent of $z$ will be eliminated in this procedure and we are left with
\begin{equation}
\begin{aligned}
{\sf d}p &= -z\,dt-\alpha{\sf d}\zeta + \gamma \int_z^{\alpha\zeta}b_1dz'dt\\
&\quad  + \gamma^2\left[\left(\frac{1}{2}(z^2 - \alpha^2\gamma^2)+(z-\alpha\gamma)h\right){\sf d}(\nabla\cdot \mathbf{u}_0) + \int_z^{\alpha\zeta}b_2 dz'dt\right] + \mathcal O(\gamma^3).
\end{aligned}
\end{equation}
We can now determine the gradient of the pressure and collect terms that are of order $\mathcal O(\gamma^3)$ or equivalent in the remainder. Since $b_1$ does not depend on the horizontal coordinates, the gradient of $b_1$ vanishes and we have
\begin{equation}
\nabla {\sf d}p = \alpha(1+\widetilde{b})\nabla{\sf d}\zeta + \gamma^2\left[\left(\frac{1}{2}z^2 + zh\right){\sf d}\nabla(\nabla\cdot \mathbf{u}_0) + \int_z^0 \nabla b_2 dz'dt\right]+ \mathcal O(\gamma^3,\alpha^2\gamma,\alpha\gamma^2),
\end{equation}
where the contribution of $\widetilde{b}$ is due to the evaluation of $b_1$ at the free surface boundary. By taking the vertical average of the pressure gradient and switching the order of integration on the $b_2$ term, we obtain
\begin{equation}
\overline{\nabla {\sf d}p} = \alpha(1+\widetilde{b})\nabla{\sf d}\zeta + \gamma^2\left(\frac{1}{3}h^2 {\sf d}\nabla(\nabla\cdot \mathbf{u}_0) + \overline{(z+h)\nabla b_2}dt\right)+ \mathcal O(\gamma^3,\alpha^2\gamma,\alpha\gamma^2).
\end{equation}
At this stage, we can make a choice. We can use the averaged equation for the advection of buoyancy \eqref{eq:ASEB}, or we can use the expanded buoyancy equation and find an equation for the evolution $\overline{(z+h)\nabla b_2}$. The latter choice dictates that we look at the expanded buoyancy equation at order $\mathcal O(\gamma^2)$, where we have
\begin{equation}
{\sf d}b_2 - S(z+h)(\nabla\cdot \mathbf{u}_0)dt = 0.
\end{equation}
Here we have used \eqref{eq:b1} and \eqref{eq:w0}. By taking the gradient, then multiplying by $(z+h)$ and taking the average, we obtain after some algebra
\begin{equation}
{\sf d}\overline{(z+h)\nabla b_2} = S\left(\frac{1}{3}h^2\nabla(\nabla\cdot \mathbf{u}_0)\right)dt.
\end{equation}
Similar to the derivation of the Great Lake equations, the difference between the average of the nonlinearity and the product of the average is of higher order, since $\mathbf{u}_0$ is independent of the vertical coordinate. Therefore, we can also express $\overline{\mathbf{u}} = \mathbf{u}_0 + \epsilon\overline{\mathbf{u}_1} + \mathcal O(\epsilon^2)$. At this stage, one follows \cite{camassa1992dispersive} to introduce the variables
\begin{equation}
\begin{aligned}
\mathbf{A}:&= \overline{(z+h)\nabla b_2},\\
\mathbf{D}:&= \frac{1}{3}h^2\nabla(\nabla\cdot \overline{\mathbf{u}}),
\end{aligned}
\label{def:AD}
\end{equation}
and writes the following set of stochastic partial differential equations (SPDEs),
\begin{equation}
\begin{aligned}
\frac{1}{{\rm Fr}}\,{\sf d}\overline{\mathbf{u}} + (\overline{{\sf d}\boldsymbol \chi}_t\cdot\nabla)\overline{\mathbf{u}} + (\nabla\boldsymbol\xi_i)\cdot\overline{\mathbf{u}}\circ dW_t^i &= -\frac{\alpha}{{\rm Fr}^2}(1+\widetilde{b})\nabla{\sf d}\zeta - \mathbf{A}dt + {\sf d}\mathbf{D}\\
& \qquad - \frac{1}{{\rm Ro}}f\hat{\mathbf{z}}\times\overline{{\sf d}\boldsymbol \chi}_t - \frac{1}{{\rm Ro}}\nabla(\overline{\boldsymbol \xi}_i\cdot \mathbf{R})\circ dW_t^i,\\
\frac{\alpha}{{\rm Fr}}{\sf d}\zeta + \nabla\cdot(\alpha\zeta + h)\overline{{\sf d}\boldsymbol \chi}_t &= 0,\\
{\sf d}\mathbf{A} &= S\mathbf{D}dt.
\end{aligned}
\label{eq:ch1992}
\end{equation}
where $\overline{{\sf d}\boldsymbol \chi_t} $ is defined in equation \eqref{eq:sfchi-bar-def}.

Equations \eqref{eq:ch1992} comprise the stochastic version of those obtained in \cite{camassa1992dispersive}, provided one sets the dynamic boundary condition to $p=\widetilde{p}$, rather than zero. \bigskip


In the special case of deterministic, irrotational motion around the quiescent state $\overline{\mathbf{u}}=0$, the covector quantities $\mathbf{A}$ and $\mathbf{D}$ form an oscillator pair which oscillates with the Brunt-V\"ais\"al\"a frequency $S$. Also, in the deterministic case, an elimination procedure allows one to derive the Kadomtsev-Petviashvili equation and subsequently the Korteweg-De Vries equation for shallow water waves, as is done in \cite{camassa1992dispersive}. The direct approach for the derivations for water wave equations requires the substitution of the velocity field into the free surface equation, which requires time derivatives. In the stochastic case, however, one cannot take these time derivatives; so, the corresponding stochastic shallow water wave equations cannot be derived by using SALT. If instead, one takes a pathwise approach so that at least one time derivative can be taken, then the corresponding water-wave equations can be derived in this framework. In the next subsection, a hierarchy of stochastic water-wave equations is derived from the variational point of view.\bigskip

The set of equations \eqref{eq:ch1992} can be solved by observing that the operator $F$, defined by $F\overline{\mathbf{u}}:= \overline{\mathbf{u}} - \frac{\gamma^2}{3}h^2\nabla(\nabla\cdot\overline{\mathbf{u}})$, is a positive definite, self-adjoint and invertible operator. The Kelvin circulation theorem takes the following form for the equations in \eqref{eq:ch1992},
\begin{equation}
\frac{1}{{\rm Fr}}{\sf d}\oint_{c(\overline{{\sf d}\boldsymbol \chi_t})}\left(\overline{\mathbf{u}} - \mathbf{D} + \frac{1}{{\rm Ro}}\mathbf{R}\right)\cdot d\mathbf{x} = - \frac{1}{{\rm Fr}^2}\oint_{c(\overline{{\sf d}\boldsymbol \chi}_t)} \left((\overline{{\sf d}\boldsymbol \chi_t}\cdot\nabla)\mathbf{D} + (\nabla\overline{\boldsymbol \xi}_i)\cdot\mathbf{D}\circ dW_t^i- \mathbf{A}dt\right)\cdot d\mathbf{x}.
\end{equation}
Note that besides the buoyancy term $\mathbf{A}$, also transport terms show up on the right hand side. These transport terms indicate that these fluid equations are not geometric, in the sense that geometric fluid equations will only feature the relevant forces on the right hand side. The reason that these transport terms appear is that strict asymptotics sees the advection constraint \eqref{eq:nondimconstraint} as two individual terms, rather than as two objects that should always go together. Possibly, a multiscale analysis approach would be able to resolve this problem. This issue is discussed extensively in \cite{gjaja1996self}. We will resolve this issue by linking these two objects in a variational principle for a system closely related to \eqref{eq:ch1992}. First we will investigate the one-dimensional equations related to \eqref{eq:ch1992}.

\subsection{Stochastic Benjamin-Bona-Mahony equations}\label{sec:bbm}
From the stochastic CH92 equations in \eqref{eq:ch1992}, one cannot derive the stochastic Kadomtsev-Petviashvili equation and further simplify to obtain the stochastic Korteweg--De Vries equation. This is due to the fact that an elimination procedure involving time derivatives was used. However, by restricting to one dimensional motion, we do obtain the stochastic versions of familiar one dimensional water wave models. To be able to restrict to one dimension, we ignore the effect of rotation. The variable $\mathbf{A}$ is related to the buoyancy at higher order. By replacing $b_2$ with the vertical average $\overline{b}_2$ in the definition of $\mathbf{A}$ in \eqref{def:AD}, we can explicitly evaluate the integral. In calculating the integral, we keep in mind that the equations are written up to order $\mathcal O(\gamma^2)$. This requires us to drop the free surface terms that arise due to the vertical integral. The equations that we obtain from \eqref{eq:ch1992} are 
\begin{equation}
\begin{aligned}
\frac{1}{{\rm Fr}} {\sf d}\overline{u} - \frac{\gamma^2}{3\,{\rm Fr}^2}h^2 {\sf d}\overline{u}_{xx} + \overline{{\sf d}\chi}_t\, \overline{u}_x + \overline{u} \,(\overline{\xi}_i)_x\circ dW_t^i &= -\alpha(1+\widetilde{b}){\sf d}\eta_x  - \frac{\gamma^2}{2}
h^2 (\overline{b}_2)_x dt,\\
\frac{1}{{\rm Fr}}{\sf d}\eta + (\eta\,\overline{{\sf d}\chi}_t)_x &= 0,\\
\frac{1}{2}h^2 {\sf d}\overline{b}_2 &= \frac{S}{3}h^2 u_{xx}dt.
\end{aligned}
\label{eq:BBM3}
\end{equation} 
The set of equations given by \eqref{eq:BBM3} can be interpreted as a non--unidirectional, stochastic version of the Benjamin-Bona-Mahony (BBM) equation, first derived in \cite{benjamin1972model}, that includes the effects of depth and buoyancy stratification. Since this set \eqref{eq:BBM3} consists of three equations, we will refer to this set as BBM3. Upon ignoring the effect of buoyancy stratification, we will obtain the two component version of BBM3, which we will call BBM2. This set of equations is given by 
\begin{equation}
\begin{aligned}
\frac{1}{{\rm Fr}} {\sf d}\overline{u} - \frac{\gamma^2}{3\,{\rm Fr}^2}h^2 {\sf d}\overline{u}_{xx} + \overline{{\sf d}\chi}_t\, \overline{u}_x + \overline{u} \,(\overline{\xi}_i)_x\circ dW_t^i &= -\alpha(1+\widetilde{b}){\sf d}\eta_x,\\
\frac{1}{{\rm Fr}}{\sf d}\eta + (\eta\,\overline{{\sf d}\chi}_t)_x &= 0,\\
\end{aligned}
\label{eq:BBM2}
\end{equation}
The two component version \eqref{eq:BBM2} still is affected by the variations of the free surface. We assume that the bathymetry is flat, which means that we let $h\mapsto h_0$ and $h_0$ is constant in space and in time. We also assume that the free surface elevation is zero. These assumptions lead to the stochastic BBM equation, given by 
\begin{equation}
\frac{1}{{\rm Fr}}{\sf d}\overline{u} - \frac{\gamma^2}{3\,{\rm Fr}^2}h_0^2{\sf d}\overline{u}_{xx} + \overline{{\sf d}\chi}_t\, \overline{u}_x + \overline{u}\,(\overline{\xi}_i)_x\circ dW_t^i = 0.
\label{eq:BBM}
\end{equation}
Upon including linear wave speed in formulation of \eqref{eq:BBM} and ignoring stochasticity, we arrive at the celebrated BBM equation \cite{benjamin1972model},
\begin{equation}
\frac{1}{{\rm Fr}}\overline{u}_t - \frac{\gamma^2}{3\,{\rm Fr}^2}h_0^2 \overline{u}_{xxt} + \overline{u}\,\overline{u}_x + \kappa\overline{u}_x = 0.
\label{eq:uniBBM}
\end{equation}
Here $\kappa$ is a positive constant that enforces unidirectionality. The deterministic unidirectional BBM equation \eqref{eq:uniBBM} is similar in shape to the Korteweg--De Vries equation, but is not completely integrable. Next, we consider the averaging procedure in this section from the Euler-Poincar\'e perspective.

\subsection{Averaged Euler-Poincar\'e Lagrangian for short time - small wave scaling}\label{sec:variationalshorttime}
In the previous section, we used direct asymptotics to derive the stochastic version of the equations in \cite{camassa1992dispersive}. These equations failed to satisfy the Kelvin circulation theorem in a reasonable form. This difficulty will be overcome in the Euler-Poincar\'e approach, because the variational approach is able to cope with arbitrary Strouhal number. The starting point is the thermal rotating Green-Naghdi Lagrangian in \eqref{lag:trgn}. This time, we are not interested in the rigid lid limit, so our action is given by
\begin{equation}
S_{TRGN} = \int_{t_1}^{t_2} \ell_{TRGN}\,dt.
\label{action:trgn}
\end{equation}
We will now take variations in much the same way as done for the Great Lake equations in the Euler-Poincar\'e approach. However, there is a crucial difference. In the present scaling regime, the Strouhal number ${\rm Sr}$ is not equal to unity. Instead, we have $Sr=1/{\rm Fr}$, which is the inverse Froude number. Consequently, in the present case, the Euler-Poincar\'e variations of the velocities are taken as,
\begin{equation}
\begin{aligned}
\delta \overline{\mathbf{u}}\,dt &= {\rm Sr}\,{\sf d}\mathbf{v} - [\overline{{\sf d}\boldsymbol \chi_t}, \mathbf{v}] = \frac{1}{{\rm Fr}}{\sf d}\mathbf{v} - [\overline{{\sf d}\boldsymbol \chi_t}, \mathbf{v}].
\label{var:st}
\end{aligned}
\end{equation} 
The averaging has already occured in deriving the Lagrangian \eqref{lag:trgn}, where the Strouhal number does not explictly appear. In the variational approach, the Strouhal number appears in the variation of the velocity field. We stick with the ${\rm Sr}$ notation to show the flexibility that one has with the variational approach. By selecting the value of the Strouhal number later, the results of the previous section can be recovered by truncating higher order terms. The variational derivatives of the nondimensional Lagrangian $\ell_{rtGN}$ in equation \eqref{lag:trgn} are the following:
\begin{equation}
\begin{aligned}
\frac{\delta \ell_{TRGN}}{\delta \overline{\mathbf{u}}} 
&= 
\eta\overline{\mathbf{u}} - \frac{\sigma^2}{3}\nabla(\eta^3\nabla\cdot \overline{\mathbf{u}}) - \frac{\sigma^2}{2}\nabla\big(\eta^2(\overline{\bu}\cdot\nabla h)\big) + \frac{\sigma^2}{2}\eta^2(\nabla\cdot\overline{\bu})\nabla h + \sigma^2\eta(\overline{\bu}\cdot\nabla h)\nabla h + \frac{1}{{\rm Ro}}\eta\mathbf{R}\,,
\\
\frac{\delta \ell_{TRGN}}{\delta \eta} 
&=
 \frac{1}{2}|\overline{\mathbf{u}}|^2 + \frac{\sigma^2}{2}(\eta\nabla\cdot\overline{\bu}+\overline{\bu}\cdot\nabla h)^2
 + \frac{1}{{\rm Ro}}(\overline{\mathbf{u}}\cdot \mathbf{R}) - \frac{1}{{\rm Fr}^2}(1+\mathfrak{s}\overline{b})(\eta-h) - \frac{1}{{\rm Fr}^2}{\sf d}\pi,
\\
\frac{\delta \ell_{TRGN}}{\delta \overline{b}} 
&= 
-\frac{\mathfrak{s}}{2\,{\rm Fr}^2}(\eta^2-2\eta h) \,, \\
\end{aligned}
\label{var:TRGN}
\end{equation}
For notational convenience, we define
\begin{equation}
h\overline{\mathbf{V}} = h\overline{\mathbf{u}} + \left[-\frac{\sigma^2}{3}\nabla(\eta^3\nabla\cdot\overline{\mathbf{u}}) - \frac{\sigma^2}{2}\nabla(\eta^2\overline{\mathbf{u}}\cdot\nabla h) + \frac{\sigma^2}{2}\eta^2(\nabla\cdot\overline{\mathbf{u}})\nabla h + \sigma^2 \eta(\overline{\mathbf{u}}\cdot\nabla h)\nabla h\right].
\end{equation}
In the rigid lid case, one recovers \eqref{def:ellipticoperator}. A careful application of the Lax-Milgram theorem is able to show that $\overline{\mathbf{u}}$ depends continuously on $\overline{\mathbf{V}}$. Using this notation, an application of the stochastic Euler-Poincar\'e theorem \ref{thm:SEP} with the velocity variations given in \eqref{var:st} and the variational derivatives in \eqref{var:TRGN} of the Lagrangian $\ell_{TRGN}$ in \eqref{lag:trgn} yields the following SPDEs,
\begin{equation}
\begin{aligned}
{\rm Sr}\,{\sf d}\overline{\mathbf{V}}
+ (\overline{{\sf d}\boldsymbol \chi_t}\cdot\nabla)\overline{\mathbf{V}}
+ (\nabla \overline{{\sf d}\boldsymbol \chi_t})\cdot\overline{\mathbf{V}}
&=  - \frac{\alpha}{{\rm Fr}^2}\nabla\big((1+\mathfrak{s}\overline{b})\zeta\big)dt 
+ \frac{1}{2}\nabla|\overline{\mathbf{u}}|^2 dt
+ \frac{\sigma^2}{2}\nabla(\eta\nabla\cdot\overline{\bu} + \overline{\bu}\cdot\nabla h)^2\,dt
\\
&\quad
+ \frac{\mathfrak{s}}{2\,{\rm Fr}^2}(\alpha\zeta - h)\nabla\overline{b}\,dt
- \frac{1}{{\rm Ro}}f\hat{\mathbf{z}}\times\overline{{\sf d}\boldsymbol \chi_t} 
- \frac{1}{{\rm Ro}}\nabla(\overline{\boldsymbol\xi}_i\cdot \mathbf{R})\circ dW_t^i
,\\
{\rm Sr}\,\alpha{\sf d}\zeta + \nabla\cdot\big((\alpha\zeta+h) \overline{{\sf d}\boldsymbol \chi_t}\big) &= 0
,\\
{\rm Sr}\,{\sf d}\overline{b} + \overline{{\sf d}\boldsymbol \chi_t}\cdot\nabla\overline{b} &= 0.
\end{aligned}
\label{eq:KelThm-VCH92}
\end{equation}
where $\overline{{\sf d}\boldsymbol \chi_t} $ is defined in equation \eqref{eq:sfchi-bar-def}. It is useful to note that $\eta^{-1}\delta\ell_{TRGN}/\delta \overline{b} = (\mathfrak{s}/2)(\eta - 2h) = (\mathfrak{s}/2)(\alpha\zeta - h)$, since $\eta = \alpha\zeta + h$. These equations do satisfy a Kelvin circulation theorem, as they have been derived from the Euler-Poincar\'e variational principle. The circulation theorem takes the following form
\begin{equation}
\label{eq:KelThm-rsGN}
\begin{aligned}
{\rm Sr}\,{\sf d} \oint_{c(\overline{{\sf d}\boldsymbol \chi_t})} \left(\overline{\mathbf{V}}
+ \frac{1}{{\rm Ro}}\mathbf{R} \right)\cdot d\mathbf{x} 
&= \oint_{c(\overline{{\sf d}\boldsymbol \chi_t})} \frac{\mathfrak{s}}{2\,{\rm Fr}^2}(\alpha\zeta - h) \nabla \overline{b} \cdot d\mathbf{x}
\\&= 
\int\!\!\int_{\partial S = {c(\overline{{\sf d}\boldsymbol \chi_t})}} 
\frac{\mathfrak{s}}{2\,{\rm Fr}^2} \nabla(\alpha\zeta - h)\times \nabla \overline{b}\cdot d\mathbf{S} \,dt\,.
\end{aligned}
\end{equation}

As expected from equations \eqref{eq:circLaw} and \eqref{thm:KCEB} for the Kelvin circulation theorem which follows from the Euler-Poincar\'e equation \eqref{eq:EPeq} in three dimensions, circulation is created by misalignment of the gradients of vertically averaged buoyancy $\overline{b}$ and its dual quantity $\eta^{-1}\delta\ell_{TRGN}/\delta \overline{b}$, for the thermal rotating Green-Naghdi Lagrangian in equation \eqref{lag:trgn}. This is a balanced statement, because gradients of the bathymetry are assumed to be small. Interestingly, the misalignment of the gradient of vertically averaged buoyancy $\overline{b}$ and the difference $(\alpha\zeta - h)$ generates horizontal circulation (vertical vorticity). This represents a  barotropic mechanism for cyclogenesis (emergence of horizontal circulation, or eddies) in the ocean. The dispersion relation that corresponds to the linearised, deterministic version of equations \eqref{eq:KelThm-VCH92} is discussed in Appendix \ref{app:dispersion}.  A Kelvin circulation theorem similar to that in \eqref{eq:KelThm-rsGN} holds for the thermal rotating shallow water (TRSW) equations, as discussed in Appendix \ref{app:TRSW}.\bigskip

\begin{remark}[Comparison with JEBAR for ocean currents]
For the deterministic case, one replaces $c(\overline{{\sf d}\boldsymbol \chi_t})\to c(\overline{\mathbf{u}})$ and  the circulation theorem in \eqref{eq:KelThm-rsGN} recalls an aspect of the JEBAR (Joint Effect of Baroclinicity and Bottom Relief) approach for modelling the dynamics of ocean currents \cite{sarkisyan1971combined, cane1998utility, mellor1999comments, sarkisyan2006forty, colin2016direct}. Namely, the creation of circulation in \eqref{eq:KelThm-rsGN} occurs when the gradients of certain fluid properties are not aligned with the gradient of the bottom topography, $\nabla h(\mathbf{x})$. 

There are also may differences of \eqref{eq:KelThm-rsGN} from JEBAR. In particular, the circulation dynamics in \eqref{eq:KelThm-rsGN} represents Kelvin's theorem as derived from a vertically averaged and asymptotically expanded Hamilton's principle for Euler's fluid equations for the stochastic dynamics of an incompressible, thermal, rotating fluid flow with a free upper surface moving under the influence of gravity. Nonetheless, many of the physical principles underlying the derivation of \eqref{eq:KelThm-rsGN} also relate to principles which could be applied in the oceanographic setting for JEBAR. Hence, it may be advisable to investigate the utility of the present stochastic, asymptotic, vertically-averaged variational approach for some applications in oceanography.
\end{remark}

\paragraph{Potential vorticity.}
In the circulation theorem for the rotating, thermal, Great Lake equations in equation \eqref{thm:KCGL}, the circulation is generated by the misalignment between the horizontal gradient of the bathymetry and the horizontal gradient of the buoyancy. Here, we have seen that the misalignment of horizontal gradients of the free surface height with the horizontal gradient of the buoyancy also contributes to the generation of circulation. In terms of the potential vorticity given by
\begin{equation}
q := \eta^{-1}\big({\mathbf{\hat{z}}}\cdot\nabla\times(\overline{\mathbf{V}}+{\rm Ro}^{-1}\mathbf{R})\big)
\,,\label{PV-defn}
\end{equation}
the generation of circulation is accompanied by the following 
\begin{equation}
{\rm Sr}\,{\sf d}q + (\overline{{\sf d}\boldsymbol \chi}_t\cdot\nabla)q =  \frac{\mathfrak{s}}{2{\rm Fr}^2\eta}\,{\mathbf{\hat{z}}}\cdot\nabla(\alpha\zeta-h)\times\nabla\overline{b}.
\end{equation}
This shows that PV will also be generated by this misalignment of horizontal gradients. Equations \eqref{eq:KelThm-VCH92} also possess an infinity of conserved integral quantities of the following form
\begin{equation}
\label{eq:Casimir-rsGN}
C_{f,g} = \int_{CS} \big( f(\overline{b}) +  qg(\overline{b})\big)\eta\, dx dy,
\end{equation}
for arbitrary differentiable functions $f$, $g$ and for boundary conditions ${\sf d}\overline{\boldsymbol \chi}_t \cdot \mathbf{n} = 0$, $\nabla \overline{b}\times\mathbf{n}=0$. Invariance of the vertically averaged buoyancy $\overline{b}$ as it is advected along the tangential stochastic flow on the boundary is consistent with the latter condition, which requires the boundary to be a level set of $\overline{b}$. This can be shown by means of a direct computation using the equations of motion and the boundary conditions.
\bigskip

\subsection{Stochastic Camassa-Holm equations}\label{sec:stochCHeqns}
This section considers a sequence of reductions of the Lagrangian $\ell_{TRGN}(\overline{\mathbf{u}}, \eta, \overline{b})$ in equation \eqref{lag:trgn} in one spatial dimension which will eventually lead to the stochastic Camassa-Holm (CH) equation, considered in \cite{holm2016variational,crisan2018wave}
\begin{equation}
{\rm Sr}\,{\sf d}\overline{m} 
+ \big(\overline{m}\partial_x + \partial_x\overline{m}\big)\overline{{\sf d}\chi_t} = 0\,.
\end{equation}
In one dimension, we assume a flat bathymetry profile $h_0$ and ignore the effect of rotation. Applying these approximations to the thermal Green-Naghdi Lagrangian $\ell_{TRGN}(\overline{\mathbf{u}}, \eta, \overline{b})$ in equation \eqref{lag:trgn} yields the following Lagrangian at order $\mathcal O(\sigma^2)$,
\begin{equation}
\ell_{CH3} = \int_{-\infty}^{\infty} \left(\frac{1}{2}\overline{u}^2 + \frac{\sigma^2}{6} \eta^2 \overline{u}_x^2 - \frac{1}{2\,{\rm Fr}^2}(1+\mathfrak{s}\overline{b})(\eta - 2h_0)\right)\eta\,dx,
\label{lag:CH3}
\end{equation}
where we have completed the square on the potential energy term.
The domain of flow is taken to be the entire real line, rather than a compact line between two lateral boundaries as illustrated in figure \ref{fig:domain}. Boundary conditions on the real line require the vertically averaged velocity $\overline{u}$ and its horizontal spatial derivative $\overline{u}_x$ to vanish in the limit ${|x|\to\infty}$. The variational derivatives of the Lagrangian $\ell_{CH3}$ in \eqref{lag:CH3} are given by
\begin{equation}
\begin{aligned}
\overline{m}:= \frac{\delta\ell_{CH3}}{\delta \overline{u}} 
&= 
\eta \overline{u} - \frac{\sigma^2}{3}(\eta^3\overline{u}_{x})_x,
\\
\frac{\delta\ell_{CH3}}{\delta \eta} 
&= 
- \frac{1}{{\rm Fr}^2}(1+\mathfrak{s}\overline{b})(\eta - h_0),
\\
\frac{\delta\ell_{CH3}}{\delta \overline{b}}
&= 
-\frac{\mathfrak{s}}{2\,{\rm Fr}^2}(\eta^2-2\eta h_0).
\end{aligned}
\end{equation}
An application of the stochastic Euler-Poincar\'e theorem \ref{thm:SEP} then leads to the following set of three stochastic equations
\begin{equation}
\begin{aligned}
{\rm Sr}\,{\sf d}\overline{m} 
+ \big(\overline{m}\partial_x + \partial_x\overline{m}\big)\overline{{\sf d}\chi}_t
&= 
-\frac{1}{{\rm Fr}^2}\eta\big((1+\mathfrak{s}\overline{b})(\eta-h_0)\big)_x\,dt + \frac{\mathfrak{s}}{2\,{\rm Fr}^2}(\eta^2-2\eta h_0)\overline{b}_x\,dt,
\\
{\rm Sr}\,{\sf d}\eta + (\eta\, \overline{{\sf d}\chi}_t)_x &= 0,
\\
{\rm Sr}\,{\sf d}\overline{b} + \overline{{\sf d}\chi}_t\,\overline{b}_x &= 0.
\end{aligned}
\label{eq:CH3}
\end{equation}
The set of equations \eqref{eq:CH3} defines the three-component stochastic Camassa-Holm system (CH3). The stochastic evolution equation for momentum $\overline{m}$ includes the effects of varying depth and horizontal variations of the buoyancy. There follows a continuity equation for depth, $\eta$, and a scalar advection equation for buoyancy, $\overline{b}$. \bigskip

\begin{remark}[Is the deterministic CH3 case completely integrable?]
An investigation is underway elsewhere to determine whether the Lie--Poisson Hamiltonian system of CH3 equations in \eqref{eq:CH3} is completely integrable  in the deterministic case, where it simplifies to
\begin{equation}
\begin{aligned}
{\rm Sr}\,{\partial_t}\overline{m} 
+ \big(\overline{m}\partial_x + \partial_x\overline{m}\big)\overline{u}
&= 
-\frac{1}{{\rm Fr}^2}\eta\big((1+\mathfrak{s}\overline{b})(\eta-h_0)\big)_x + \frac{\mathfrak{s}}{2\,{\rm Fr}^2}(\eta^2-2\eta h_0)\overline{b}_x\,,
\\
{\rm Sr}\,{\partial_t}\eta + (\eta\, \overline{u} )_x &= 0,
\\
{\rm Sr}\,{\partial_t}\overline{b} + \overline{u}\,\overline{b}_x &= 0.
\end{aligned}
\label{eq:CH3-det}
\end{equation}
\end{remark}

We proceed farther now in the stochastic case by assuming that the vertically averaged buoyancy  $\overline{b}$ is constant in both space and time, so that we may replace $\overline{b}(x,t)\mapsto b_0$; a constant, or equivalently, by letting the stratification parameter tend to zero, $\mathfrak{s}\to 0$. Under this assumption, the Lagrangian $\ell_{CH3}$ simplifies, since the buoyancy term no longer contributes to the dynamics, and we arrive at the following Lagrangian $\ell_{CH2}$ for the stochastic two component Camassa-Holm (CH2) system:
\begin{equation}
\ell_{CH2} = \int_{-\infty}^{\infty} \left(\frac{1}{2}\overline{u}^2 + \frac{\sigma^2}{6}\eta^2 \overline{u}_x^2 - \frac{1}{2\,{\rm Fr}^2}(\eta - 2h_0)\right)\eta\, dx.
\label{lag:CH2}
\end{equation}
The variational derivatives of the Lagrangian $\ell_{CH2}$ in \eqref{lag:CH2} are given by
\begin{equation}
\begin{aligned}
\overline{m}:=\frac{\delta\ell_{CH2}}{\delta \overline{u}} 
&= 
\eta\overline{u} - \frac{\sigma^2}{3}(\eta^3\overline{u}_x)_x,
\\
\frac{\delta\ell_{CH2}}{\delta \eta} 
&= 
- \frac{1}{{\rm Fr}^2}(\eta - h_0).
\end{aligned}
\end{equation}
An application of the stochastic Euler-Poincar\'e theorem \ref{thm:SEP} with these variational derivatives  yields the following motion equation and advection law,
\begin{equation}
\begin{aligned}
{\rm Sr}\,{\sf d}\overline{m} 
+ \big(\overline{m}\partial_x + \partial_x\overline{m}\big)\overline{{\sf d}\chi}_t
&= 
-\frac{1}{{\rm Fr}^2}\eta\eta_x\,dt,
\\
{\rm Sr}\,{\sf d}\eta + (\eta\, \overline{{\sf d}\chi}_t)_x &= 0.
\end{aligned}
\label{eq:CH2}
\end{equation}
The set of equations \eqref{eq:CH2} is the stochastic two component Camassa-Holm (CH2) system. In the deterministic case, this set of equations is a completely integrable Hamiltonian system, as shown first by \cite{chen2006two}.

Finally, we will assume that the free surface elevation in the CH2 Lagrangian $\ell_{CH2}$ in \eqref{lag:CH2} is negligible. This assumption neglects the potential energy term in $\ell_{CH2}$, which then reduces to
\begin{equation}
\ell_{CH} = \int_{-\infty}^{\infty} \left(\frac{1}{2}\overline{u}^2 + \frac{\sigma^2}{6}h_0^2 \overline{u}_x^2\right) h_0 \, dx.
\label{lag:CH}
\end{equation}
The variation of the CH Lagrangian \eqref{lag:CH} with respect to the velocity $\overline{u}$ yields
\begin{equation}
\overline{m}:=\frac{\delta\ell_{CH}}{\delta \overline{u}} = h_0\overline{u} - \frac{\sigma^2}{3}h_0^3\overline{u}_{xx} 
\end{equation}
An application of the stochastic Euler-Poincar\'e theorem \ref{thm:SEP} then implies the SPDE,
\begin{equation}
{\rm Sr}\,{\sf d}\overline{m} 
+ \big(\overline{m}\partial_x + \partial_x\overline{m}\big)\overline{{\sf d}\chi}_t
= 0.
\label{eq:CH-nodisp}
\end{equation}
Equation \eqref{eq:CH-nodisp} is the \textit{dispersionless} stochastic Camassa-Holm equation, whose singular `peakon' solutions have been studied in \cite{holm2016variational, crisan2018wave}. Including cubic linear dispersion in the stochastic Camassa-Holm equation yields
\begin{equation}
{\rm Sr}\,{\sf d}\overline{m} 
+ \big(\overline{m}\partial_x 
+ \partial_x\overline{m} + \gamma \partial_x^3\big)\overline{{\sf d}\chi}_t
= 0\,.
\label{eq:CH-cubicdisp}
\end{equation}
The solution properties of this equation has been studied in \cite{holm2016variational, bendall2019perspectives}. When terms of order $O(\sigma^2)$ are neglected in equation \eqref{eq:CH-cubicdisp}, it reduces further  to the stochastic KdV equation, 
\begin{equation}
{\rm Sr}\,{\sf d}\overline{u} 
+ \big(\overline{u}\partial_x 
+ \partial_x\overline{u} + \gamma \partial_x^3\big)\overline{{\sf d}\chi}_t
= 0\,,
\label{eq:SKdV-cubicdisp}
\end{equation}
which has been studied in \cite{woodfield2019thesis}. The deterministic CH equation was first derived in \cite{camassa1993integrable, camassa1994new}, by using asymptotics on the Hamiltonian side. Here the stochastic CH equation has been derived by means of asymptotics in the Lagrangian for the rotating, thermal, Green-Naghdi equations \eqref{lag:trgn} followed by applying the stochastic Euler-Poincar\'e theorem to the approximated Lagrangian at a variety of levels.

\subsection{Differences between the Newtonian and variational approaches}\label{sec:NewVarDiffs}
There are several striking differences between the equations that one derives from the Newtonian approach and from the Euler-Poincar\'e approach, as illustrated with underbraces below. The most important difference is that the time derivative of $\mathbf{D}$ no longer appears explicitly in the equations above. Instead, the dynamical variable $\mathbf{V}$ appears naturally, as it did for the Great Lake equations  in \eqref{thm:KCGL}. The pressure and the buoyancy term also take slightly different forms. The averaged equations \eqref{eq:stEBav} indicate that the usage of the buoyancy equation is natural. In the Newtonian approach, the buoyancy only has dynamics at order $\sigma^4$, since $b_1$ was calculated explicitly and shown only to depend on the vertical coordinate. This explains the sole appearance of $b_2$ in the buoyancy equation. In the variational approach, we do not calculate the explicit profile of $b_1$, but instead we introduce a vertically averaged buoyancy in the Lagrangian. This means that the buoyancy is still allowed to vary horizontally, which can be seen in the equation for the buoyancy. The effect of the horizontal dependence of the buoyancy is important for the generation of horizontal circulation, as noticed in \eqref{eq:KelThm-rsGN}. Below we have expressed the two sets of equations in terms of the same variables so that the differences and similarities are clear.
\bigskip

\noindent CH92 equations:
\begin{align}
\begin{split}
\frac{1}{{\rm Fr}}{\sf d}\overline{\mathbf{u}} 
- \frac{\sigma^2}{3\,{\rm Fr}}h^2{\sf d}\nabla(\nabla\cdot \overline{\mathbf{u}})
&+ (\overline{{\sf d}\boldsymbol \chi}_t\cdot\nabla)\overline{\mathbf{u}} 
+ (\nabla\overline{{\sf d}\boldsymbol\chi}_t)\cdot\overline{\mathbf{u}}
\\&= -\frac{\alpha}{{\rm Fr}^2}\nabla\big((1+\widetilde{b})\zeta\big)\,dt
+ \frac{1}{2}\nabla|\overline{\mathbf{u}}|^2dt
- \overline{(z+h)\nabla b_2}\,dt
\\
& \quad - \frac{1}{{\rm Ro}}f\hat{\mathbf{z}}\times\overline{{\sf d}\boldsymbol \chi}_t
- \frac{1}{{\rm Ro}}\nabla(\overline{\boldsymbol \xi}_i\cdot \mathbf{R})\circ dW_t^i,
\\
\frac{\alpha}{{\rm Fr}}{\sf d}\zeta 
+ \nabla\cdot\big((\alpha\zeta + h)\overline{{\sf d}\boldsymbol \chi}_t\big) &= 0,
\\
{\sf d}\overline{(z+h)\nabla b_2} 
&= \frac{S}{3}h^2\nabla(\nabla\cdot \overline{\mathbf{u}})dt.
\end{split}
\label{eq:NCH92}
\end{align}
Thermal rotating Green--Naghdi equations:
\begin{equation}
\begin{aligned}
h\overline{\mathbf{V}} = h\overline{\mathbf{u}} + \left[-\frac{\sigma^2}{3}\nabla(\eta^3\nabla\cdot\overline{\mathbf{u}}) \right.&\left.- \frac{\sigma^2}{2}\nabla(\eta^2\overline{\mathbf{u}}\cdot\nabla h) + \frac{\sigma^2}{2}\eta^2(\nabla\cdot\overline{\mathbf{u}})\nabla h + \sigma^2 \eta(\overline{\mathbf{u}}\cdot\nabla h)\nabla h\right],\\
{\rm Sr}\,{\sf d}\overline{\mathbf{V}}
+ (\overline{{\sf d}\boldsymbol \chi}_t\cdot\nabla)\overline{\mathbf{V}}
+ (\nabla \overline{{\sf d}\boldsymbol \chi}_t)\cdot\overline{\mathbf{V}}
&=  - \frac{\alpha}{{\rm Fr}^2}\nabla\big((1+\mathfrak{s}\overline{b})\zeta\big)dt 
+ \frac{1}{2}\nabla|\overline{\mathbf{u}}|^2 dt
+ \frac{\sigma^2}{2}\nabla(\eta\nabla\cdot\overline{\bu} + \overline{\bu}\cdot\nabla h)^2\,dt
\\
&\quad
+ \frac{\mathfrak{s}}{2\,{\rm Fr}^2}(\alpha\zeta - h)\nabla\overline{b}\,dt
- \frac{1}{{\rm Ro}}f\hat{\mathbf{z}}\times\overline{{\sf d}\boldsymbol \chi}_t
- \frac{1}{{\rm Ro}}\nabla(\overline{\boldsymbol\xi}_i\cdot \mathbf{R})\circ dW_t^i
,\\
{\rm Sr}\,\alpha{\sf d}\zeta + \nabla\cdot\big((\alpha\zeta+h) \overline{{\sf d}\boldsymbol \chi}_t\big) &= 0
,\\
{\rm Sr}\,{\sf d}\overline{b} + \overline{{\sf d}\boldsymbol \chi}_t\cdot\nabla\overline{b} &= 0.
\end{aligned}
\label{eq:trgnfinal}
\end{equation}
where $\overline{{\sf d}\boldsymbol \chi_t} $ is defined in equation \eqref{eq:sfchi-bar-def}. Evaluating the Strouhal number ${\rm Sr} = \frac{1}{{\rm Fr}}$ and truncating at order $\mathcal{O}(1)$ in \eqref{eq:trgnfinal} provides the following set of equations 
\begin{equation}
\begin{aligned}
\frac{1}{\rm Fr}{\sf d}\overline{\mathbf{u}} - \frac{\sigma^2}{3\,{\rm Fr}}h^2{\sf d}\nabla(\nabla\cdot\overline{\bu}) 
&+ (\overline{{\sf d}\boldsymbol \chi}_t\cdot\nabla)\overline{\mathbf{u}}
+ (\nabla \overline{{\sf d}\boldsymbol \chi}_t)\cdot\overline{\mathbf{u}}
\\
&=  - \frac{\alpha}{{\rm Fr}^2}\nabla\big((1+\mathfrak{s}\overline{b})\zeta\big)dt 
+ \frac{1}{2}\nabla|\overline{\mathbf{u}}|^2 dt
+ \frac{\mathfrak{s}}{2\,{\rm Fr}^2}(\alpha\zeta - h)\nabla\overline{b}\,dt
\\&\quad
- \frac{1}{{\rm Ro}}f\hat{\mathbf{z}}\times\overline{{\sf d}\boldsymbol \chi}_t
- \frac{1}{{\rm Ro}}\nabla(\overline{\boldsymbol\xi}_i\cdot \mathbf{R})\circ dW_t^i
,\\
\frac{\alpha}{{\rm Fr}}{\sf d}\zeta + \nabla\cdot\big((\alpha\zeta+h) \overline{{\sf d}\boldsymbol \chi}_t\big) &= 0
,\\
\frac{1}{{\rm Fr}}{\sf d}\overline{b} + \overline{{\sf d}\boldsymbol \chi}_t\cdot\nabla\overline{b} &= 0.
\end{aligned}
\label{eq:trgnfinal2}
\end{equation}
There are still some differences between \eqref{eq:NCH92} and \eqref{eq:trgnfinal2}. In the variational approach, we introduce the vertically averaged buoyancy which gives rise to terms that create horizontal circulation, rather than introducing an explicit profile. The original CH92 equations in \eqref{eq:NCH92} were derived in \cite{camassa1992dispersive} by applying vertical averaging and strict asymptotics in the unapproximated equations in the form of Newton's force law for the fluid. Asymptotics in Strouhal number breaks the Kelvin circulation theorem. The thermal rotating Green-Naghdi equations in \eqref{eq:trgnfinal} have the following advantages over the CH92 equations

\begin{enumerate}
\item
They introduce a dynamical equation for the vertically averaged buoyancy, $\overline{b}$; 
\item
The dynamics of the vertically averaged buoyancy, $\overline{b}$, contributes to the pressure terms;
\item
They restore the Kelvin circulation theorem seen in equation \eqref{eq:KelThm-VCH92}; 
\item
They reveal a barotropic mechanism for horizontal circulation (cyclogenesis), as seen in equation \eqref{eq:KelThm-VCH92}; and
\item They allow for a hierarchy of Camassa-Holm equations to be derived, see subsection \ref{sec:stochCHeqns}.
\end{enumerate}

\section{Conclusion}\label{sec:Conclude}

\paragraph{Summary.}
This paper has extended the work of \cite{camassa1992dispersive} and \cite{camassa1996long, camassa1997long} by casting it into the framework of Hamilton's variational principle and including the multi-time effects of the Strouhal number and the barotropic effects of vertically-integrated buoyancy with horizontal gradients. As a result, a variety of new terms representing new effects relative to \cite{camassa1992dispersive} and \cite{camassa1996long, camassa1997long} have appeared in the resulting equations. For example, in the variational CH92 equations \eqref{eq:KelThm-VCH92} written in Kelvin circulation form in \eqref{eq:KelThm-rsGN} one sees how horizontal circulation (convection) is generated by an misalignment of horizontal gradient of vertically averaged buoyancy with the horizontal gradients of  bathymetry and/or surface elevation. 

Having extended the earlier work of \cite{camassa1992dispersive} and \cite{camassa1996long, camassa1997long} in a variational setting and expressed the results in Kelvin circulation form, the paper has also taken advantage of the variational framework  of \cite{holm2015variational} to include the effects of stochastic advective Lie transport (SALT). Including the effects of SALT introduces a new capability to quantify the uncertainty and then use data assimilation to reduce the uncertainty of the solutions of these equations due to unmodelled, or unresolved effects. A protocol for doing this has been been developed in \cite{cotter2018modelling, cotter2019numerically, cotter2019particle}. This protocol regards SALT as a type of `informed randomness' described by spatially correlated noise obtained from observed or simulated high-resolution data. This protocol may be applied to the present class of fluid equations. In order to reduce the investigation of these equations to their simplest form, the paper has derived the unidirectional version of the equation set in \eqref{eq:KelThm-VCH92} in the variational setting. This reduction has yielded stochastic versions of a family of CH equation, including the one derived in \cite{camassa1993integrable, camassa1994new}. These stochastic CH equations describe the interaction of solitons with noise. The first developments in this direction for the stochastic CH equation have already been studied in \cite{holm2016variational, holm2018new, crisan2018wave, bendall2019perspectives}. 

Two diagrams sketched below provide `roadmaps' of the two routes of simplification we have taken in this paper by using asymptotic expansions in the various small parameters for the ordering scheme in equation \eqref{os:orderingscheme}. The  Newtonian approach is shown in figure \ref{directapproach}. The corresponding road map for the variational approach is shown  in figure \ref{variationalapproach}.  

\begin{figure}[H]
\begin{center}
\textbf{Newtonian Approach}
\end{center}
\bigskip

\centering
\begin{tikzcd}
[row sep = 3em, 
column sep=2em, 
cells = {nodes={top color=green!25, bottom color=blue!25,draw=blue!90}},
arrows = {draw = black, rightarrow, line width = .03cm}]
& & {\begin{matrix}
\text{3D Euler equations for}\\ \text{for stratified, rotating}\\ \text{incompressible fluids}
\end{matrix}} \arrow[d,"{\begin{matrix} \text{small buoyancy}\\ \text{stratification}\end{matrix}}", shorten <= 1mm, shorten >= 1mm]  & &\\
& & \hyperref[eq:SEB]{\begin{matrix} \text{3D Euler Boussinesq equations}\\ \text{for stratified, rotating,}\\ \text{incompressible fluids}\end{matrix}} \arrow[dr, "{\begin{matrix} \text{vertical average,}\\ \text{long time - very small wave}\end{matrix}}", shorten <= 4mm, shorten >= 4mm] \arrow[dl, "{\begin{matrix} \text{vertical average,}\\ \text{short time - small wave}\end{matrix}}" swap, shorten <= 4mm, shorten >= 4mm]& &\\
& \hyperref[stswnewton]{\text{CH92 equations}} \arrow[d, "\text{restricting to 1D}", swap, shorten <= 1mm, shorten >= 1mm]  & & \hyperref[sec:rotatinggreatlake]{\begin{matrix} \text{Rotating, thermal,}\\\text{Great Lake equations} \end{matrix}} \arrow[d, "\text{small aspect ratio}",  shorten <= 1mm, shorten >= 1mm] & \\
& \hyperref[sec:bbm]{\begin{matrix} \text{BBM3, BBM2, BBM and} \\ \text{unidirectional BBM equations} \end{matrix}} & & \hyperref[sec:rotatinglake]{\text{Rotating Lake equations}} &
\end{tikzcd}
\caption{Diagram of derivations from the direct (or Newtonian) point of view. Each blue box refers to the set of equations that corresponds to the model referred to in the box. Clicking on the box will take the reader to the corresponding section. Above each arrow is the approximation that is necessary to transition from one set of equations to the next. Note that the short time - small wave approximation does not lead to rotating thermal Green-Naghdi, but to the CH92 equations. These lead to Benjamin-Bona-Mahony type equations when restricted to one dimensional motion. The rotating, stratified Euler equations are not linked because these equations have not been written down in this document.}
\label{directapproach}
\end{figure}
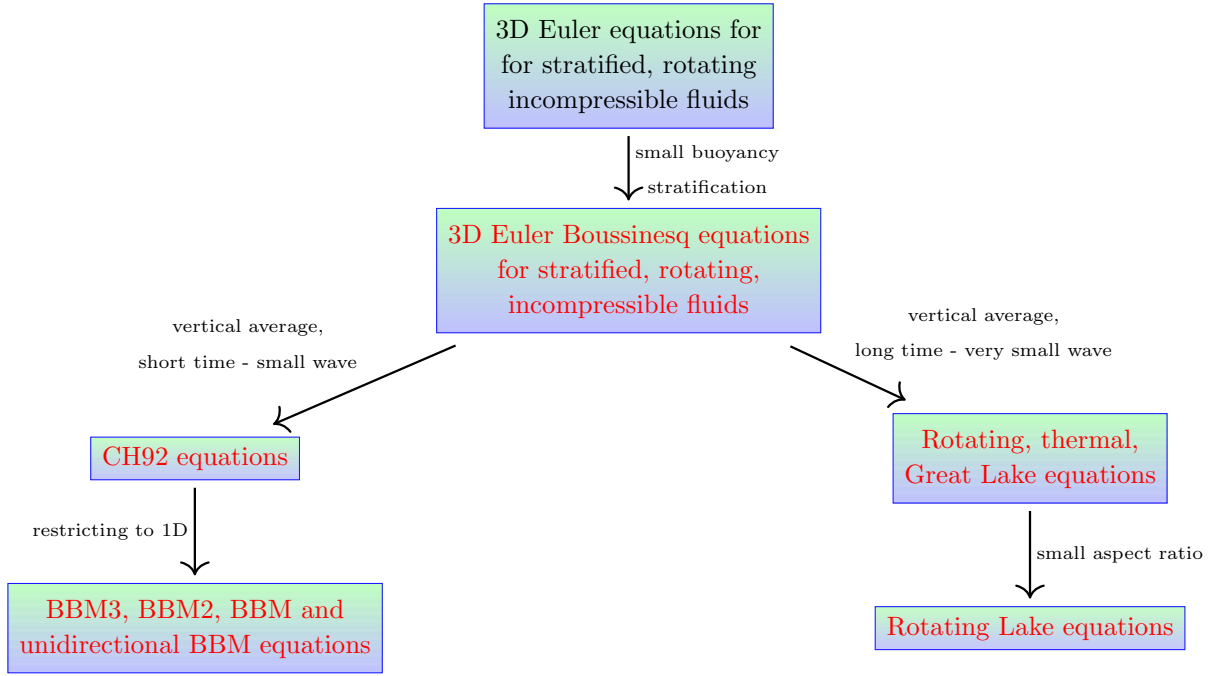

\begin{figure}[H]
\begin{center}
\textbf{Variational Approach}
\end{center}
\bigskip

\centering
\begin{tikzcd}
[row sep = 3em, 
column sep=2em, 
cells = {nodes={top color=green!25, bottom color=blue!25,draw=blue!90}},
arrows = {draw = black, rightarrow, line width = .03cm}]
& & \hyperref[lag:rsE]{\begin{matrix}
\text{3D Euler equations for}\\ \text{for stratified, rotating}\\ \text{incompressible fluids}
\end{matrix}} \arrow[d,"{\begin{matrix} \text{small buoyancy}\\ \text{stratification}\end{matrix}}", shorten <= 1mm, shorten >= 1mm] & &\\
& & \hyperref[lag:EB]{\begin{matrix} \text{3D Euler Boussinesq equations}\\ \text{for stratified, rotating,}\\ \text{incompressible fluids}\end{matrix}} \arrow[dr, "{\begin{matrix} \text{vertical average,}\\ \text{long time - very small wave}\end{matrix}}", shorten <= 4mm, shorten >= 4mm] \arrow[dl, "{\begin{matrix} \text{vertical average,}\\ \text{short time - small wave}\end{matrix}}" swap, shorten <= 4mm, shorten >= 4mm]& &\\
& \hyperref[sec:variationalshorttime]{\begin{matrix} \text{Rotating, thermal}\\ \text{Green--Naghdi equations} \end{matrix}} \arrow[d, swap, shorten <= 1mm, shorten >= 1mm] & & \hyperref[sec:variationallongtime]{\begin{matrix} \text{Rotating, thermal,}\\ \text{Great Lake equations} \end{matrix}} \arrow[d, "\text{small aspect ratio}", shorten <= 1mm, shorten >= 1mm] & \\
& \hyperref[sec:stochCHeqns]{\begin{matrix} \text{CH3, CH2, CH, KdV}\\ \text{equations}\end{matrix}}& & \hyperref[sec:variationallongtime]{\begin{matrix} \text{Rotating Lake equations}\end{matrix}} &
\end{tikzcd}
\caption{Diagram of derivations from the variational point of view. Each blue box refers to the Lagrangian that corresponds to the model referred to in the box. By clicking on the box the reader is taken to the corresponding section. Above each arrow is the approximation that is necessary to transition from one Lagrangian to the next.}
\label{variationalapproach}
\end{figure}
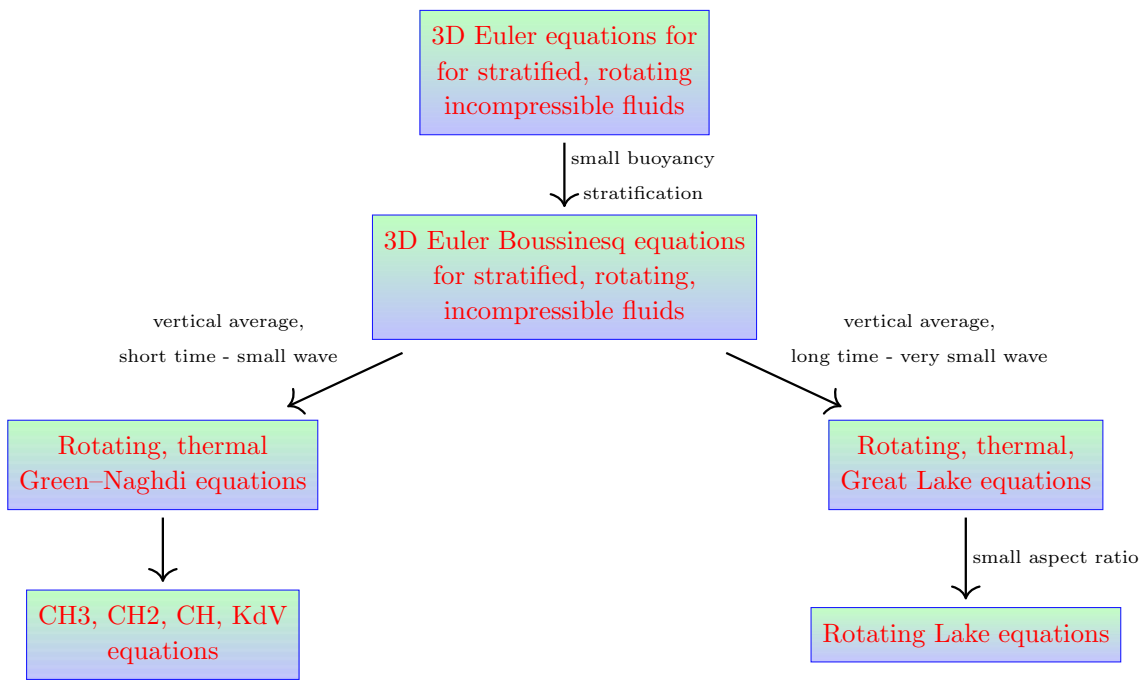

In section \ref{sec:StochVarPrinc} we investigated whether the SALT approach was compatible with the asymptotic expansions. It was shown that an additional assumption on the magnitude of the gradient of the bathymetry was required for the SALT version to be consistent with the deterministic situation. Except for this additional assumption, SALT was verified to be compatible with the methods of asymptotic analysis. From the variational point of view, this was to be expected. Any fluid model which has a corresponding Lagrangian can be made stochastic with the approach of \cite{holm2015variational}. However, boundary conditions need to be made consistent with the derivation of the equations. A simpler, but also important `sanity check' was passed, by confirming that the stochastic Lake and Great Lake equations successfully recover the deterministic Lake and Great Lake equations when the noise terms are absent.
\bigskip

In section \ref{sec:LT-Vsmall}, we showed that the Great Lake equations in \eqref{thm:KCGL} may be derived using a direct approach, by combining vertical averaging of the nondimensional Euler-Boussinesq equations with asymptotic analysis in a long time - very small wave scaling regime. The resulting averaged equations can be closed. One may also derive the same equations by vertically averaging the Lagrangian and applying the Euler-Poincar\'e theorem. In both situations, an averaging principle is required which respects the boundary conditions for the Euler-Boussinesq equations. The road map of these derivations is sketched on the right-hand branches of figures \ref{directapproach} and \ref{variationalapproach}. 
\bigskip

In section \ref{sec:STL-smallwavelimit}, we worked in a short time - small wave scaling regime, following the  left-hand branches of figures \ref{directapproach} and \ref{variationalapproach}. In this scaling regime, the Strouhal number does not equal unity. Instead, the Strouhal number is the inverse of the Froude number, which was taken to be small in this scaling regime. Consequently, the material derivative was no longer balanced  in the asymptotic expansion. Because of this imbalance, the direct asymptotic expansion approach failed to derive the rotating thermal Green-Naghdi equations in this scaling regime. However, the variational approach was able to take an arbitrary Strouhal number into account. In this scaling regime, the variational approach provided a set of equations reminiscent of the Green-Naghdi equations, and which had the geometric structure required to possess a Kelvin circulation theorem. 
Thus, the Strouhal number played a crucial role in determining the differences between the direct approach and the variational approach in the short time - small wave scaling regime. 
In addition, by further approximating the asymptotic expansion of the wave Lagrangian in Hamilton's principle, in Section \ref{sec:stochCHeqns} we derived several stochastic variants of the Camassa-Holm equation and the Korteweg - de Vries equation for one dimensional unidirectional propagation. Finally, in section \ref{sec:NewVarDiffs} we discussed the differences between the Newtonian and variational approaches in this scaling regime by making a detailed comparison of the equations and explaining the implications of the additional terms in the variational approach which were missing in the direct approach. 

\subsection{Outlook and open problems. What to do?}

This paper has integrated several methodologies into a research framework for investigating the various  effects of wave-current interaction in thermal shallow water flows. Several methodologies were required because wave-current interaction involves several elements. Different time scales exist for flow and wave propagation, as indicated by the different regimes of Strouhal number. This means that simultaneous interactions take place among various physical effects with different times scales. For example, we have seen that nonlinear interactions arise among advective transport, dispersive nonlinear wave propagation, stratification and generation of circulation in the interplay of waves, topography and stratification. This is not to even mention the effects of shear on the propagation of waves and the effects of wave perturbations on unstable flow equilibria. 

Because of these various interacting elements, modelling the wave-current interaction process involves many uncertainties. These uncertainties arise from the combination of incomplete sparse observations and the `irreducible imprecision' of numerical simulations arising because of under-resolution  and the wide variety of choice in numerical simulation algorithms. In the hopes of providing a methodology for systematically quantifying these uncertainties, this paper has introduced stochastic advection by Lie transport  (SALT) in the derivation of the various new equations arising in the ramifications of the asymptotic expansions studied here. We believe that the SALT approach could eventually be made useful for stochastic parameterisation and uncertainty quantification of wave-current interaction, for example, in describing the effects of sub-mesoscale unresolved ocean dynamics on the larger, slower, resolvable oceanic flow. Combined with judicious data assimilation approaches based on the earlier work of  \cite{cotter2018modelling, cotter2019numerically, cotter2019particle}, one can hope that in some cases these uncertainties may even be reduced. The progress made here suggests that further pursuit of the SALT approach for stochastic parameterisation may soon be fruitful in the context of wave-current interaction of dispersive nonlinear waves in shallow water with horizontal buoyancy gradients. In the mean time, the present paper has combined asymptotic expansions and vertical averaging with the stochastic variational framework to formulate the SALT approach for the various thermal shallow water equations which descend from Euler's three-dimensional fluid equations under approximation by asymptotic expansions and vertical averaging.

\section{Acknowledgments}
We are grateful for many suggestions for improvements offered in discussions with A. Bethencourt de Leon, C. Cotter, D.O. Crisan, S. Ephrati, J.D. Gibbon, R. Hu, E. Johnson, A. Mashayekhi, W. Pan, D. Papageorgiou, O. Street and S. Takao. EL thanks R. Klein for his inspirational lectures on asymptotic analysis and his careful explanation of the Strouhal number. The work of DDH was partially supported by EPSRC standard grant [grant number EP/N023781/1]. EL was supported by EPSRC grant [grant number EP/L016613/1] and is grateful for the warm hospitality at the Imperial College London EPSRC Centre for Doctoral Training in the Mathematics of Planet Earth during the course of this work.

\bibliographystyle{alpha}
\bibliography{biblio}

\appendix 


\section{Linear dispersion relations for deterministic equilibria of Green--Naghdi equations}\label{app:dispersion}
In the coupled set of stochastic Green--Naghdi equations \eqref{eq:KelThm-VCH92}, there are no time independent solutions. That is, there are no equilibria in the presence of noise.  Hence, in order to investigate the wave behaviour of the solutions of these equations near a steady state, we must switch off the noise, and investigate the equilibria of the deterministic equations. By writing the equations in componentwise form, assuming that the bathymetry $h_0$ is flat and assuming that the Coriolis parameter $f_0$ is constant, linearising around $(\overline{u},\overline{v},\zeta,\overline{b}) = (0,0,0,0)$ yields a set of equations with constant coefficients, given by
\begin{equation}
\begin{aligned}
\frac{1}{{\rm Fr}}\overline{u}_t - \frac{\sigma^2}{3\,{\rm Fr}}h_0^2\overline{u}_{xxt} &= -\frac{\alpha}{{\rm Fr}^2}\zeta_x - \frac{\sigma^2}{2{\rm Fr}^2}h_0\overline{b}_x + \frac{f_0}{{\rm Ro}}\overline{v},
\\
\frac{1}{{\rm Fr}}\overline{v}_t - \frac{\sigma^2}{3\,{\rm Fr}}h_0^2\overline{v}_{yyt} &= -\frac{\alpha}{{\rm Fr}^2}\zeta_y - \frac{\sigma^2}{2{\rm Fr}^2}h_0\overline{b}_y - \frac{f_0}{{\rm Ro}}\overline{u},
\\
\frac{1}{{\rm Fr}}\zeta_t &= -h_0(\overline{u}_x + \overline{v}_y),
\\
\frac{1}{{\rm Fr}}\overline{b}_t &= 0.
\end{aligned}
\label{eq:linearrsGN}
\end{equation}
We can now substitute the travelling wave Ansatz $(\overline{u},\overline{v},\zeta,\overline{b}) = (\overline{u}_0,\overline{v}_0,\zeta_0,\overline{b}_0)e^{{\rm i}(\mathbf{k}\cdot\mathbf{x} - \omega t)}$ into \eqref{eq:linearrsGN}. Standard procedures in linear algebra then imply the dispersion relation as the roots of a quartic polynomial; namely,
\begin{equation}
\begin{aligned}
\omega(\mathbf{k}) &= 0,\\
\omega(\mathbf{k}) &= \pm \sqrt{\frac{\frac{{\rm Fr}^2 f_0^2}{{\rm Ro}^2} + \alpha h_0 |\mathbf{k}|^2 + \frac{2\alpha\sigma^2 h_0^3}{3} k^2l^2}{1 + \frac{\sigma^2 h_0^2}{3}|\mathbf{k}|^2 + \frac{\sigma^4h_0^4}{9} k^2l^2}}.
\end{aligned}
\label{disp-relation}
\end{equation}
In the dispersion relation, $\omega(\mathbf{k})$, the quantity $\mathbf{k} = (k, l)$ is the wave vector in two horizontal dimensions. The zero frequency dispersion relation corresponds to geostrophically balanced motion; uniform in time. When the aspect ratio goes to zero the second expression for the frequency yields dispersion relation for inertio-gravity (or Poincar\'e) waves. At high wave numbers, the wave oscillation  frequency tends to a limiting constant; regularised by nonhydrostatic dispersion. 

Upon further restricting to one-dimensional motion without rotation, the dispersion relation \eqref{disp-relation} takes the form
\begin{equation}
\begin{aligned}
\omega(k) &= 0,\\
\omega(k) &= \pm \frac{\sqrt{\alpha h_0} k}{\sqrt{1 + \frac{\sigma^2 h_0^2}{3}k^2}}\,,
\end{aligned}
\end{equation}
and we can compute the phase velocity $v_p=\omega/k$ and the group velocity $v_g=d\omega/dk$ to be
\begin{equation}\label{SW-disp}
\begin{aligned}
v_p(k) &= \pm \frac{\sqrt{\alpha h_0}}{\sqrt{1 + \frac{\sigma^2 h_0^2}{3}k^2}},\\
v_g(k) &= \pm \frac{\sqrt{\alpha h_0}}{(1+\frac{\sigma^2 h_0^3}{3}k^2)^{3/2}}.
\end{aligned}
\end{equation} 
Equation \eqref{SW-disp} shows the dispersion of shallow water waves, as excitations of longer wavelength travel faster than excitations of shorter wavelength. 


\section{The stochastic thermal rotating shallow water (TRSW) model}\label{app:TRSW}

The  thermal rotating shallow water (TRSW) model describes an upper active layer of fluid motion with horizontally varying buoyancy and an inert lower layer. 
The TRSW model is an extension of the RSW model and a simplification of the various models we have discussed in the text. This TRSW model comprises an upper active layer of fluid motion with horizontally varying buoyancy and an inert lower layer. Since the lower layer is inert, the TRSW model is sometimes called a 1.5 layer model \cite{warneford2013quasi}. For a discussion of a fully multilayer model with nonhydrostatic pressure, see \cite{cotter2010square}.

The TRSW equations are expressed using the following definition for the (nonnegative) buoyancy $b(\bx,t) = (\bar{\rho} - \rho(\bx,t))/\bar{\rho}$, where  $\rho$ is the (time and space dependent) mass density of the active upper layer, $\bar{\rho}$ is the uniform mass density of the inert lower layer. We let $\eta =  \eta(\bx, t)$ be the thickness of the active layer, where $\bx=(x,y)$ is the horizontal vector position, and $t$ is time. The nondimensional deterministic TRSW equations are 
\begin{equation}   
\frac{D}{Dt} \bu+ \frac{1}{{\rm Ro}}f\zh\times\bu   
+ \frac{1}{{\rm Fr}^2}\nabla(b\zeta) - \frac{1}{2\,{\rm Fr}^2}(\zeta-h)\nabla b = 0\, ,   
\qquad 
\frac{\partial \eta}{\partial t} + \nabla\cd (\eta\bu) = 0\, ,   
\qquad
\frac{D b}{Dt} = 0\,,
\label{trsw-eqns}  
\end{equation}
with notation ${\rm Ro}$ for Rossby number and the standard advective time derivative $\frac{D}{Dt}=\partial_t + \bu\cdot\nabla$. The boundary conditions are 
\begin{equation} 
\mathbf{n}\cd \bu=0
\quad\hbox{and}\quad
\mathbf{n}\times\nabla b = 0\,,
\label{trsw-bdy}  
\end{equation}
meaning that fluid velocity $\bu$ is tangential and buoyancy $b$ is constant on the boundary of the domain of flow.\bigskip

Upon introducing the following  stochastic vector field in $\mathbb{R}^2$ for fluid transport  
\begin{equation}
\label{eq:chi-defnA}
{\sf d}\boldsymbol \chi_{t} := \mathbf{u}(\mathbf{x},t)dt + \sum_{i=1}^M \boldsymbol \xi_i(\mathbf{x})\circ dW_t^i\,,
\end{equation}
we can derive the stochastic TRSW equations. The deterministic equations in \eqref{trsw-eqns}  follow as Euler-Poincar\'e equations for the action integral
\begin{equation}
S = \int_0^T\, \ell_{TRSW}(\mathbf{u},\eta,b) \,dt 
=
\int_0^T\!\!\! \int_{CS} \left(\frac{1}{2} |\bu|^2 + \frac{1}{{\rm Ro}}\bu\cdot\bR(\bx) - \frac{1}{2\,{\rm Fr}^2} b(\eta-2h)\right) \eta\,dx\,dy\,dt  \,,
\label{trsw-lag2}   
\end{equation}
where $CS$ denotes some horizontal surface. The stochastic TRSW equations are derived by first evaluating the variational derivatives for the Lagrangian in the action integral \eqref{trsw-lag2} as
\begin{align}
\begin{split}
\frac{1}{\eta}\frac{{\delta} l}{{\delta} \mathbf{u}} &= \bu+\frac{1}{{\rm Ro}}\bR(\bx) =: \mathbf{V}(\bx,t),
\\
\frac{{\delta} l}{{\delta} \eta} &= \frac{1}{2} |\bu|^2 + \frac{1}{{\rm Ro}}\bu\cdot\bR(\bx) -  \frac{1}{{\rm Fr}^2}b (\eta-h),
\\
\frac{{\delta} l}{{\delta} b} &= -  \frac{1}{{\rm Fr}^2}(\eta^2-2\eta h).
\end{split}
\label{vars-trsw}
\end{align}
Next, we apply the stochastic Euler-Poincar\'e theorem \ref{thm:SEP} with the variational derivatives as above and obtain
\begin{equation}
\begin{aligned}
{\sf d}\bu + ({\sf d}\boldsymbol \chi_t\cdot\nabla)\bu  + (\nabla \boldsymbol \xi_i)\cdot\bu\circ dW_t^i &= - \frac{1}{{\rm Fr}^2}\nabla(b \zeta)\,dt + \frac{1}{2\,{\rm Fr}^2}(\zeta-h)\nabla b\,dt - \frac{1}{{\rm Ro}}f\hat{\mathbf{z}}\times {\sf d}\boldsymbol \chi_t - \frac{1}{{\rm Ro}}\nabla(\boldsymbol \xi_i\cdot\mathbf{R})\circ dW_t^i,\\
{\sf d}\eta + \nabla\cdot(\eta {\sf d}\boldsymbol \chi_t) &= 0,\\
{\sf d}b + ({\sf d}\boldsymbol \chi_t\cdot\nabla)b &= 0.
\end{aligned}
\label{eq:STRSW}
\end{equation}
In \eqref{eq:STRSW}, we used $\zeta = \eta - h$ for the free surface elevation.
\begin{remark}
The \textit{stochastic} Euler-Poincar\'e equation may be written in three dimensional vector notation as,
\begin{equation}
{\sf d}
\Big(\frac{1}{\eta} \frac{{\delta} l}{{\delta} \mathbf{u}}\Big)
\,-\, {\sf d}\boldsymbol \chi_{t}\times {\rm \curl}
\Big(\frac{1}{\eta} \frac{{\delta} l}{{\delta} \mathbf{u}}\Big)
\,+\, \nabla\Big({\sf d}\boldsymbol \chi_{t}\cdot\frac{1}{\eta}
\frac{{\delta} l}{{\delta} \mathbf{u}}
\,-\, \frac{{\delta} l}{{\delta} \eta} \,dt \Big)
+ \frac{1}{\eta} \frac{{\delta} l}{{\delta}b} \nabla b \,dt
=0\,.
\label{EP-comp2-trsw-stoch}
\end{equation}
For the Lagrangian in \eqref{trsw-lag2} with variational derivatives given in \eqref{vars-trsw} the stochastic Euler-Poincar\'e equation in \eqref{EP-comp2-trsw-stoch} implies
\begin{equation}
{\sf d}\mathbf{V}
\,-\, {\sf d}\boldsymbol \chi_{t}\times {\rm \curl}\, \mathbf{V}
\,+\, \nabla\Big( \mathbf{V} \cdot \boldsymbol \xi_i(\mathbf{x})\circ dW_t^i
+ \frac{1}{2} |\bu|^2\,dt  \Big)
+ \frac{1}{{\rm Fr}^2}\nabla(b \zeta)\,dt - \frac{1}{2\,{\rm Fr}^2}(\zeta-h)\nabla b\,dt
=0\,.
\label{EP-stochtrsw}
\end{equation}
\end{remark}

\begin{remark}
The stochastic TRSW equations \eqref{eq:STRSW} imply the following Kelvin circulation law
\begin{equation}
{\sf d}\oint_{c({\sf d}\boldsymbol \chi_t)} \frac{1}{\eta} \frac{{\delta} l}{{\delta} \mathbf{u}} \cdot d\bx 
= -\, \oint_{c({\sf d}\boldsymbol \chi_t)} \frac{1}{\eta} \frac{{\delta} l}{{\delta} b} \nabla b \cdot d\bx
\,,
\label{EP-Kelvin-Noether-trsw}   
\end{equation}
where $c({\sf d}\boldsymbol \chi_t)$ is a closed loop moving with stochastic horizontal fluid velocity ${\sf d}\boldsymbol \chi_t(\bx,t)$ in two dimensions. 
Evaluating for the variational derivatives for TRSW in \eqref{vars-trsw} yields
\begin{equation}
{\sf d}\oint_{c({\sf d}\boldsymbol \chi_t)}\mathbf{V} \cdot d\bx 
= \frac{1}{2\,{\rm Fr}^2}  \oint_{c({\sf d}\boldsymbol \chi_t)} (\zeta-h) \nabla b \cdot d\bx
= \frac{1}{2\,{\rm Fr}^2}  \int\!\!\int_{\partial S = c({\sf d}\boldsymbol{\chi_t})}
\nabla(\zeta-h) \times\nabla b \,d\mathbf{S}\,dt.
\,,
\label{EP-KN-trsw}   
\end{equation}
One sees in equation \eqref{EP-KN-trsw} that misalignment of the horizontal gradients of free surface elevation $\zeta$, bathymetry $h$ and buoyancy $\gamma^2$ will generate circulation, cf. the corresponding Kelvin circulation theorems in equations \eqref{GL-KelThm} and  \eqref{eq:KelThm-rsGN}. 
\end{remark}
\bigskip

\begin{remark}
The evolution of potential vorticity on fluid parcels for the TRSW equations in \eqref{eq:STRSW} is given by
\begin{equation}
{\sf d}q + ({\sf d}\boldsymbol \chi_t\cdot\nabla)q = \frac{1}{2\,{\rm Fr}^2\,\eta}J(\eta,b),
\end{equation}
where the potential vorticity is defined by
\begin{equation}
q:= \frac{\varpi}{\eta}, \qquad \hbox{and} \qquad \varpi := \hat{\mathbf{z}}\cdot\nabla\times\mathbf{V},
\end{equation}
and
\begin{equation}
J(\eta,b) = \hat{\mathbf{z}}\cdot\nabla\eta\times\nabla b = -\nabla\cdot(\eta\hat{\mathbf{z}}\times\nabla b)
\end{equation}
is the Jacobian of the depth $\eta$.
\end{remark}
\bigskip

\begin{remark}
The stochastic TRSW equations \eqref{eq:STRSW} have an infinite number of conserved integral quantities
\begin{equation}
C_{f,g} = \int_{CS} \big( f(b) +  q g(b)\big) \eta\,dxdy,
\label{eq:casimirstrsw}
\end{equation}
for the boundary conditions given in \eqref{trsw-bdy} and any differentiable functions $f$ and $g$.
\end{remark}
\bigskip

\begin{remark}
The Legendre transform which determines the Hamiltonian ${\sf d}h$ for the stochastic TRSW equations is defined as\,\footnote{Notice that the Hamiltonian ${\sf d}h$ in \eqref{LegXform-VCH092} is a semimartingale. Recall the definition \ref{def:smartingale}.}
\begin{equation}
{\sf d}h  (\mu ,\eta,b)
:= 
\big\langle \mu , {\sf d}\boldsymbol \chi_t \big\rangle - \ell_{TRSW} (\mathbf{u},\eta,b)dt
\,,
\label{LegXform-VCH092}
\end{equation}
in which the angle brackets in the definition of the Legendre transform denote the $L^2$ pairing over the horizontal cross-section $CS$. The Hamiltonian form of the stochastic TRSW equations is given by
\begin{align}
\frac{1}{{\rm Fr}}{\sf d} F =\Big\{ F, {\sf d}h\Big\}
=
- \int_{\Omega}
\begin{bmatrix}
{\delta F}/{\delta \mu_j} \\  {\delta F}/{ \delta \eta} \\ {\delta F}/{ \delta b} 
\end{bmatrix}^T
\begin{bmatrix}
\mu_j\partial_i + \partial_j \mu_i & \eta \partial_i & -\,b_{,i}
\\ 
\partial_j \eta & 0 & 0 \\
b_{,j} & 0 & 0
\end{bmatrix}
\begin{bmatrix}
{\delta ({\sf d}h }) / {\delta \mu_j} = {\sf d} \chi_t^{\,j}\
\\
{\delta ({\sf d}h })/{ \delta \eta} = - \, {\delta \ell_{TRSW} }/{\delta \eta}
\\  
{\delta ({\sf d}h })/{ \delta b} = - \, {\delta \ell_{TRSW} }/{\delta b}
\end{bmatrix}
\,dx\,dy\,.
\label{Ham-matrix-EB-stoch}
\end{align}
The conserved integral quantities $C_{f,g}$ defined in \eqref{eq:casimirstrsw} are Casimirs of the Lie--Poisson bracket in \eqref{Ham-matrix-EB-stoch} which persist when the Hamiltonian is made stochastic.  This means that these equations describe stochastic coadjoint motion in function space on level sets of the Casimir functionals $C_{f,g}$. Thus, the SALT introduction of stochasticity into the TRSW equations preserves their Lie--Poisson bracket and thereby preserves their geometric interpretation. 
\end{remark}

\end{document}